\documentclass[12pt]{article}

\usepackage{xr-hyper}
\makeatletter

\usepackage{graphicx} 
\usepackage[hidelinks]{hyperref} 
\usepackage{natbib} 
\usepackage{microtype} 
\usepackage{subcaption} 
\usepackage{enumitem} 
\usepackage{pifont} 
\usepackage[normalem]{ulem} 
\usepackage{comment}
\usepackage{dingbat}
\usepackage{siunitx}

\usepackage{amssymb,amsmath,amsthm,amsfonts}
\usepackage{mathtools}
\usepackage{bbm}
\usepackage{color,xcolor}

\newcommand{\one}{\mathbbm{1}}
\newcommand{\E}{\operatorname{\mathsf{E}}}
\newcommand{\Q}{\operatorname{\mathsf{P}}}
\newcommand{\R}{\mathbb{R}}

\newcommand{\Aa}{\mathcal{A}}
\newcommand{\Bb}{\mathcal{B}}
\newcommand{\Cc}{\mathcal{C}}
\newcommand{\Ii}{\mathcal{I}}
\newcommand{\Jj}{\mathcal{J}}

\newcommand{\Pp}{\mathcal{P}}
\newcommand{\Ss}{\mathcal{S}}

\newcommand{\GPD}{\operatorname{GPD}} 
\newcommand{\Exp}{\operatorname{Exp}} 
\newcommand{\expec}{\operatorname{\mathsf{E}}} 
\newcommand{\cov}{\operatorname{\mathsf{cov}}}
\newcommand{\var}{\operatorname{\mathsf{var}}}
\newcommand{\point}{\,\cdot\,} 
\newcommand{\1}{\mathbf{1}} 
\newcommand{\eps}{\varepsilon}
\newcommand{\Unif}{\operatorname{Unif}} 

\newcommand{\Fbar}{\overline{F}}
\newcommand{\GB}{G_{\Bb}}
\newcommand{\GBbar}{\overline{G}_{\Bb}}
\newcommand{\zit}{z_{i}^{(t)}}
\newcommand{\Zit}{Z_{F_{i}}^{(t)}}

\renewcommand{\subset}{\subseteq}
\newcommand{\ZFt}{Z_{F}^{(t)}} 
\newcommand{\ZFtn}{Z_F^{(t_n)}}

\def\jon#1{{\textcolor{red}{\bf $[$JK: #1$]$}}}
\def\sam#1{{\textcolor{blue}{\bf $[$SA: #1$]$}}}

\definecolor{darkteal}{rgb}{0,0.35,0.35}
\newcommand{\js}[1]{\textcolor{darkteal}{\sffamily\small\upshape [JS: #1]}}

\newtheorem{propo}{Proposition}
\newtheorem{lemma}[propo]{Lemma}
\newtheorem{theorem}[propo]{Theorem}
\newtheorem{corollary}[propo]{Corollary}
\theoremstyle{definition}
\newtheorem{defin}[propo]{Definition}
\newtheorem{remark}[propo]{Remark}
\newtheorem{example}[propo]{Example}

\usepackage[font=small]{caption}

\newcommand{\blind}{1}
\begin{document}

\def\spacingset#1{\renewcommand{\baselinestretch}%
{#1}\small\normalsize} \spacingset{1}

\newcommand*{\addFileDependency}[1]{
\typeout{(#1)}
%
%
\@addtofilelist{#1}
%
\IfFileExists{#1}{}{\typeout{No file #1.}}
}\makeatother

\newcommand*{\myexternaldocument}[1]{%
\externaldocument{#1}%
\addFileDependency{#1.tex}%
\addFileDependency{#1.aux}%
}

\myexternaldocument{supplement2}


\if1\blind
{
  \title{\bf Tail calibration of probabilistic forecasts}
  \author{Sam Allen\thanks{Corresponding author: \url{sam.allen@stat.math.ethz.ch}} \\
    Seminar for Statistics, ETH Zurich\\
    and \\
    Jonathan Koh\\
    IMSV, University of Bern\\
    and \\
    Johan Segers\\
    Dept. of Mathematics, KU Leuven, and LIDAM/ISBA, UCLouvain\\
    and \\
    Johanna Ziegel\\
    Seminar for Statistics, ETH Zurich}
  \maketitle
} \fi

\if0\blind
{
  \bigskip
  \bigskip
  \bigskip
  \begin{center}
    {\LARGE\bf Tail calibration of probabilistic forecasts}
\end{center}
  \medskip
} \fi

\bigskip
\begin{abstract}
Probabilistic forecasts comprehensively describe the uncertainty in the unknown future outcome, making them essential for decision making and risk management. While several methods have been introduced to evaluate probabilistic forecasts, existing evaluation techniques are ill-suited to the evaluation of tail properties of such forecasts. However, these tail properties are often of particular interest to forecast users due to the severe impacts caused by extreme outcomes. In this work, we introduce a general notion of tail calibration for probabilistic forecasts, which allows forecasters to assess the reliability of their predictions for extreme outcomes. We study the relationships between tail calibration and standard notions of forecast calibration, and discuss connections to peaks-over-threshold models in extreme value theory. Diagnostic tools are introduced and applied in a case study on European precipitation forecasts.
\end{abstract}

\noindent%
\vfill

\newpage

\spacingset{1.9} 

\section{Introduction}\label{sec:intro}

Forecasts for future events should be reliable, or calibrated. For example, if an event is predicted to occur with a certain probability, this forecast will only be useful for decision making and risk assessment if the event does indeed transpire with the predicted probability. Since extreme events often have the largest impacts on forecast users, calibrated forecasts for these outcomes are particularly valuable. It is therefore essential that methods are available to assess the performance of probabilistic forecasts for extreme events. However, the evaluation of forecasts with regards to extreme events is difficult and prone to methodological pitfalls \citep{LerchEtAl2017,BellierEtAl2017}.

Proper scoring rules quantify the accuracy of probabilistic forecasts, thereby allowing forecasters to be ranked and compared objectively \citep[see e.g.][]{GneitingRaftery2007}. However, despite their widespread use, no proper scoring rule can distinguish the true distribution of the observations from a forecast distribution with incorrect tail behavior \citep{BrehmerStrokorb2019}. Weighted scoring rules, which are commonly employed to emphasize particular outcomes during forecast evaluation \citep{DiksEtAl2011,GneitingRanjan2011,HolzmannKlar2017,LerchEtAl2017,AllenEtAl2023b,olafsdottir.2024.wscrps}, also suffer from these limitations. We extend the negative result of \citet{BrehmerStrokorb2019} by demonstrating that it cannot be circumvented by employing a class of scoring rules to assess forecast performance; see Appendix \ref{sec:scores} in the supplement. Alternative tools are therefore required to evaluate forecasts for extreme events. 

In this paper, we focus on probabilistic forecasts for real-valued outcomes. A probabilistic forecast takes the form of a probability distribution $F$ over all possible values of the outcome $Y \in \R$, and is typically stated as a cumulative distribution function (cdf). While scoring rules assess relative forecast performance, it is also important to determine whether probabilistic forecasts are calibrated.  Forecast calibration is often assessed in practice using histograms or \emph{pp}-plots of probability integral transform (PIT) values \citep{DieboldEtAl1998,GneitingEtAl2007}. Such diagnostic tools additionally allow practitioners to ascertain what systematic biases manifest in the forecasts, thereby helping to improve future predictions.

While several definitions of calibration have been proposed \cite[see][for a recent review]{GneitingResin2022}, all focus on overall forecast performance, and do not allow practitioners to evaluate forecasts made for extreme outcomes. In this paper, we introduce a general notion of forecast \emph{tail calibration} that quantifies the reliability of predictions when forecasting extreme events. We derive theoretical implications of tail calibration, study the connection with peaks-over-threshold models in extreme value theory, and provide examples where standard calibration does and does not imply tail calibration. In general, tail calibration is not implied by classical definitions of calibration, highlighting that it provides additional information that is not available from existing forecast evaluation methods.

In some fields, \emph{forecast calibration} refers to re-calibration methods that adapt forecasts so that they are calibrated. This paper does not concern re-calibration, though it facilitates the introduction of \emph{tail re-calibration} methods. Tail re-calibration methods currently do not exist, largely due to a lack of tools to evaluate the resulting forecasts. Our definitions of tail calibration remedy this by allowing forecasters to assess whether their forecasts for extreme outcomes are reliable, and if not, to understand how the forecasts are miscalibrated, for example, if the tails of the predictive distribution are too light or too heavy.

In Section~\ref{sec:calib}, we revisit notions of calibration for probabilistic forecasts, and introduce concepts that are suitable when evaluating forecasts for extreme events. Connections between tail calibration and well-known results in extreme value theory are explored in Section~\ref{sec:mda}. Section~\ref{sec:diagnostics} presents several examples and introduces diagnostic tools to assess tail calibration in practice, and these diagnostic tools are then applied to a canonical weather forecasting data set in Section~\ref{sec:casestudy}. Section~\ref{sec:conclusion} summarizes our findings. The appendices in the supplement contain sections about the following topics: additional proofs and calculations; a result about the evaluation of tail behavior using multiple scoring rules; a definition of, and results on, marginal tail calibration; additional simulation results evaluating well-known forecasts in the literature with respect to their probabilistic tail calibration; and explanations on the construction of confidence intervals in our diagnostic tools. An R package and R code to reproduce the results herein is provided as supplementary material.

\section{Tail calibration}\label{sec:calib}

Let $( \Omega, \Aa, \Q )$ be a probability space supporting a random pair
$
(F, Y)
$, where $F$ is a random probabilistic forecast for the outcome variable $Y$. We restrict attention to real-valued outcomes and treat $F$ as a random cumulative distribution function (cdf). We assume that $F$ is almost surely continuous, as this case is most relevant with regards to extremes and leads to a simpler presentation. However, generalizations are possible; see Remark \ref{rem:F_jumps}.

\subsection{Classical notions of calibration}

\cite{Tsyplakov2011} refers to a probabilistic forecast $F$ as \emph{auto-calibrated} if, for all $y \in \R$, $F(y) = \Q ( Y \leq y \mid F)$ almost surely. That is, given that the distribution $F$ has been issued as the forecast, the observation will arise according to this distribution. This provides an intuitive requirement that a forecast must satisfy in order to be considered reliable. More generally, for some $\sigma$-algebra $\Bb\subseteq \Aa$, a forecast is called \emph{(probabilistically) $\Bb$-calibrated} if 
\begin{equation}\label{eq:auto-cal:3}
    \Q(Z_F \le u \mid \Bb ) = u \quad \text{almost surely, for all $u \in [0,1]$},
\end{equation}
where $Z_F = F(Y)$ is the probability integral transform (PIT) of $Y$ with respect to $F$. In other words, $Z_F$ is required to have a standard uniform distribution, $Z_F\sim \Unif(0,1)$, and to be independent of $\Bb$. \citet[Proposition~3.4]{Modeste2023} shows that if $\sigma(F)\subseteq \Bb$, then $\Bb$-calibration of $F$ is equivalent to  $F(y) = \Q ( Y \leq y \mid \Bb)$ almost surely, for all $y \in \R$. The case $\Bb = \sigma(F)$ yields auto-calibration. \citet[Definition~3]{Tsyplakov2014} refers to \eqref{eq:auto-cal:3} as conditional probabilistic calibration, and for $\sigma(F)\subseteq \Bb$, $\Bb$-calibration is equivalent to the forecast $F$ being ideal with respect to $\Bb$ \citep{GneitingRanjan2013}.

\begin{remark}[Discontinuous forecasts]
\label{rem:F_jumps}
We assume that the forecast cdf $F$ is continuous almost surely, but the notions of calibration that rely on $Z_F$ can also be defined for possibly discontinuous $F$. In this case, let $V$ be a standard uniform random variable that is independent of $(F,Y)$ and define $Z_F = F(Y-) + V[ F(Y) - F(Y-) ] $, where $F(y-)$ is the left-hand limit of $F$ at $y$ \citep[Definition 2.5]{GneitingRanjan2013}. The result of \citet[Proposition 3.4]{Modeste2023} continues to hold with this definition of $Z_F$.
\end{remark}

Auto-calibration is a strong and intuitive notion of calibration, and several statistical tests for auto-calibration have been proposed; see for example \citet{StrahlZiegel2017}, \citet{Modeste2023}, and related literature on testing forecast encompassing \citep{HarveyLeybourneETAL1998}. However, it is typically difficult to assess auto-calibration in practice using diagnostic tools. Instead, it is common to identify relevant deficiencies of forecasts by employing powerful diagnostics that assess necessary criteria for auto-calibration.

Following \cite{GneitingEtAl2007}, the forecast distribution $F$ is said to be \emph{marginally calibrated} if $\E [ F(y) ] = \Q( Y \leq y )$ for all $y \in \R$, and \emph{probabilistically calibrated} if $\Q ( Z_{F} \leq u ) = u$ for all $u \in [0,1]$, i.e. if $Z_F \sim \Unif(0,1)$. Marginal calibration implies that the average probability assigned to $Y$ exceeding a threshold is equal to the true unconditional probability, for any threshold. Probabilistic calibration corresponds to \eqref{eq:auto-cal:3} when $\Bb$ is the trivial $\sigma$-algebra, $\{\varnothing, \Omega\}$, and implies that all central prediction intervals of $F$ have the correct coverage. 

Marginal and probabilistic calibration are unconditional notions of calibration, whereas auto-calibration conditions on the forecast distribution \citep{GneitingResin2022}. Intermediate notions of forecast calibration exist, including quantile and threshold calibration, which condition on functionals of the forecast distribution; see Appendix~\ref{app:marginal} in the supplement. \citet{GneitingEtAl2007} show that marginal and probabilistic calibration are neither necessary nor sufficient conditions for each other. However, if a forecast is auto-calibrated, then it is both marginally and probabilistically calibrated, whereas the converse is not true. More generally, for $\Bb \subset \sigma(F)$, $\Bb$-calibration is generally weaker than auto-calibration, whereas it is generally stronger if $\sigma(F) \subset \Bb$.  

\begin{example}[Unfocused forecaster]\label{ex:1}
   To see that probabilistic calibration is strictly weaker than auto-calibration when $F$ is not deterministic, consider the \emph{unfocused} forecaster of \citet[Example~3]{GneitingEtAl2007}. Suppose $Y\sim N(\mu, 1)$ follows a normal distribution with mean $\mu$ and variance 1 conditionally on $\mu$, where $\mu \sim N(0, 1)$, and let $F = (1/2) N(\mu, 1) + (1/2)N(\mu + \tau, 1)$ with $\tau$ equal to $1$ and $-1$ with probability $0.5$, independently of $\mu$. Figure~\ref{fig:unfocused} shows an example of this forecast for a fixed $\mu$. This forecaster is either positively or negatively biased, with equal probability. As a result, this forecaster is not auto-calibrated (and also not marginally calibrated). However, perhaps unintuitively, the biases cancel out such that the forecast is probabilistically calibrated.
\end{example}

\begin{figure}
    \centering
    \includegraphics[width=3.3in]{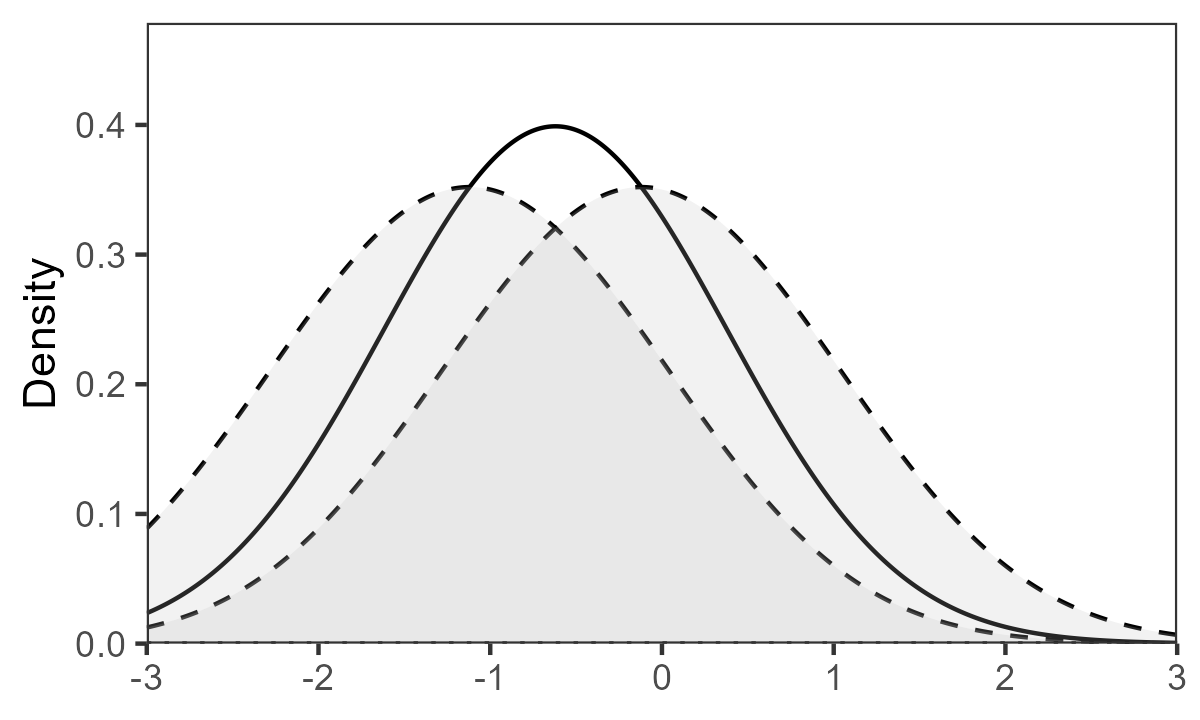}
    \caption{The unfocused forecaster in Example~\ref{ex:1} issues each of the dashed densities with probability one half, while the true outcome distribution is the solid density.}
    \label{fig:unfocused}
\end{figure}

\subsection{Notions of tail calibration}
\label{sec:tailcalib}

To assess the calibration of forecasts for extreme events, we adopt the standard approach of extreme value theory and decompose $\Q(Y> x+t ) = \Q(Y > t) \Q(Y > x+t \mid Y > t)$, 
for $x \ge 0$ and large $t$, i.e., $\Q(Y > t)$ positive but close to zero. Then, we separately consider calibration with respect to both factors on the right-hand side. Under auto-calibration, 
\[
\Q(Y > x+t \mid F) = \Q(Y>t \mid F)\; \frac{\Q(Y>x+t \mid F)}{\Q(Y > t \mid F)} = (1-F(t)) (1-F_t(x))
\]
with the \emph{conditional forecast excess distribution} $F_t$ defined as
\begin{equation}
\label{eq:Ftx}
F_{t}(x) = \frac{F(x + t) - F(t)}{1-F(t)}, \quad x \ge 0,
\end{equation}
if $F(t) < 1$, and $F_t(x) = 1$, otherwise.
To evaluate the calibration of the forecast's tails, we propose checking two things for large $t$: first, whether the forecast survival function is close to the conditional survival function of the outcome variable, that is, whether $1-F(t)$ is close to $\Q(Y>t \mid F)$; and second, whether the conditional forecast excess distribution is close to the conditional excess distribution of $Y$, that is, whether $1-F_t(x)$ is close to $\Q(Y>x+t \mid F)/\Q(Y > t \mid F)$ for all $x \ge 0$. We will call these two aspects the \emph{occurrence} and the \emph{severity} of excesses, respectively.

In order to come up with non-trivial notions of calibration for large $t$, it is important to keep the following remark in mind.

\begin{remark}[Scaling issue]
\label{rem:limit_F_t}
A direct comparison of the true and forecast conditional excess distributions $\Q(Y - t \le x \mid Y > t)$ and $F_t(x)$, for $x \ge 0$ and large $t$, would potentially yield uninformative results because of a scaling issue. To see this, suppose $Y$ and $F$ both have unbounded support. For fixed $x > 0$, if both $Y$ and $F$ are heavy-tailed, then both probabilities tend to $0$ as $t \to \infty$, whereas if both have a tail that is lighter than exponential, then both probabilities tend to $1$ almost surely. In the two cases, the difference between $\Q(Y - t \le x \mid Y > t)$ and $F_t(x)$ tends to zero as $t \to \infty$, even though the tails of $Y$ and $F$ could be very different.
\end{remark}

To avoid this scaling issue , we define the \emph{excess probability integral transform (PIT)} as
\begin{equation}\label{eq:exPIT}
\ZFt = F_{t}(Y - t),
\end{equation}
and propose the following definitions of tail calibration.

\begin{defin}\label{def:TAC_new}
Let $x_Y = \sup \{ x \in \R : \Q(Y \le x) < 1 \}$ be the upper endpoint of $Y$ and let $\Bb \subseteq \Aa$ be a $\sigma$-algebra. Assume that the upper endpoint of the conditional distribution of $Y$ given $\Bb$ is equal to $x_Y$, i.e., for all $t < x_Y$, we have, almost surely, both $\Q(Y > t \mid \Bb) > 0$ and $\Q(Y > x_Y \mid \Bb) = 0$. 
\begin{enumerate}
    \item[(a)] The forecast $F$ is \emph{tail $\Bb$-calibrated} for $Y$ if $\E[1 - F(t) \mid \Bb] > 0$ almost surely for all $t < x_Y$ and if, for all $u \in [0,1]$, we have
    \begin{equation}
    	\label{eq:tailBCt}
    	\frac{\Q\bigl(\ZFt \le u, Y > t \mid \Bb\bigr)}{\E[1 - F(t) \mid \Bb]} \to u, 
    	\qquad \text{almost surely as $t \to x_Y$.}
    \end{equation}

    \item[(b)] The forecast $F$ is \emph{tail auto-calibrated} for $Y$ if it is tail $\sigma(F)$-calibrated, that is, $\1 - F(t) > 0$ almost surely for all $t < x_Y$, and
    \begin{equation*}
    	\frac{\Q\bigl(\ZFt \le u, Y > t \mid F\bigr)}{1 - F(t)} \to u, 
    	\qquad \text{almost surely as $t \to x_Y$.}
    \end{equation*}

    \item[(c)] The forecast $F$ is \emph{probabilistically tail calibrated} for $Y$ if it is tail $\{ \varnothing, \Omega \}$-calibrated, that is, $\E[1 - F(t)] > 0$ for all $t < x_Y$, and
    \begin{equation*}
    	\frac{\Q\bigl(\ZFt \le u, Y > t\bigr)}{\E[1 - F(t)]} \to u, 
    	\qquad \text{as $t \to x_Y$.}
    \end{equation*}
\end{enumerate}

\end{defin}

\begin{remark}[Marginal tail calibration]
	\label{rem:marginal}
    Marginal and probabilistic calibration are two unconditional notions of forecast calibration. Due to its greater relevance in practice, we opt for probabilistic calibration as a starting point for defining notions of tail calibration. In Appendix~\ref{app:marginal} in the supplement, we explore an option based on marginal calibration. 
\end{remark}

The definition of $\Bb$-calibration at \eqref{eq:auto-cal:3} motivates the definition of tail calibration with respect to an arbitrary information set $\Bb$. The $\sigma$-algebra $\Bb$ encodes the information with respect to which we would like to assess the forecaster's calibration; tail auto-calibration and probabilistic tail calibration emerge as special cases. Probabilistic tail calibration is similar in essence to the standard notion of probabilistic calibration, except that it relies on the excess PIT associated with the conditional forecast excess distribution and outcomes above a large threshold. \cite{AllenEtAl2023b} and \cite{MitchellWeale2023} consider a similar framework to assess forecast calibration when interest is in exceedances of a fixed threshold. By taking the limit as~$t$ tends to the upper endpoint, we provide a theoretically attractive framework with which to evaluate forecasts made for extreme events, addressing the lack of existing methods to achieve this.

The left-hand side of Eq.~\eqref{eq:tailBCt} can be written as a \emph{combined ratio} via
\[
	\frac{\Q\bigl(\ZFt \le u, Y > t \mid \Bb\bigr)}{\E[1 - F(t) \mid \Bb]}
	=
	\frac{\Q\bigl(Y > t \mid \Bb\bigr)}{\E[1 - F(t) \mid \Bb]}
	\cdot
	\frac{\Q\bigl(\ZFt \le u, Y > t \mid \Bb\bigr)}{\Q\bigl(Y > t \mid \Bb\bigr)}.
\]
The single condition at \eqref{eq:tailBCt} is equivalent to two conditions together: For all $u \in [0,1]$,
\begin{align}
    \frac{\Q(Y > t \mid \Bb)}{\E[1 - F(t) \mid \Bb]}
    &\to 1, 
    \quad \text{almost surely as $t \to x_Y$;} 
    \label{eq:tailBC:occ}\\
    \frac{\Q\bigl(\ZFt \le u, Y > t \mid \Bb\bigr)}{\Q(Y > t \mid \Bb)} 
    &\to u, 
    \quad \text{almost surely as $t \to x_Y$.} 
    \label{eq:tailBC:sev}
\end{align}

That \eqref{eq:tailBC:occ} and \eqref{eq:tailBC:sev} imply \eqref{eq:tailBCt} is easy to see. For the converse, set $u = 1$ in \eqref{eq:tailBCt} to conclude \eqref{eq:tailBC:occ}, and then combine \eqref{eq:tailBCt} with \eqref{eq:tailBC:occ} to obtain \eqref{eq:tailBC:sev}. Condition~\eqref{eq:tailBC:occ} concerns the \emph{occurrence} of excesses over a high threshold: Conditionally on $\Bb$, we would like the exceedance probabilities over~$t$ to be forecast correctly, at least asymptotically. In other words, we would like the tails of $Y$ and $F$ to be equivalent, given $\Bb$. Condition~\eqref{eq:tailBC:sev} concerns the \emph{severity} of excesses over~$t$ once they occur, and arguably needs more justification. 
Condition~\eqref{eq:tailBC:sev} is related to the characterization of $\Bb$-calibration in \eqref{eq:auto-cal:3}. Asymptotically, under tail $\Bb$-calibration, the conditional probability $\Q\bigl(\ZFt \le u, Y > t \mid \Bb\bigr)$ factorizes as $\Q(Y>t\mid \Bb) \cdot u$, and hence the excess PIT is (approximately) uniformly distributed and independent of $\Bb$; when $\Bb = \{\varnothing, \Omega\}$, condition~\eqref{eq:tailBC:sev} can be simplified to $\Q\bigl(\ZFt \le u \mid Y > t \bigr)\to u$ almost surely as $t \to x_Y$. The single condition at \eqref{eq:tailBCt} in Definition~\ref{def:TAC_new} has the following persuasive interpretation: Given the information $\Bb$, the conditional probability of $\{Y > t\} \cap \{\ZFt \le u\}$ factorizes asymptotically as $(1-F(t)) \cdot u$.

When assessing the quality of forecasts in the tail, it is not sufficient to focus only on condition~\eqref{eq:tailBC:sev}, the \emph{severity ratio}, and forget about condition~\eqref{eq:tailBC:occ}, the \emph{occurrence ratio}. Suppose $F$ is a deterministic forecast such that there exists $t_0 < x_Y$ and $c > 0$ such that $1 - F(t) = c \Q(Y > t)$ for all $t \ge t_0$. Then $F$ and $Y$ have the same conditional excess distributions for thresholds above $t_0$ and thus $F$ satisfies condition \eqref{eq:tailBC:sev}, even though the tail probability ratio $\Q(Y > t) / (1-F(t))$ can be arbitrarily small or large. Conversely, only considering the occurrence ratio \eqref{eq:tailBC:occ} without the severity ratio \eqref{eq:tailBC:sev} is also insufficient, as the following example shows.

\begin{example}[Misinformed forecaster]
\label{ex:mis}
For $\gamma > 0$, let $\Delta_1$ and $\Delta_2$ be two independent $\Gamma(1/\gamma,1/\gamma)$ random variables, with $\Gamma(a, b)$ the gamma distribution with shape $a > 0$ and scale $b > 0$. Conditionally on $(\Delta_1,\Delta_2)$, let the outcome $Y$ be an $\Exp(\Delta_1)$ random variable, $\Q(Y \le x \mid \Delta_1, \Delta_2) = 1-e^{-\Delta_1 x}$ for $x \ge 0$, and consider the \emph{misinformed forecaster} $F = \Exp(\Delta_2)$, i.e., $F(x) = 1 - e^{-\Delta_2 x}$ for $x \ge 0$. For $t \ge 0$, we have
\[
	\Q(Y \le t) = \E[F(t)] = \E[1 - e^{-\Delta_j t}]
	= 1 - (1+\gamma t)^{-1/\gamma} = \GPD_{1,\gamma}(t),
\]
the generalized Pareto distribution with scale parameter $1$ and shape parameter $\gamma$ (see also Section~\ref{sec:mda} below). For $\Bb = \{\varnothing, \Omega\}$, condition~\eqref{eq:tailBC:occ} is trivially satisfied: by construction, the marginal occurrence ratio $\Q(Y > t)$ is forecast correctly on average by $\E[1-F(t)]$. However, condition~\eqref{eq:tailBC:sev} is violated, since the excess PIT is calculated with respect to the conditional excess distribution $F_t = \Exp(\Delta_2)$, while the conditional distribution of $Y-t$ is $\Exp(\Delta_1)$. Intuitively, $Y$ is large when the rate parameter $\Delta_1$ is close to zero, but then the scale $1/\Delta_1$ of the excess $Y-t$ is (under-)estimated by $1/\Delta_2$ instead. As a consequence, the limit in \eqref{eq:tailBC:sev} is $0$ instead of $u \in (0, 1)$. Details are provided in Appendix~\ref{app:proofs} in the supplement.
\end{example}

For non-random probabilistic forecasts $F$, the situation simplifies considerably, as it is sufficient to evaluate the occurrence ratio in order to assess tail $\Bb$-calibration.

\begin{propo}
\label{propo:tailequiv}
If the continuous probabilistic forecast $F$ is non-random, then $F$ is tail $\Bb$-calibrated if and only if 
\begin{equation}
    \label{eq:tailBC3}
    \frac{\Q(Y > t \mid \Bb)}{1-F(t)} \to 1, \qquad \text{almost surely as $t \to x_Y$.}
\end{equation}
\end{propo}

The proof of Proposition~\ref{propo:tailequiv} is given in Appendix~\ref{app:proofs} in the supplement. It makes use of the following lemma on the upper endpoint of the outcome $Y$, which is also useful for further results. Its proof is also given in Appendix~\ref{app:proofs}.

\begin{lemma}
    \label{lem:endpoint}
    On the probability space $(\Omega, \Aa, \Q)$, let $F$ be a random probabilistic forecast, let $Y$ be a real-valued outcome variable with upper endpoint $x_Y \in \R \cup \{+\infty\}$, and let $\Bb \subseteq \Aa$ be a $\sigma$-algebra. Assume $F$ is continuous a.s.\ and assume the upper endpoint of the conditional distribution of $Y$ given $\Bb$ is $x_Y$. If \eqref{eq:tailBC:sev} holds, and in particular if $F$ is tail $\Bb$-calibrated for $Y$, then $\Q(Y = x_Y) = 0$.
\end{lemma}

\begin{remark}[Endpoint]
    \label{rem:endpoint}
    In Definition~\ref{def:TAC_new}, it is assumed that the upper endpoint of the conditional distribution of $Y$ given $\Bb$ is not random and is equal to $x_Y$. In some scenarios, this assumption could be restrictive. For example, if the information in $\Bb$ determines whether an insurance-policy holder has a positive probability of filing a large claim $Y$ or not, then the upper endpoint of the conditional distribution is random. One possible extension of the definition of tail $\Bb$-calibration would be to assume that there exists a $\Bb$-measurable random variable $x_{\Bb}$ such that the upper endpoint of $Y$ given $\Bb$ is $x_{\Bb}$. Details for this added generality would have to be worked out.
\end{remark}

\subsection{Implications of tail calibration}

The unfocused forecaster in Example~\ref{ex:1} shows that probabilistic calibration is a strictly weaker requirement than auto-calibration if the forecast distribution $F$ is not deterministic. Similarly, the following example introduces a \emph{tail unfocused} forecaster, which is probabilistically tail calibration but not tail auto-calibrated; see also Example~\ref{ex:optimistic} in Section~\ref{sec:diagnostics}. 

\begin{example}[Tail unfocused forecaster]\label{ex:12}
Let $\xi, \tau$ be independent random variables, $\xi \sim \Exp(1)$ unit-exponential and $\Q(\tau = +1) = \Q(\tau = - 1) = 1/2$. Define $a_{\tau} = \log \left( (2 + \tau) / 2 \right)$. Let $Y = \xi$ and, given $\tau$, let $F$ be the cdf of $\xi + a_{\tau}$. The random variable $Y$ is unit-exponential, while $F$ is the conditional cdf of $Y$ with a random location shift of $a_+$ or $a_-$. The \emph{tail unfocused} forecaster $F$ is probabilistically tail calibrated but not tail auto-calibrated for $Y$. 
\end{example}

In concert with the following proposition, it follows that tail auto-calibration is a strictly stronger requirement than probabilistic tail calibration. 

\begin{propo}\label{prop:impl1}
    Let $\Cc\subseteq \Bb \subseteq \Aa$ be $\sigma$-algebras. Suppose that $F$ is tail $\Bb$-calibrated for $Y$, and that the convergence in Definition \ref{def:TAC_new} is in $L^\infty$. Then $F$ is also tail $\Cc$-calibrated for $Y$.
\end{propo}

\begin{proof}
We have
	\begin{align*}
		&\frac{\Q(\ZFt \le u, Y > t \mid \Cc)}{\E[1 - F(t) \mid \Cc]} - u
		= \E \left[
		\frac{\Q(\ZFt \le u, Y > t \mid \Bb)}{\E[1 - F(t) \mid \Cc]}
		\;\Big|\; \Cc
		\right] - u \\
		&= \E \left[
		\frac{\Q(\ZFt \le u, Y > t \mid \Bb)}{\E[1 - F(t) \mid \Bb]} \cdot
		\frac{\E[1 - F(t) \mid \Bb]}{\E[1 - F(t) \mid \Cc]}
		\;\Big|\; \Cc
		\right]
		- u \E \left[
		\frac{\E[1 - F(t) \mid \Bb]}{\E[1 - F(t) \mid \Cc]}
		\;\Big|\; \Cc
		\right] \\
		&= \E \left[
		\left\{
		\frac{\Q(\ZFt \le u, Y > t \mid \Bb)}{\E[1 - F(t) \mid \Bb]}
		- u
		\right\} \cdot
		\frac{\E[1 - F(t) \mid \Bb]}{\E[1 - F(t) \mid \Cc]}
		\;\Big|\; \Cc
		\right].
	\end{align*}
	Since the term in curly braces tends to zero in $L^\infty$, $F$ is tail $\Cc$-calibrated for $Y$. 
\end{proof}

\begin{corollary}\label{cor}
If the convergence in Definition \ref{def:TAC_new} is in $L^\infty$ for $\Bb = \sigma(F)$, tail auto-calibration implies probabilistic tail calibration. The converse is false.
\end{corollary}
\begin{proof}
    This is a consequence of Proposition \ref{prop:impl1} with $\Bb = \sigma(F)$ and $\Cc = \{ \varnothing, \Omega \}$. For the converse, see Examples~\ref{ex:12} and~\ref{ex:optimistic}.
\end{proof}

Proposition \ref{prop:impl1} shows that the larger the conditioning set $\Bb$, the stronger the notion of tail calibration. Since proper scoring rules are not available as an evaluation tool for tails of forecast distributions, one could alternatively rank forecasts for extreme events by the size of the $\sigma$-algebra $\Bb$ with respect to which they are tail calibrated. \cite{StrahlZiegel2017} formalize a similar concept for classical notions of calibration: a forecast is \emph{cross-calibrated} if it is calibrated with respect to the information of all other relevant forecasters. This approach could equivalently be employed using notions of tail calibration, see Section \ref{sec:btail}.

We can also identify relationships between standard notions of calibration and their tail counterparts. These relationships are summarized in Figure \ref{fig:implications}.

\begin{propo}\label{prop:probtotail}
    \begin{enumerate}[label=(\alph*)]
        \item If $\sigma(F)\subset \Bb$, then $\Bb$-calibration implies tail $\Bb$-calibration.
        In particular, auto-calibration implies tail auto-calibration but the converse is false.
        \item Probabilistic calibration does not imply probabilistic tail calibration or vice versa. 
\end{enumerate}
\end{propo}

\begin{proof}
    Suppose that $F$ is $\Bb$-calibrated with $\sigma(F) \subset \Bb$. For $B \in \Bb$, $u \in [0,1]$, and $t < x_Y$, we obtain
    \begin{align*}
        \E[\one\{F_t(Y - t) \le u, \, Y>t\}\one_B]  &= \E[\E[\one\{Y \le F^{-1}(u(1-F(t))+F(t))\}\one\{Y>t\}\mid \Bb]\one_B]\\
        &= \E[(u(1-F(t)) + F(t) - F(t))\one_B] = u \E[(1-F(t))\one_B],
    \end{align*}
    since the conditional distribution of $Y$ given $\Bb$ is $F$, and $F$ is assumed to be continuous. This implies $\E\bigl[\one\{\ZFt \le u, \, Y>t\} - u (1-F(t)) \mid \Bb\bigr] = 0$ and thus $\Q\bigl(\ZFt \le u, \, Y>t \mid \Bb\bigr) = u \E[1-F(t) \mid \Bb]$, showing the sufficiency in part~(a). For the converse, see Example~\ref{ex:nonrandom}.
    
    Part~(b) is shown by Examples~\ref{ex:nonrandom} and~\ref{ex:unfoc}.
\end{proof}

\begin{figure}
    \centering
    \includegraphics[width=0.5\textwidth]{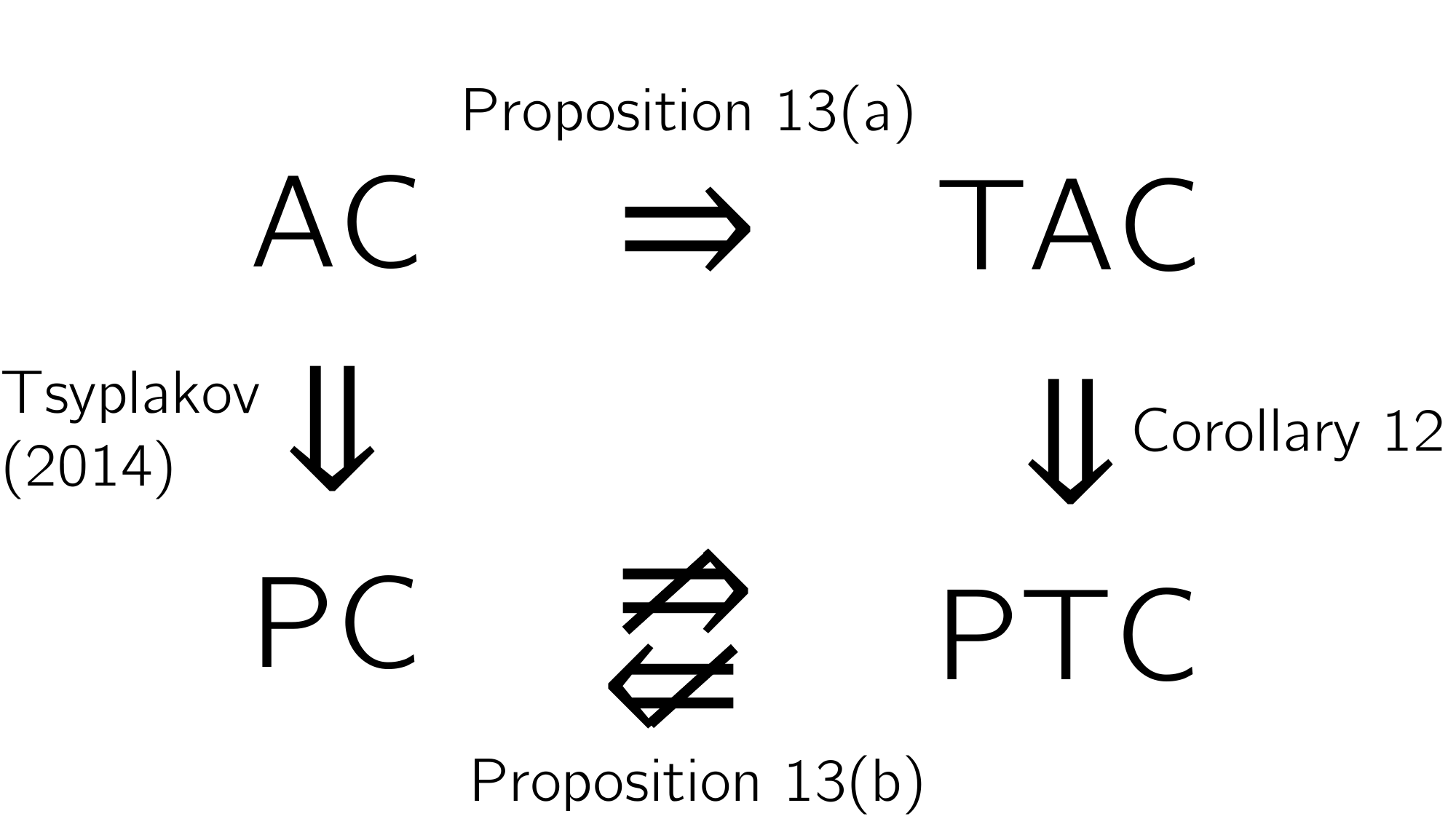}
    \caption{Hierarchies of the notions of calibration and tail calibration. Auto-calibration (AC) implies probabilistic calibration (PC), as well as tail auto-calibration (TAC), which in turn implies probabilistic tail calibration (PTC) (under technical conditions).}
    \label{fig:implications}
\end{figure}

While auto-calibration implies tail auto-calibration, probabilistic calibration does not imply probabilistic tail calibration. Marginal calibration also does not imply its tail counterpart, see Appendix~\ref{app:marginal} in the supplement. This highlights a shortcoming of classical definitions of calibration when interest is in the tails---in this sense, we extend the negative result of \cite{BrehmerStrokorb2019} regarding scoring rules and extremes to notions of calibration. The proposed notions of tail calibration therefore provide information that is not available from existing forecast verification tools.

\section{Probabilistic tail calibration and peaks-over-threshold models}\label{sec:mda}

We argue that probabilistic tail calibration should routinely be checked to evaluate probabilistic forecasts with regards to their tails. Therefore, it is interesting to understand how stringent the requirement of probabilistic tail calibration is. In this section, we provide sufficient conditions. Consider the following assumption on the outcome variable $Y$:
\begin{enumerate}[label=(A\arabic*)]
\item \label{ass:Ygpd} 
	The distribution of $Y$ is in the max-domain of attraction of a nondegenerate distribution.
\end{enumerate}
    Assumption \ref{ass:Ygpd} is standard in extreme value theory, and by \citet{Balkema.deHaan.1974}, it is equivalent to the existence of $\xi \in \R$ and a scaling function $\sigma_Y(\,\cdot\,) > 0$ such that
	\begin{equation}
	\label{eq:Ygpd}
		\lim_{t \to x_Y} \Q \bigl(
			Y > t + \sigma_Y(t) x \mid Y > t 
		\bigr)
		= 1 - \GPD_{1,\xi}(x) 
		= (1 + \xi x)_+^{-1/\xi}, \quad \text{for all $x \ge 0$},
	\end{equation}
    where $(y)_+ = \max(y,0)$ and where $(1+\xi x)_+^{-1/\xi}$ is to be understood as $\exp(-x)$ when $\xi=0$. The parametric family in \eqref{eq:Ygpd} specifies the generalized Pareto distribution (GPD) with shape parameter $\xi \in \R$ and unit scale. For a general scale parameter $\sigma > 0$, we define $\GPD_{\sigma,\xi}(x) = \GPD_{1,\xi}(x/\sigma)$, $x \ge 0$. The GPD is the basis of the peaks-over-threshold method in \citet{Davison-Smith.1990} used to approximate the tails of distributions satisfying~\ref{ass:Ygpd}. The shape parameter determines the rate of tail decay, with power-law decay for $\xi>0$, exponential decay for $\xi=0$, and polynomial decay towards a finite upper bound for $\xi<0$.
    
    Next, consider a similar assumption on the random forecast distribution $F$, formulated in terms of the conditional forecast excess distribution $F_t$ in \eqref{eq:Ftx}: 
\begin{enumerate}[label=(A\arabic*)]
  \setcounter{enumi}{1}
\item \label{ass:Fgpd}
	The upper endpoint $x_F = \sup \{ x \in \R : F(x) < 1 \} \in \R \cup \{+\infty\}$ of $F$ is deterministic and satisfies $x_F = x_Y$, and there exist $\eta \in \R$ and a deterministic scaling function $\sigma_F(\,\cdot\,) > 0$ so that
	\begin{equation}
	\label{eq:Fgpd}
		\lim_{t \to x_F}
		\Q \left( 
			\left| 
				1 - F_t(\sigma_F(t)x) 
				- (1 + \eta x)_+^{-1/\eta} 
			\right| 
			> \eps 
			\, \Big| \, 
			Y > t
		\right)
		= 0, \quad \text{for all $x \ge 0$, $\eps > 0$}.
	\end{equation}
\end{enumerate}
	In words, conditionally on $Y > t$, the rescaled conditional forecast excess distribution $F_t(\sigma_F(t)\point)$ converges in probability to a fixed GPD.

\begin{propo}[Sufficient condition for probabilistic tail calibration]
\label{pro:suffptcmtc}
	Suppose that Assumptions~\ref{ass:Ygpd} and~\ref{ass:Fgpd} are fulfilled and that $\lim_{t \to x_Y} \Q(Y > t) / \E[1 - F(t)] = 1$. The following conditions are equivalent:
	\begin{enumerate}[label=(\alph*)]
	\item $\xi = \eta$ and $\lim_{t \to x_Y} \sigma_F(t)/\sigma_Y(t) = 1$;
	\item $F$ is probabilistically tail calibrated for $Y$.
	\end{enumerate}
\end{propo}

For the proof of Proposition~\ref{pro:suffptcmtc}, see the more general Proposition~\ref{pro:suffptcmtc2} in Appendix~\ref{app:marginal} in the supplement.
If the scale parameter $\sigma_F(t)$ in \eqref{eq:Fgpd} is allowed to be random, then $F$ can be probabilistically tail calibrated for $Y$ even when the GPD shape parameters $\xi$ and $\eta$ in \eqref{eq:Ygpd} and \eqref{eq:Fgpd} differ. For instance, if $\Delta \sim \Gamma(1/\gamma,1/\gamma)$ for some $\gamma > 0$ and if, conditionally on $\Delta$, the distribution of $Y$ is $F = \Exp(\Delta)$, i.e., $\Q(Y \le x \mid \Delta) = 1-e^{-\Delta x} = F(x)$ for $x \ge 0$, then the unconditional distribution of $Y$ is $\GPD_{1,\gamma}$, so that \eqref{eq:Ygpd} holds with $\xi = \gamma$, while $F_t$ is $\Exp(\Delta)$ too, and thus $F_t(x/\Delta) = 1 - e^{-x}$, which is \eqref{eq:Fgpd} with random scale parameter $\sigma_F(t) = 1/\Delta$ and with $\eta = 0$.

Assumptions~\ref{ass:Ygpd}, \ref{ass:Fgpd}, and property~(a) of Proposition~\ref{pro:suffptcmtc} are sufficient but not necessary conditions for probabilistic tail calibration.
Indeed, according to Proposition~\ref{propo:tailequiv}, for non-random probabilistic forecasts, it suffices that $\Q(Y>t)/(1-F(t)) \to 1$ as $t \to x_Y$.

\section{Examples and diagnostic tools}\label{sec:diagnostics}

In practice, we observe forecast-observation pairs $(F_{1}, y_{1}), \dots, (F_{n}, y_{n})$ that are realizations of $(F, Y)$, and evaluation must be performed using this finite sample. Let $\Ii_{t} = \{ i \in \{ 1, \dots, n \} : y_{i} > t \}$ be the set of indices corresponding to observations that exceed $t$, put $n_{t} = |\Ii_{t}|$, and let $F_{i, t}$ denote the conditional excess distribution of forecast $F_{i}$ at threshold $t$.

\subsection{Probabilistic tail calibration}
\label{sec:PTCdiag}

The convergence relations~\eqref{eq:tailBCt} and~\eqref{eq:tailBC:sev} hold uniformly in $u \in (0, 1]$; a precise statement is given in Lemma~\ref{lem:unifconv-u} in Appendix~\ref{app:proofs} in the supplement. Hence, for $F$ to be tail $\Bb$-calibrated, the graph of
\[
	u \mapsto \frac{\Q\bigl(\ZFt \le u, Y > t \mid \Bb\bigr)}{\E[1-F(t)\mid \Bb]}
\]
should be close to the diagonal in the supremum distance. This motivates the following diagnostic plots.

First, consider the simplest case, $\Bb = \{\varnothing, \Omega\}$, corresponding to probabilistic (tail) calibration. Probabilistic calibration is easy to assess in practice by checking whether the PIT values $F_{1}(y_{1}), \dots, F_{n}(y_{n})$ resemble a sample from a standard uniform distribution. To assess probabilistic tail calibration using a finite sample, consider, for some large threshold $t$, the $n_{t}$ excess PIT values
\[
    \zit = F_{i, t} (y_{i} - t) = \frac{F_i(y_i)-F_i(t)}{1-F_i(t)}
\]
for $i \in \Ii_{t}$, which correspond to realizations of the random variable $\ZFt$. 

We can then calculate the combined ratio
\begin{equation}\label{eq:emp_ratio}
    \hat{R}_{t}(u) 
    = \frac{\sum_{i \in \Ii_{t}} \one\{\zit \leq u \}}{ \sum_{i = 1}^{n} (1 - F_{i}(t))},
\end{equation}
and plot this quantity as a function of $u \in [0, 1]$. If the curve lies close to the diagonal, there is evidence to suggest that the forecasts are probabilistically tail calibrated. Results in \cite{AllenEtAl2023} suggest that this diagnostic plot is also valid for non-extreme choices of $t$, and we recover a \emph{pp}-plot of the PIT values $z_i = F_i(y_i)$ when $t = -\infty$, from which we can assess standard probabilistic calibration.

If the graph of the combined ratio $u \mapsto \hat{R}_t(u)$ does not lie along the diagonal, then it may be beneficial to consider additional diagnostics derived from \eqref{eq:tailBC:occ} and \eqref{eq:tailBC:sev} instead of \eqref{eq:tailBCt}. For example, if the graph is consistently below the diagonal, this could either be because the numerator is too small, suggesting a bias in the conditional distribution, or because the denominator is too large, suggesting a bias in the forecast exceedance probabilities.

The first condition, \eqref{eq:tailBC:occ}, requires that the forecasts issue reliable threshold exceedance probabilities. This can be assessed empirically by checking whether the occurrence ratio 
\begin{equation}
	\label{eq:ocratio}
	\frac{\sum_{i=1}^{n} \one\{ y_{i} > t\}}{\sum_{i=1}^{n} (1 - F_{i}(t))},
\end{equation}
tends towards one as the threshold $t$ increases. A ratio larger (smaller) than one suggests that the forecasts under-estimate (over-estimate) threshold exceedance probabilities. Previously, other works have evaluated the performance of probabilistic forecasts with regards to tails by considering threshold exceedance probabilities \citep[see e.g.][]{WilliamsEtAl2014,TaillardatFougeresETAL2019,AllenEtAl2023}. Assessing these forecasts using reliability diagrams is closely related to our assessment of the occurrence ratio using \eqref{eq:ocratio}.

The second condition, \eqref{eq:tailBC:sev}, considers the conditional excess distribution given that the threshold has been exceeded, allowing us to evaluate forecasts for the severity of an extreme event. Under probabilistic tail calibration, the excess PIT values $\zit$ for $i \in \Ii_{t}$ should resemble a sample from a standard uniform distribution. This can again be examined using a histogram or a \emph{pp}-plot, and the shape of the diagnostic plot can be used to identify whether the conditional distribution is heavy- or light-tailed \citep{AllenEtAl2023}.

While \eqref{eq:emp_ratio} and \eqref{eq:ocratio} are defined in terms of a fixed threshold $t$, the threshold could also be chosen adaptively. For example, we may wish to assess forecasts made at different points in time or space, for which the notion of an extreme event varies. In this case, the diagnostic tools described above could be employed with a threshold that depends on the forecast index $i$, corresponding to high conditional quantiles of the outcome, for example. Choosing the threshold $t$ in practice is often a compromise between $t$ being large enough to represent the tail regime of the data, whilst also permitting sufficient data from which robust results can be obtained. Formal tests for tail calibration could be derived, for example based on the uniformity of the excess PIT values \citep[e.g.][]{MitchellWeale2023}, though these are not considered in depth here; this is discussed further in Section \ref{sec:casestudy}.

\begin{example}[Tail auto-calibration does not imply auto-calibration]\label{ex:nonrandom} 
    Let $F$ be a deterministic forecast for which there exists a threshold $t_0$ such that $\Q(Y > t_0) > 0$ and $F(y) = \Q(Y \le y)$ for all $y \ge t_0$. Trivially, $F$ is tail auto-calibrated, and hence probabilistically tail calibrated. However, $F(y)$ and $\Q(Y \le y)$ can be very different for $y < t_0$, in which case $F$ is not probabilistically calibrated. For example, consider the outcome $Y \sim \GPD_{1, 1/4}$ together with the forecast distribution $F(x) = \GPD_{4/5, 1/4}(x - 1)$ for $x < 5$ and $F(x) = \GPD_{1, 1/4}(x)$ for $x \ge 5$. 
    Figure \ref{fig:sim_nonrandom} shows the diagnostic plots for probabilistic tail calibration in this case.
\end{example}

\begin{figure}
    \centering
    \includegraphics[width=0.32\textwidth]{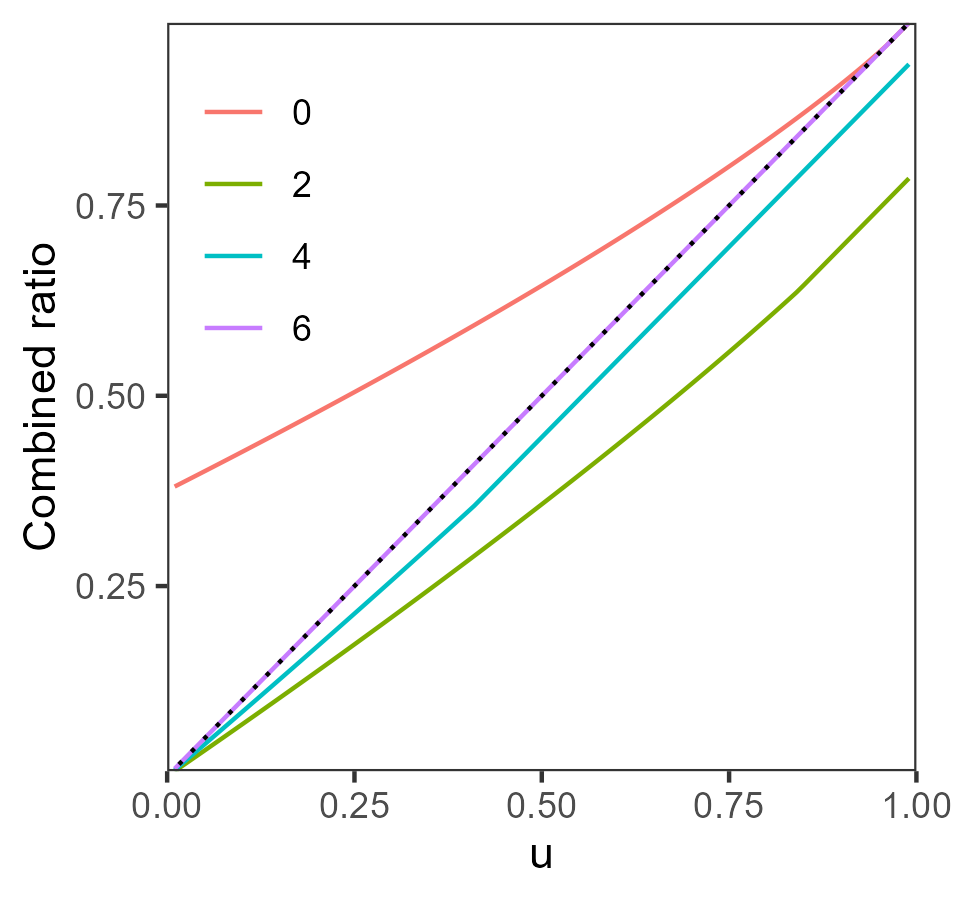}
    \includegraphics[width=0.32\textwidth]{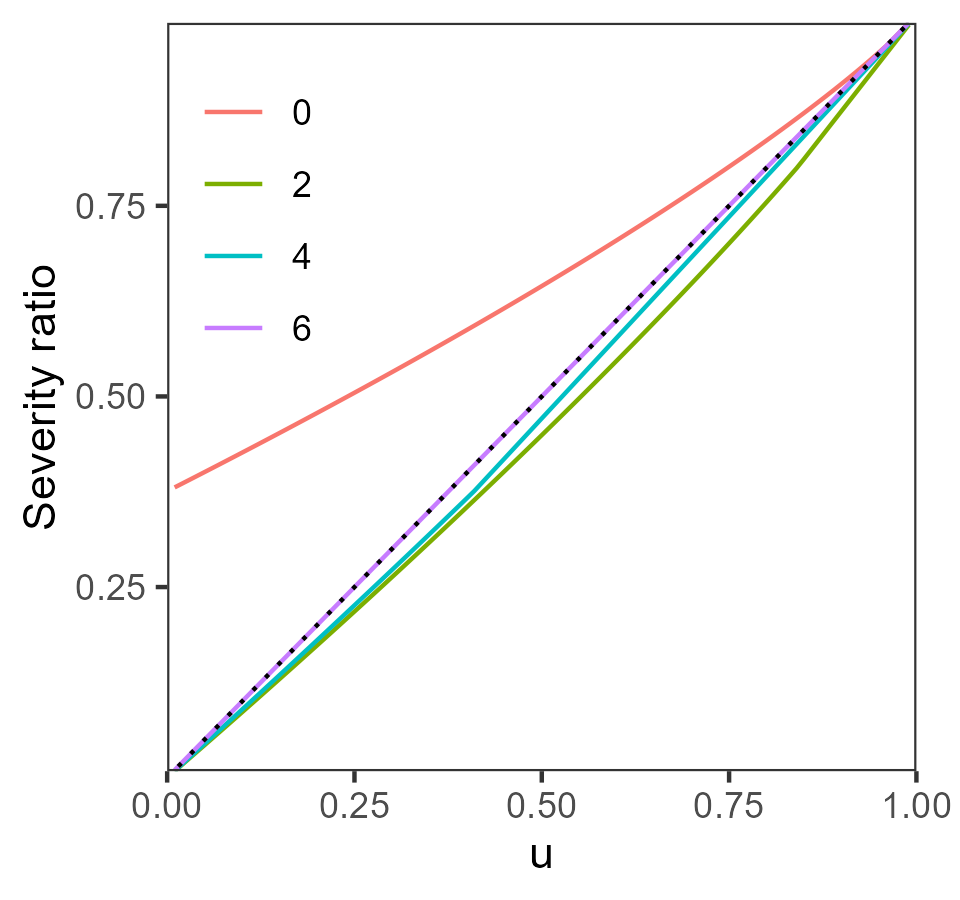}
    \includegraphics[width=0.32\textwidth]{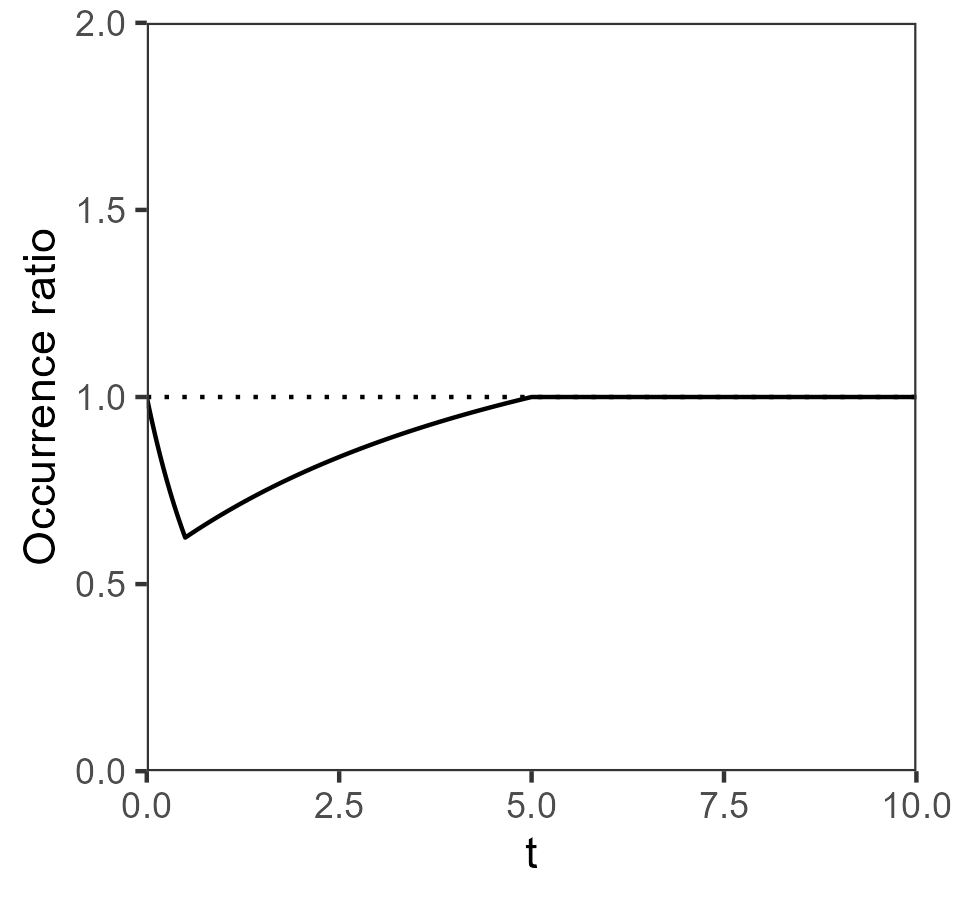}
    \caption{Probabilistic tail calibration diagnostic plots for Example~\ref{ex:nonrandom}. The different colors correspond to different thresholds $t$. Left: combined ratio \eqref{eq:emp_ratio}; middle: \emph{pp}-plot of $\zit$ for $i \in \mathcal{I}_t$; right: occurrence ratio \eqref{eq:ocratio} as a function of $t$.}
    \label{fig:sim_nonrandom}
\end{figure}

\begin{example}[Probabilistic calibration does not imply probabilistic tail calibration]\label{ex:unfoc}
    Consider a generalization of the unfocused forecaster in \cite{GneitingEtAl2007}, see Example \ref{ex:1}. Let $G$ be a distribution function supported on a compact interval with 
continuous positive density. Let $Y \sim G$, let $\tau$ be independent of $Y$ and equal to $+1$ or $-1$ with probability $1/2$ and, given $\tau$, let $F(y) = \frac{1}{2} \{G(y)+G(y+\tau)\}$ for $y \in \R$. Then $F$ is probabilistically calibrated but not necessarily probabilistically tail calibrated; see Appendix~\ref{app:proofs} for details. Figure~\ref{fig:sim_unfocused} illustrates the example for $G = \Unif(0,1)$: the curves for $t = -1$ confirm probabilistic calibration but those for $t > -1$ indicate the lack of probabilistic tail calibration.
\end{example}

\begin{figure}
    \centering
    \includegraphics[width=0.32\textwidth]{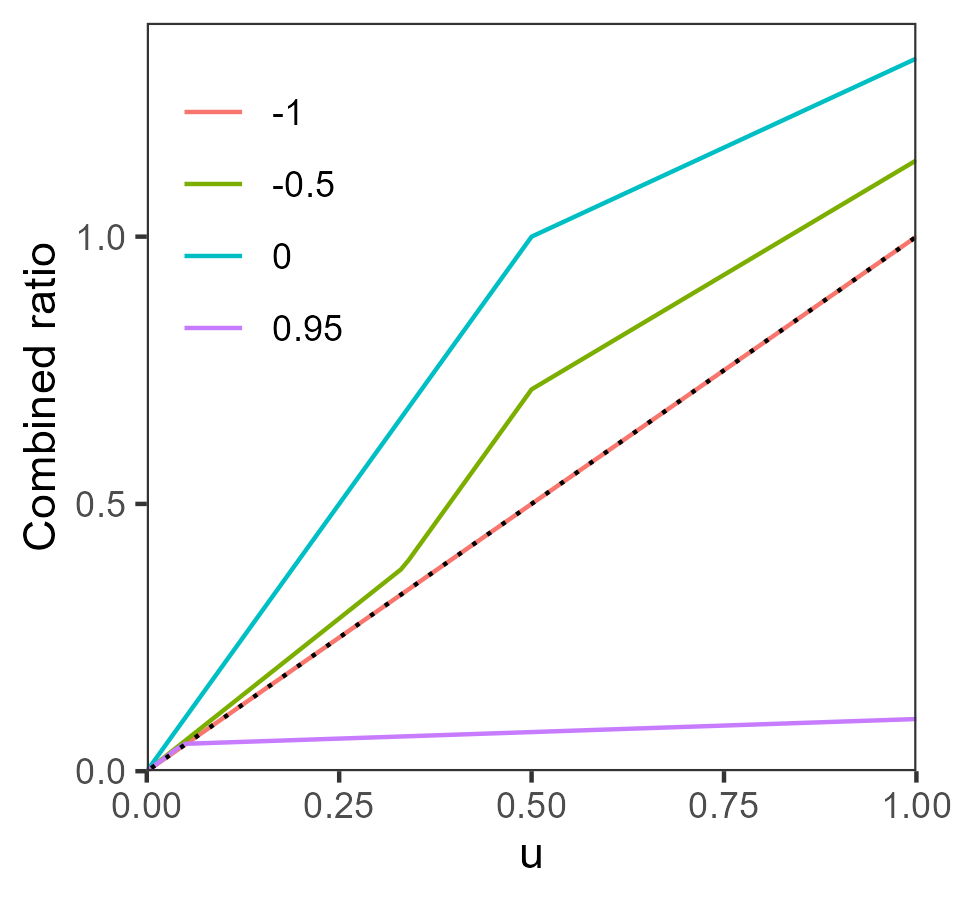}
    \includegraphics[width=0.32\textwidth]{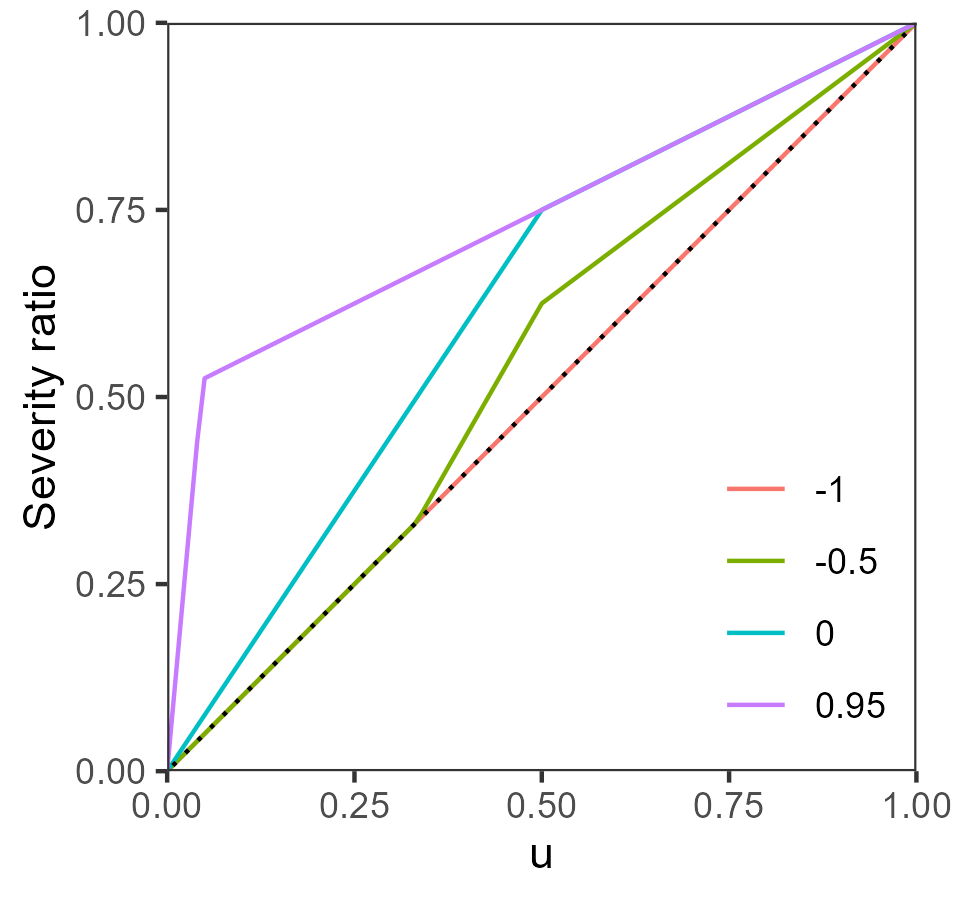}
    \includegraphics[width=0.32\textwidth]{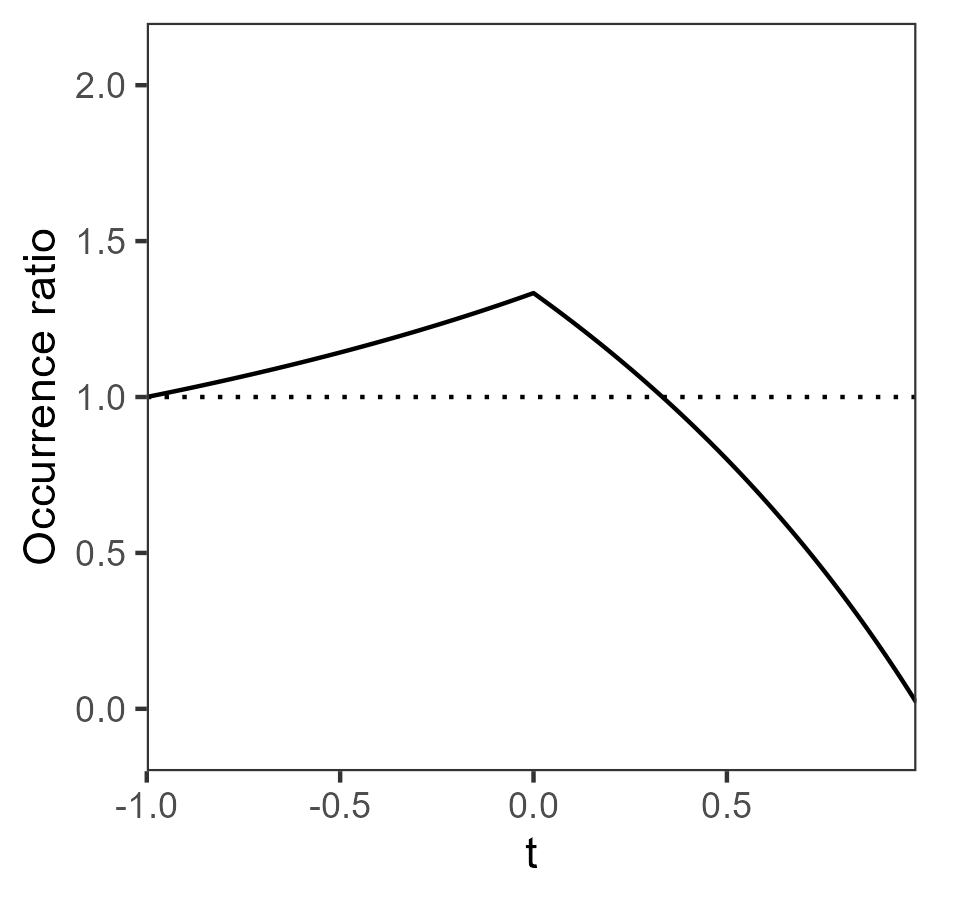}
    \caption{Probabilistic tail calibration diagnostic plots for the unfocused forecaster in Example~\ref{ex:unfoc} with $G$ the uniform distribution on $[0, 1]$. The different colors correspond to different thresholds $t$. Left: combined ratio \eqref{eq:emp_ratio}; middle: \emph{pp}-plot of $\zit$ for $i \in \mathcal{I}_t$; right: occurrence ratio \eqref{eq:ocratio} as a function of $t$.}
    \label{fig:sim_unfocused}
\end{figure}

\begin{example}[Ideal, climatological and extremist forecasters]
	\label{ex:exponential}
    The following example was introduced by \cite{TaillardatEtAl2023} for studying methods to evaluate forecast tails. Let $\Delta \sim \Gamma ( 1 / \gamma, 1 / \gamma )$, where $\gamma > 0$ and $\Gamma (a, b)$ is the gamma distribution with shape $a > 0$ and scale $b > 0$. Conditionally on $\Delta$, let $Y \sim \Exp(\Delta)$ follow an exponential distribution with rate $\Delta$. 
    The unconditional distribution of $Y$ is the heavy-tailed $\GPD_{1, \gamma}$. As in \cite{TaillardatEtAl2023}, results are presented for $\gamma = 1 / 4$. 
    
    Consider three different forecasters: the ideal forecaster, $F_{\mathrm{id}} = \Exp (\Delta)$; the climatological forecaster, $F_{\mathrm{cl}} = \GPD_{1, \gamma}$; and the extremist forecaster, $F_{\mathrm{ex}, \nu} = \Exp (\Delta / \nu)$ for $\nu > 0$.
    The variable $\Delta$ represents a source of information that may or may not be available to the forecasters. The ideal forecaster correctly represents the conditional distribution of $Y$ given $\Delta$ and is auto-calibrated, hence also probabilistically (tail) calibrated. The climatological forecaster does not use the information $\Delta$, instead issuing the unconditional distribution of $Y$ as a forecast, but is nevertheless also auto-calibrated. 
    The extremist forecaster correctly predicts that, conditionally on $\Delta$, the outcome $Y$ follows an exponential distribution, but with a multiplicative bias on the rate (here, $\nu = 1.4$), and is not probabilistically (tail) calibrated. Given the information $\Delta$, the climatological forecaster belongs to a different tail regime than the outcome (heavy-tailed rather than light-tailed), whereas the extremist forecaster belongs to the correct tail regime (light-tailed).
    
    Figure \ref{fig:ss_ptc_reldiag} displays the diagnostic plots for the tail calibration of the three forecasters using $n = 10^{6}$ independent realizations of $\Delta$ and $Y$, see also Figure \ref{fig:ss_ptc_reldiag2} in the supplement.
    The extremist forecast clearly has a heavier tail than the true distribution of the outcomes, and the nature of this miscalibration does not change as the threshold is increased.
    
    \begin{figure}[t]
        \centering
        \includegraphics[width=0.32\textwidth]{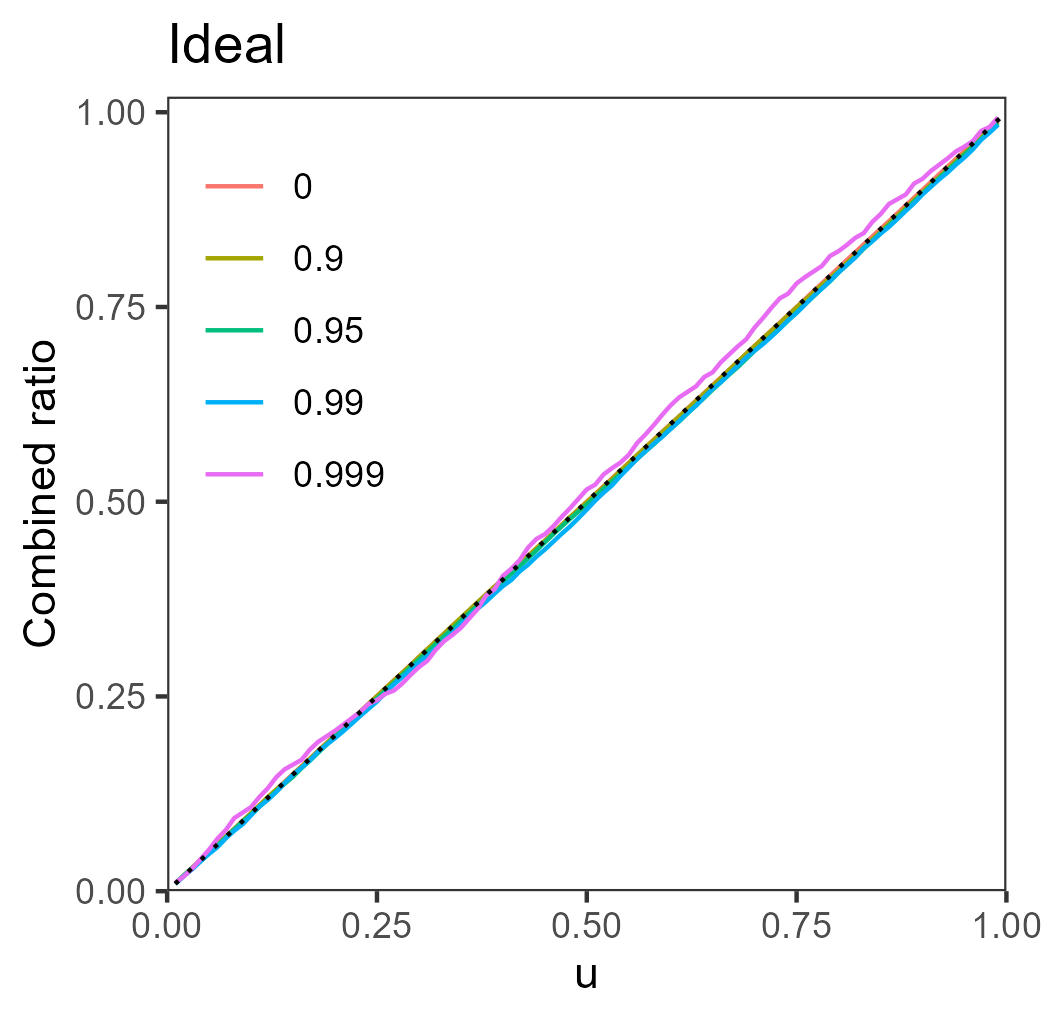}
        \includegraphics[width=0.32\textwidth]{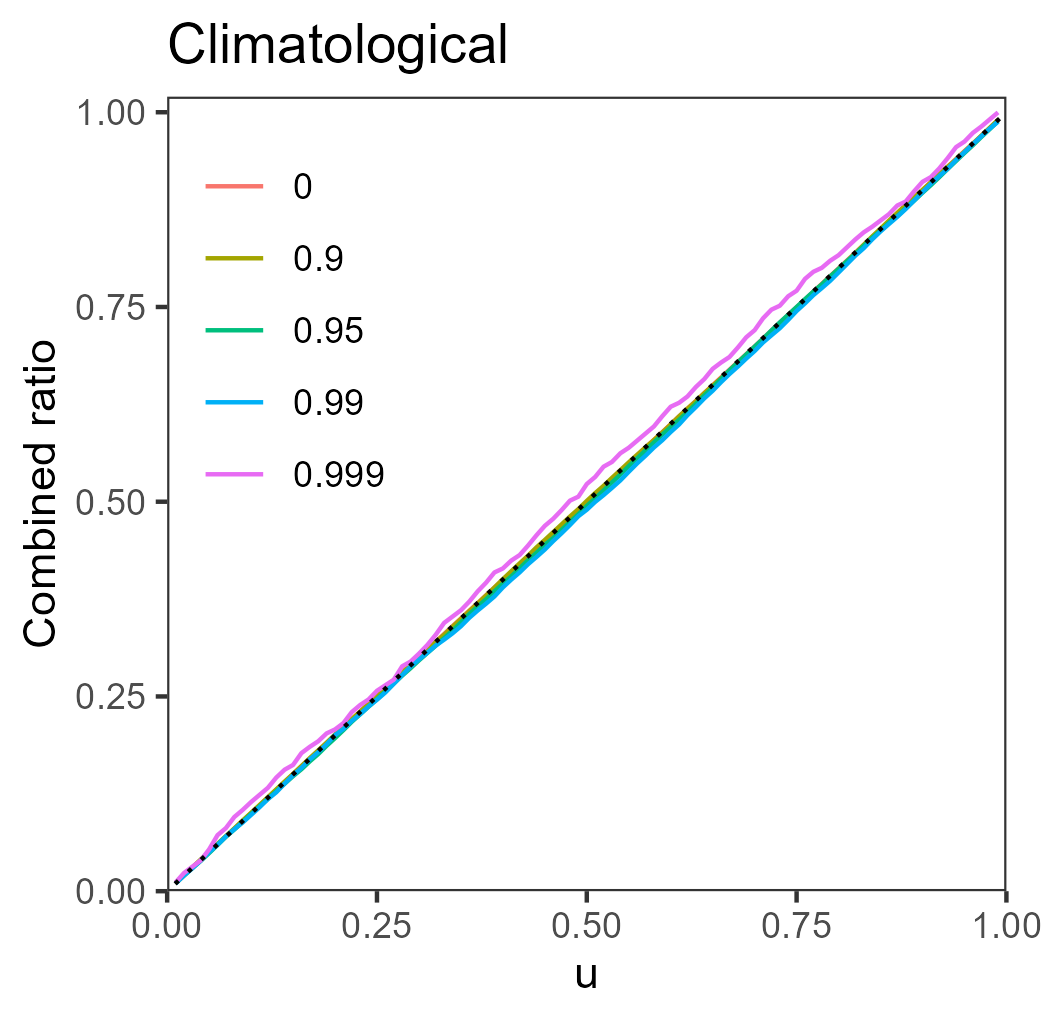}
        \includegraphics[width=0.32\textwidth]{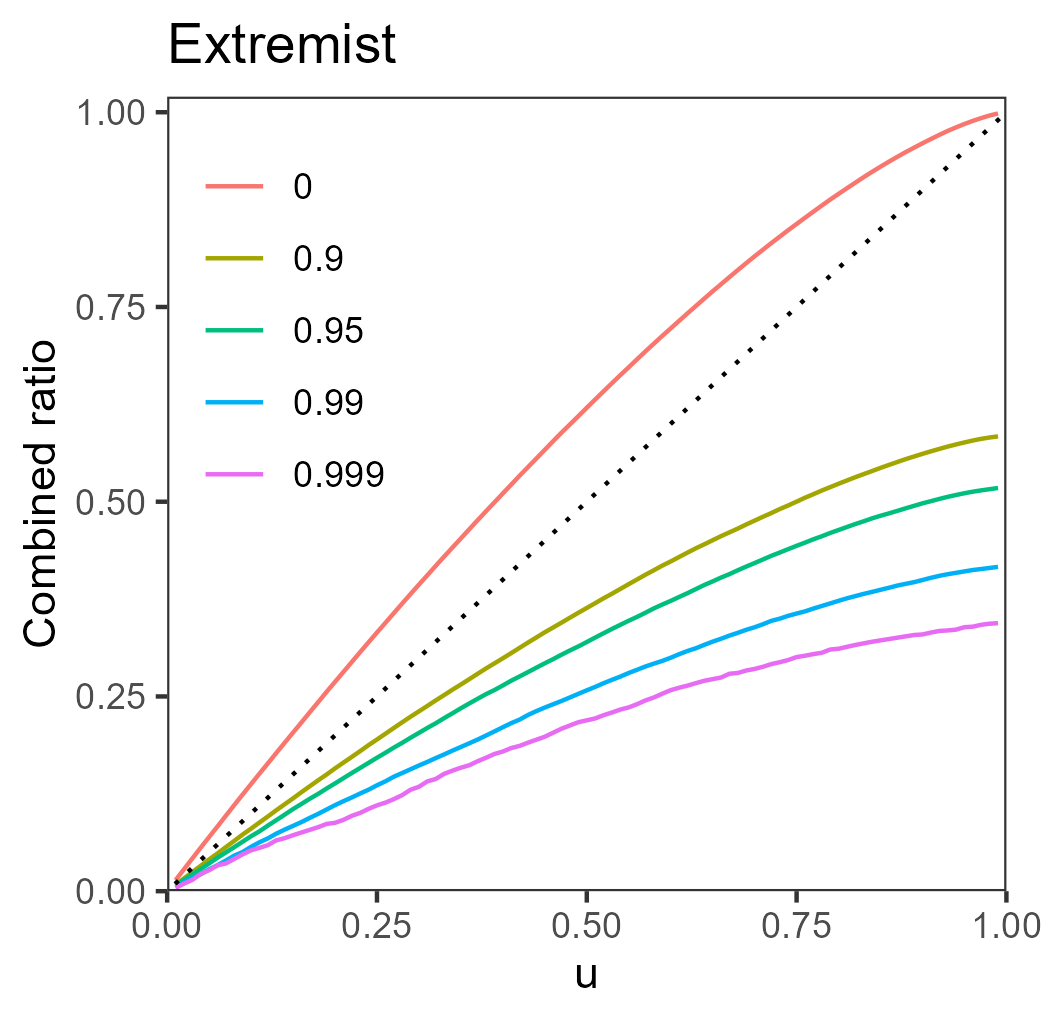}
    \caption{Combined probabilistic tail calibration diagnostic plots \eqref{eq:emp_ratio} for the ideal, climatological, and extremist forecasters in Example~\ref{ex:exponential}. Results are shown for five thresholds $t$, expressed in terms of quantiles of the $10^{6}$ observations.}
    \label{fig:ss_ptc_reldiag}
    \end{figure}
\end{example}

    In Appendix \ref{sec:normsim}, we examine the unfocused forecaster of Example \ref{ex:1} with regards to probabilistic (tail) calibration.

\subsection{Tail $\Bb$-calibration}\label{sec:btail}

While it is most common in practice to assess probabilistic calibration, it is also useful to assess calibration with respect to larger information sets. The same is true when interest is in extremes. For example, the climatological forecaster in Example~\ref{ex:exponential} is probabilistically tail calibrated because it exhibits the same tail behavior as the unconditional distribution of the outcome. However, it does not exhibit the correct tail behavior conditional on $\Delta$, and is therefore not $\sigma(\Delta)$-tail calibrated. In this section, we propose tools to assess tail $\Bb$-calibration for non-trivial $\sigma$-algebras $\Bb$.

To assess tail $\Bb$-calibration, the forecast-observation pairs $(F_{1}, y_{1}), \dots, (F_{n}, y_{n})$ can be stratified into $J \in \mathbb{N}$ mutually disjoint bins, $B_{1}, \dots, B_{J}$, depending on variables representing the information contained in $\Bb$. The diagnostic plots for probabilistic tail calibration can then be applied separately to each bin of the forecasts and observations; if the forecasts appear probabilistically calibrated for all bins, then there is evidence to suggest that they are tail $\Bb$-calibrated. For example, the combined ratio in~\eqref{eq:emp_ratio} can be calculated per bin,
\begin{equation}
\label{eq:Rtj}
    \hat{R}_{t,j}(u) = \frac{\sum_{i \in \Ii_{t} \cap \Jj_{j}} \one\{\zit \leq u \}}{ \sum_{i \in \Jj_{j}} (1 - F_{i}(t))},
\end{equation}
where $\Jj_{j} = \left\{i \in \{1, \dots, n\} : (F_{i}, y_{i}) \in B_{j} \right\}$ for $j= 1, \dots, J$, and this ratio can be displayed as a function of $u \in [0, 1]$ for each of the $J$ bins. Similarly, the excess PIT values $\zit$ and the occurrence ratio \eqref{eq:ocratio} can be assessed for each bin separately.

Displaying the diagnostic plots for each forecast method and each bin leads to a large number of plots to analyze. To simplify visualization, (tail) miscalibration can be summarized in a single value. This aligns with the general framework to assess conditional calibration proposed by \cite{Tsyplakov2011} based on moment conditions and test functions, which is akin to how forecasts are compared in \cite{GiacominiWhite2006} and \cite{NoldeZiegel2017}.
As a measure of tail calibration, consider the supremum distance $\sup_{u \in [0, 1]} | \hat{R}_{t} (u) - u |$, which quantifies the maximum absolute difference between the combined ratio and the diagonal line in the diagnostic plot proposed above. Under probabilistic tail calibration, this difference should be close to zero for large $t$. To assess tail $\Bb$-calibration, $\hat{R}_{t}$ can be replaced by $\hat{R}_{t,j}$, and this distance can be calculated for $j=1, \dots, J$.

In Example~\ref{ex:exponential}, the random variable $\Delta$ controls the information governing the observations, and to assess $\sigma(\Delta)$-calibration, we can stratify the forecast-observation pairs depending on the realized $\Delta$. In practice, where we do not know the true data generating process, $\Delta$ could be replaced by covariate information available to the forecasters, for example. The supremum distances
\begin{equation}
\label{eq:supdist}
	\sup_{u \in [0, 1]} | \hat{R}_{t,j}(u) - u |, \qquad j = 1,\ldots,J,
\end{equation}
for the ideal, climatological and extremist forecasters in Example \ref{ex:exponential} are shown as a function of the threshold in Figure~\ref{fig:ss_ptc_cond}, for $J = 3$ bins for $\Delta$. For the ideal forecaster, the difference between $\hat{R}_{t,j}(u)$ and $u$ is essentially zero for all bins, confirming that it is tail $\sigma(\Delta)$-calibrated. The climatological forecaster, on the other hand, is clearly not calibrated conditionally on $\Delta$, despite being probabilistically tail and tail auto-calibrated. These plots for conditional calibration are shown for the empirical analogue~\eqref{eq:Rtj} of the combined ratio in~\eqref{def:TAC_new}, though an analogous approach can serve to evaluate the occurrence ratio~\eqref{eq:tailBC:occ} or severity ratio~\eqref{eq:tailBC:sev}.

\begin{figure}
    \centering
    \includegraphics[width=0.3\textwidth]{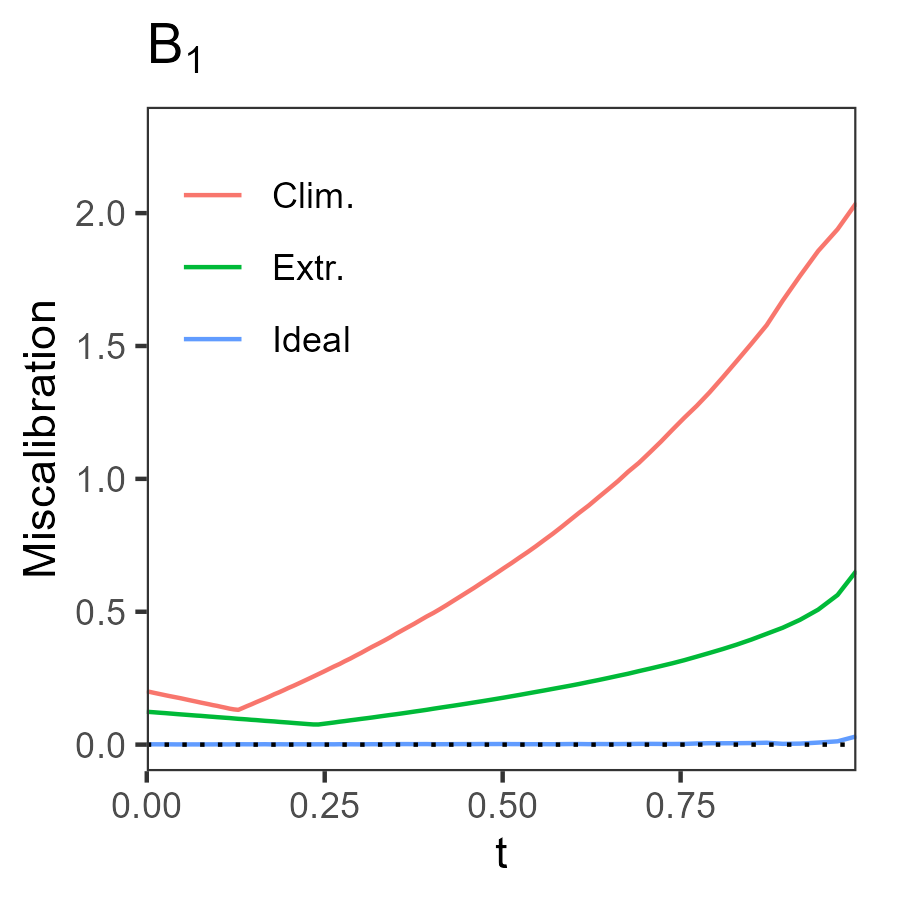}
    \includegraphics[width=0.3\textwidth]{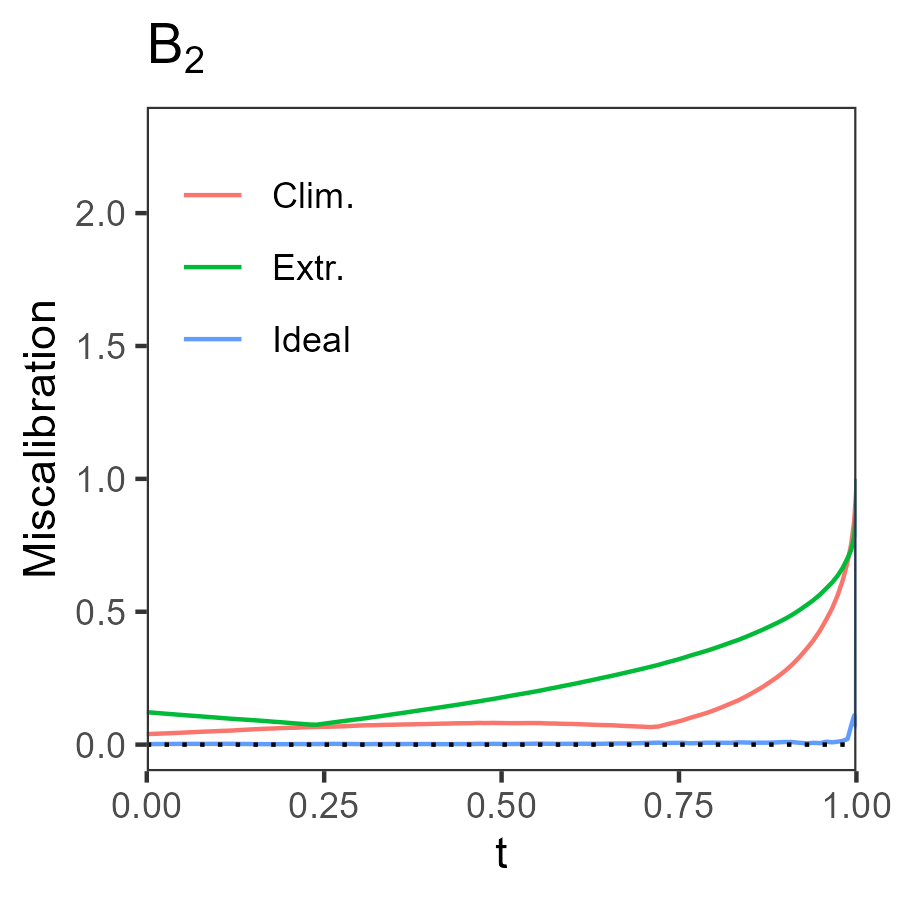}
    \includegraphics[width=0.3\textwidth]{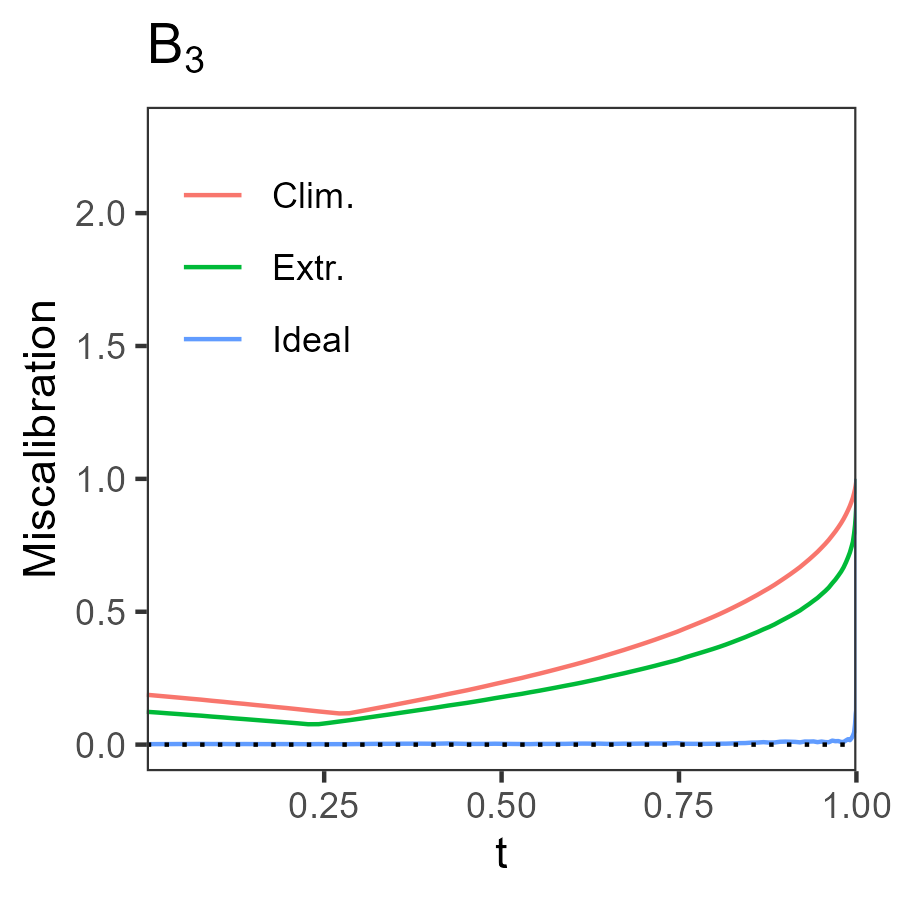}
    \caption{Maximal distance \eqref{eq:supdist} of $\hat{R}_{t,j}$ from the diagonal as a function of $t$ for the climatological (red), ideal (blue), and extremist (green) forecasters in Example~\ref{ex:exponential} for $J=3$ bins (left to right) of $\Delta$. The threshold $t$ is expressed as a quantile of the $10^{6}$ observations.}
    \label{fig:ss_ptc_cond}
\end{figure}

While the climatological forecaster in Example~\ref{ex:exponential} is not tail $\sigma(\Delta)$-calibrated, it is tail auto-calibrated. However, the following example demonstrates that a forecast can be probabilistically tail calibrated but not tail auto-calibrated.

\begin{example}[Probabilistic tail calibration does not imply tail auto-calibration]\label{ex:optimistic}
    Suppose that $\Delta \sim \Gamma(1/\gamma,1/\gamma)$ for some $\gamma > 0$ and, conditionally on $\Delta$, let  
    $X \sim \Exp(\Delta)$.
    Let the random variable $L \sim \GPD_{1,\gamma/2}$ be independent of $(\Delta, X)$.
    Define $Y$ as the second largest of $(X, 2X, L)$.
    Given $\Delta$, the \emph{optimistic forecaster} issues the random probabilistic forecast $F = \Exp(\Delta)$, the conditional distribution of $X$.
    This forecast is probabilistically tail calibrated, but not tail auto-calibrated; see Figure~\ref{fig:sim_stoch} for the diagnostic plots and see Appendix~\ref{app:proofs} for a formal analysis of a more general construction.
    An intuitive explanation is as follows: the marginal distribution of $X$ is $\GPD(1,\gamma)$, which is heavier than the marginal distribution of $L$; it follows that $\Q(Y = X \mid Y > t) \to 1$ as $t \to \infty$ and the marginal tails of $X$ and $Y$ are equivalent, yielding probabilistic tail calibration. However, conditionally on $\Delta$, the tail of $F$ is \emph{lighter} than that of $L$, so that $F$ is optimistic, i.e., it predicts that less severe outcomes will occur than reality, violating tail auto-calibration.
\end{example}

\begin{figure}
    \centering
    \includegraphics[width=0.24\textwidth]{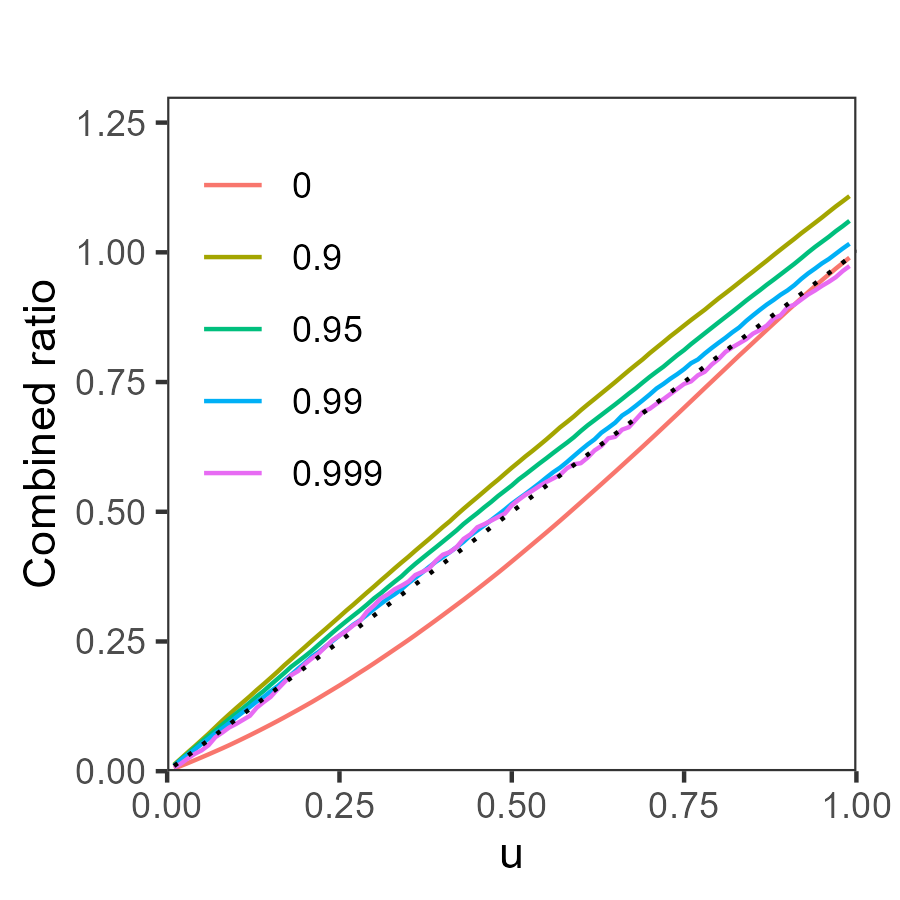}
    \includegraphics[width=0.24\textwidth]{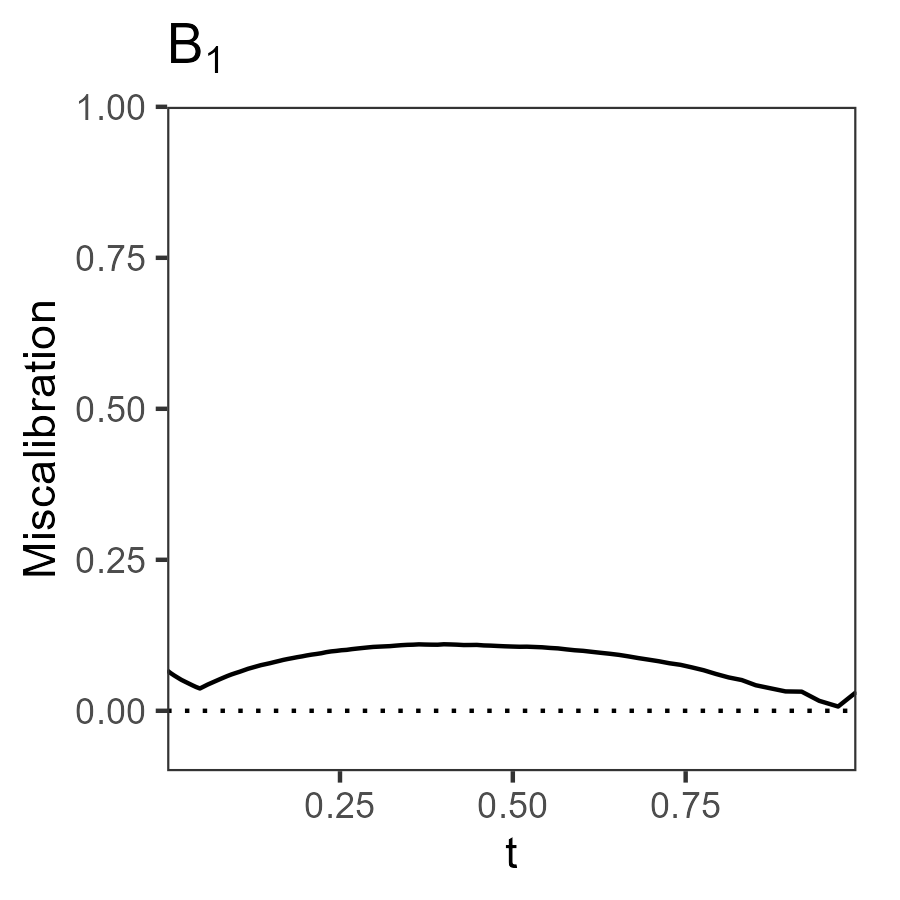}
    \includegraphics[width=0.24\textwidth]{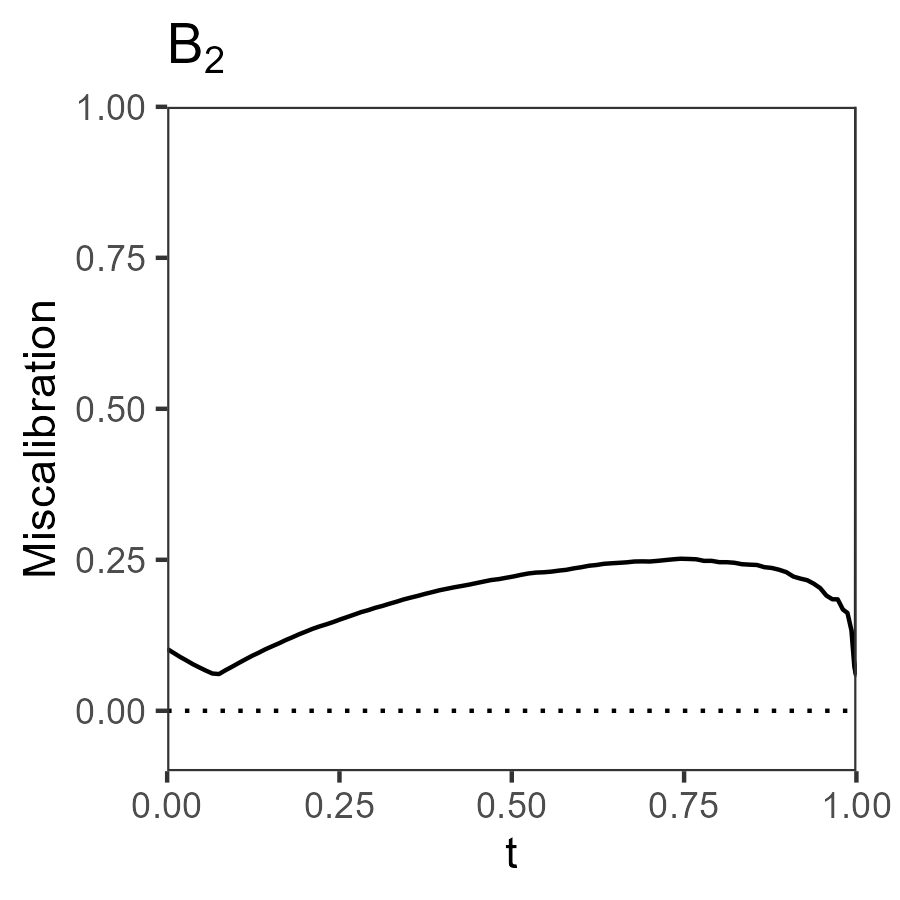}
    \includegraphics[width=0.24\textwidth]{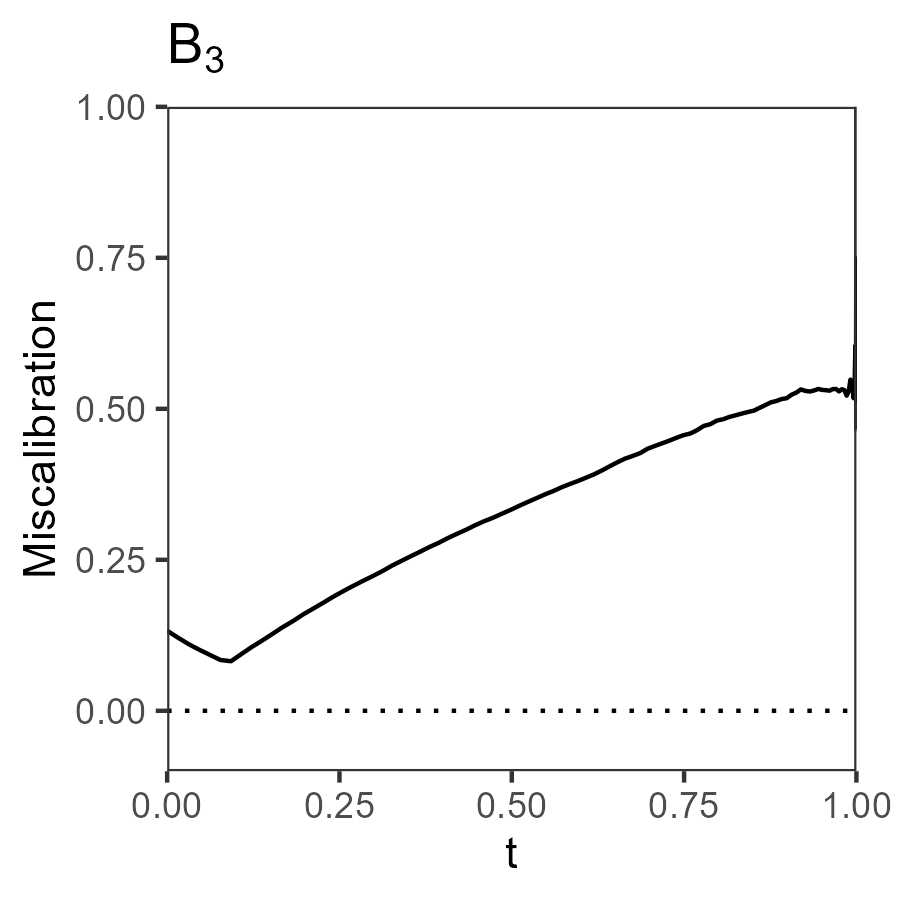}
    \caption{First plot: Combined probabilistic tail calibration diagnostic plot in~\eqref{eq:emp_ratio} for the random forecast in Example~\ref{ex:optimistic}, for five thresholds $t$. Second, third, and fourth plots: Distance \eqref{eq:supdist} of $\hat{R}_{t,j}$, $j = 1, 2, 3$, from the diagonal as a function of $t$ for $J=3$ bins of $\Delta$. In all plots, thresholds are expressed in terms of quantiles of $10^{6}$ observations.}
    \label{fig:sim_stoch}
\end{figure}

\section{Application}\label{sec:casestudy}

Reliable forecasts for extreme precipitation are highly valuable for mitigating the impacts of heavy rain and snowfall. We consider precipitation data from the European Meteorological Network's (EUMETNET) post-processing benchmark dataset \citep[EUPPBench;][]{DemaeyerEtAl2023}, which contains ensemble forecasts issued by the European Centre for Medium-range Weather Forecasts' (ECMWF) Integrated Forecast System (IFS). Forecasts are available daily for 2017 and 2018, at 112 stations across central Europe. We restrict attention to forecasts issued one day in advance, for which 81548 forecast-observation pairs are available. The forecasts are evaluated using precipitation measurements at the stations.

Three forecasting methods are compared. Firstly, we consider the 51-member ensemble forecasts generated by the ECMWF's IFS ensemble. Ensemble forecasts are collections of point forecasts, forming an empirical predictive distribution function. An empirical distribution comprised of 51 sample values is unlikely to be useful when predicting extremes; in particular, the resulting forecast will assert that high thresholds will be exceeded with probability zero. We therefore also assess the performance of a smoothed version of the discrete ensemble forecast, which converts the IFS ensemble to a continuous forecast distribution by assuming that precipitation observations follow a censored logistic distribution, with location and scale parameters that are equal, respectively, to the mean and standard deviation of the corresponding IFS ensemble members. Finally, we present results for a statistical post-processing model that aims to remove systematic biases in the IFS ensemble to yield calibrated forecasts.

The post-processing method assumes that precipitation follows a certain parametric distribution, with location and scale parameters that are affine functions of the IFS ensemble mean and standard deviation, respectively. This simple approach is generally referred to as ensemble model output statistics in the post-processing literature \citep{GneitingEtAl2005}. We implement this post-processing model with two choices of parametric distribution: a logistic distribution, which is commonly employed in studies of precipitation forecasts \citep[see e.g.][]{MessnerEtAl2014}; and a generalized extreme value (GEV) distribution, which has been proposed as a means to better forecast extreme outcomes \citep{LerchThorarinsdottir2013,Scheuerer2014}. Both distributions are censored below at zero. The parameters of the post-processing models are estimated by minimizing the continuous ranked probability score (CRPS) over 20 years of twice-weekly reforecasts at the corresponding station and lead time; further details regarding the data can be found in \cite{DemaeyerEtAl2023}.

Figure~\ref{fig:cs_pit} displays the probabilistic (tail) calibration diagnostic plots for the four forecast methods. For the discrete ensemble forecasts, the diagnostics for tail calibration are adapted as described in Remark \ref{rem:F_jumps}. Results are shown at thresholds of 5, 10, and \qty{15}{mm}, roughly corresponding to the 97$^{\text{th}}$, 99$^{\text{th}}$, and 99.5$^{\text{th}}$ percentiles of the previously observed precipitation measurements (across all stations). Similar conclusions are drawn when separate thresholds are employed at each station, corresponding to high quantiles of the local climatology, see Appendix \ref{app:cs_extra} in the supplement. The raw IFS ensemble forecasts are under-dispersed and positively biased, whereas both post-processing methods yield considerably more reliable forecasts. The IFS forecasts are additionally strongly miscalibrated in the tails. This becomes more severe as the threshold of interest increases, suggesting that the forecasts are not probabilistically tail calibrated. Smoothing the IFS ensemble unsurprisingly improves tail calibration, but appears to worsen the overall calibration of the forecasts in this example. While the censored logistic post-processing method improves upon the calibration of the raw ensembles, the resulting forecasts are still acutely miscalibrated when predicting extreme outcomes. Post-processing with a censored GEV distribution, on the other hand, yields forecasts that are both probabilistically calibrated and tail calibrated. Despite exhibiting contrasting tail behavior, the two distributions yield forecasts with very similar average CRPS; the logistic forecasts are roughly 1\% more accurate than the GEV forecasts.

The poor tail calibration of the logistic post-processed forecasts is a deficiency of the chosen distribution. This is irrespective of the number of ensemble members: Even if we had no covariate information, we could still issue the climatological distribution of the outcome, which is auto-calibrated and therefore tail auto-calibrated. However, when employing existing forecast evaluation techniques, practitioners would generally conclude that this post-processing method generates reliable forecasts, which is misleading. In contrast, our diagnostic tools for tail calibration inform us that the forecast tails are too light, resulting in unreliable forecasts for extreme events. This information can be used to improve forecasts, for example by encouraging practitioners to experiment with alternative post-processing methods, based on different distributional assumptions.

It is not trivial to derive a unified test for tail calibration based on the combined ratio at \eqref{eq:emp_ratio}, and this becomes yet more difficult considering the limit as $t \to x_{Y}$. When assessing calibration with respect to a finite threshold, \cite{MitchellWeale2023} propose separately testing whether the predicted threshold exceedance is correct, and whether the excess PIT values resemble a sample from a standard uniform distribution. Here, we implement a binomial test for the occurrence ratio, and a Kolmogorov--Smirnov test for uniformity of the excess PIT values, assuming the forecast-observation pairs are independent and identically distributed. For thresholds \qty{10}{mm} and \qty{15}{mm}, all methods other than the GEV post-processing approach return a p-value of zero (up to numerical precision) on both tests, providing evidence against these forecasts being tail calibrated. The p-values for the GEV forecasts are $0.00, 0.01$, and $0.83$ when testing the severity ratio using a Kolmogorov--Smirnov test, and $0.59, 0.13$, and $0.04$ when testing the occurrence ratio using a binomial test, at thresholds 5, 10, and \qty{15}{mm}, respectively. Both a nonparametric bootstrap that resamples forecast-observation pairs or the central limit theorem in combination with the delta method can be used to obtain pointwise confidence intervals for the occurrence, severity, and combined ratios; the latter method is discussed in detail in Appendix~\ref{app:cs_confint} in the Supplement.

\begin{figure}[!h]
    \centering
    \includegraphics[width=0.3\textwidth]{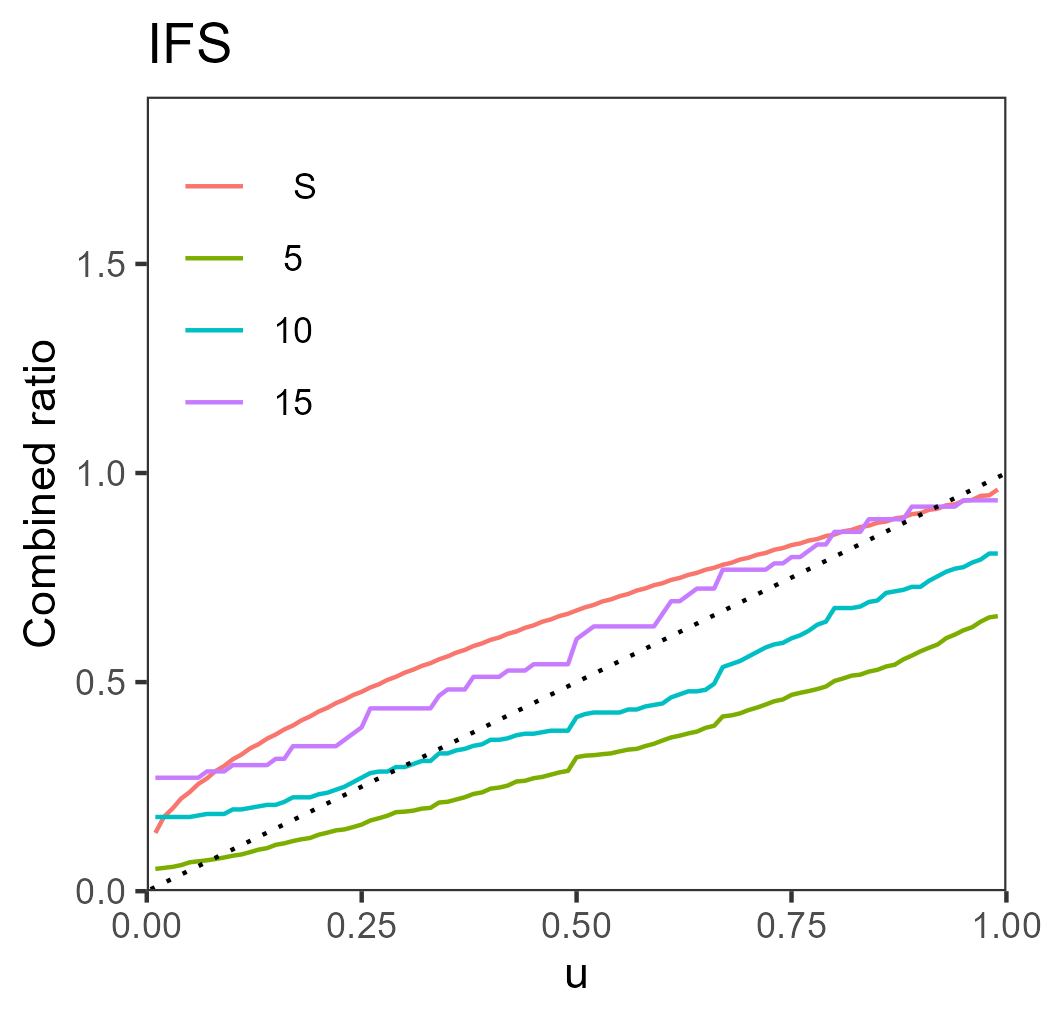}
    \includegraphics[width=0.3\textwidth]{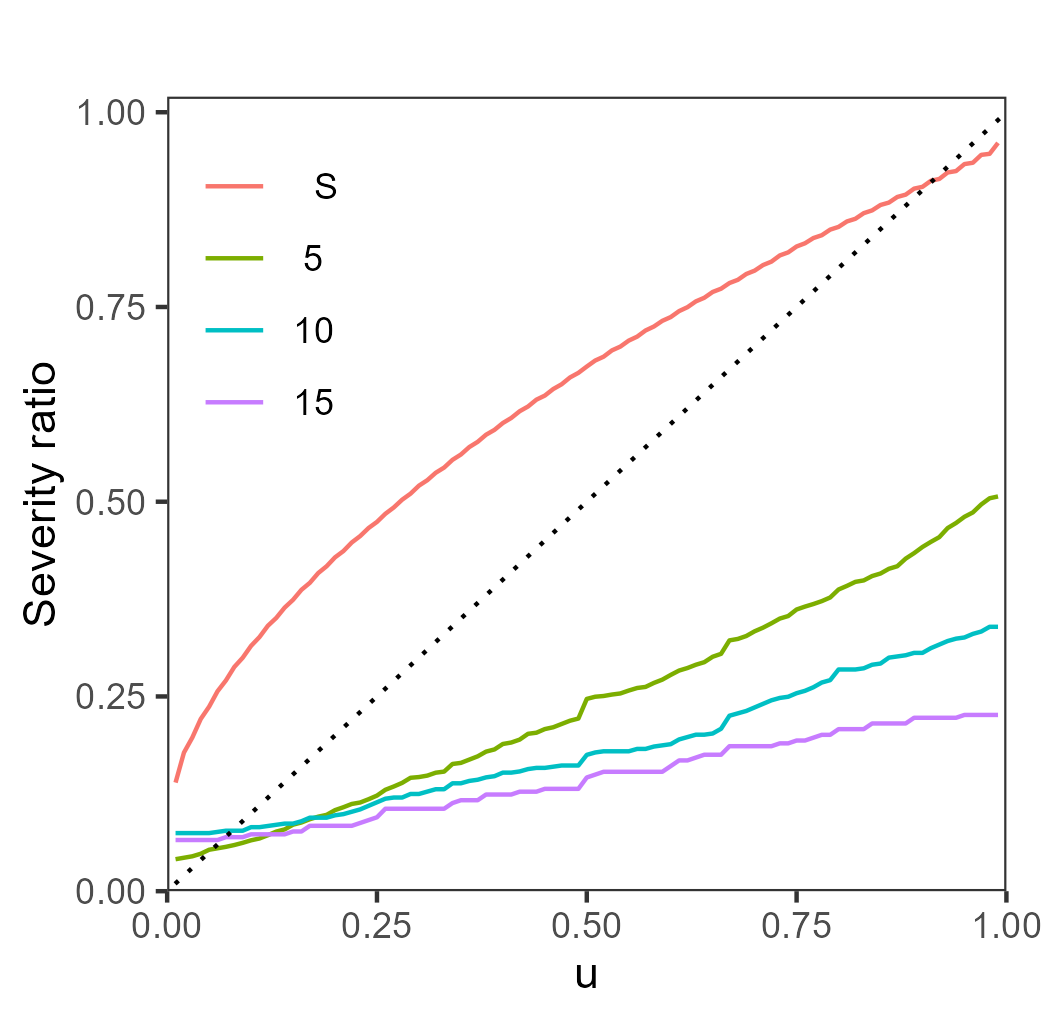}
    \includegraphics[width=0.3\textwidth]{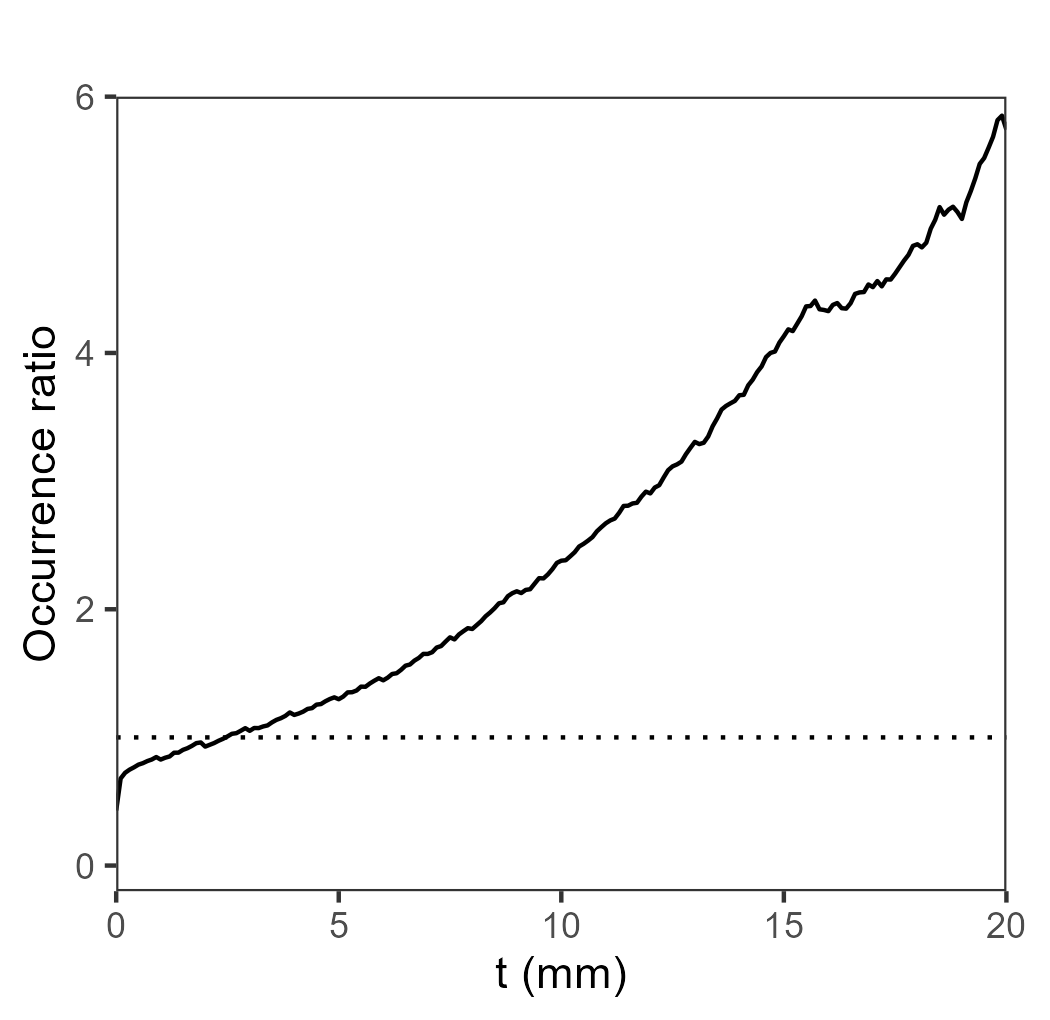}
    \includegraphics[width=0.3\textwidth]{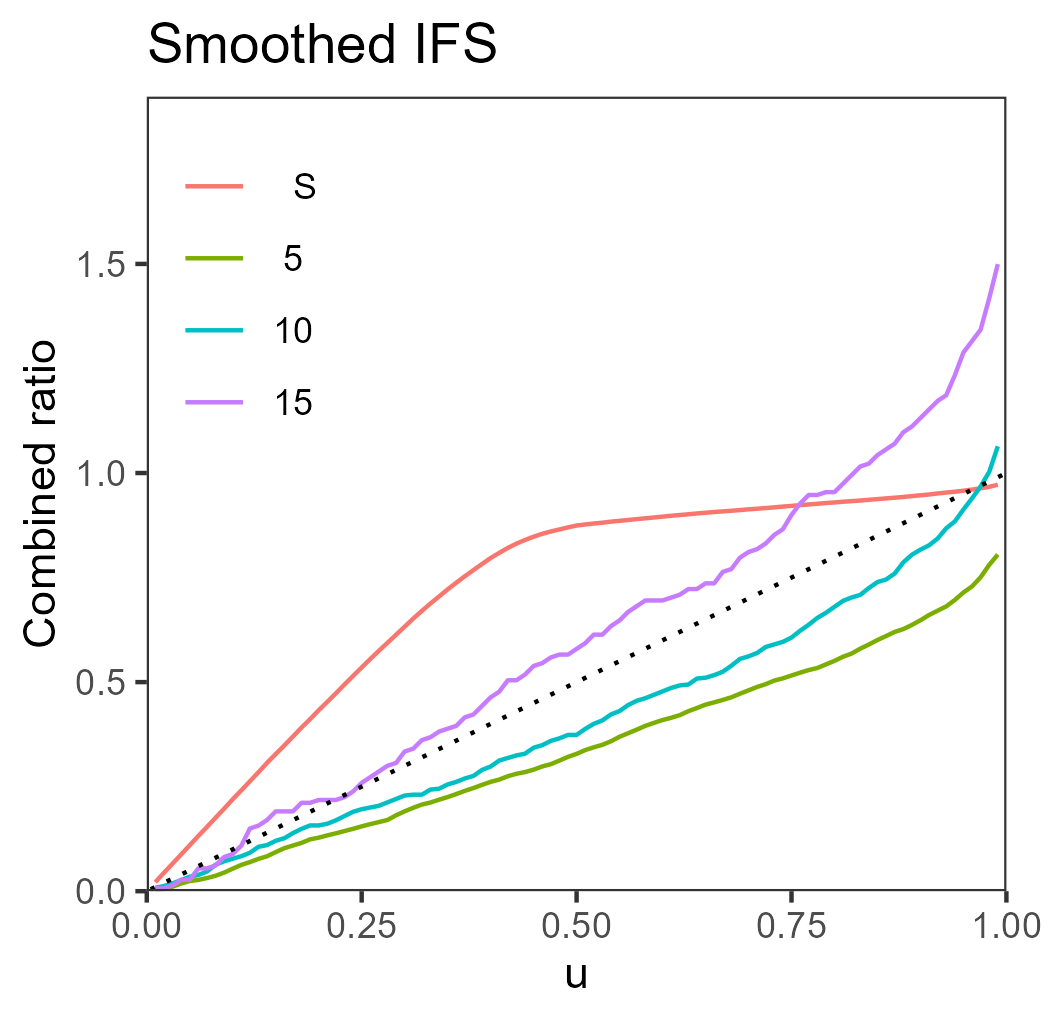}
    \includegraphics[width=0.3\textwidth]{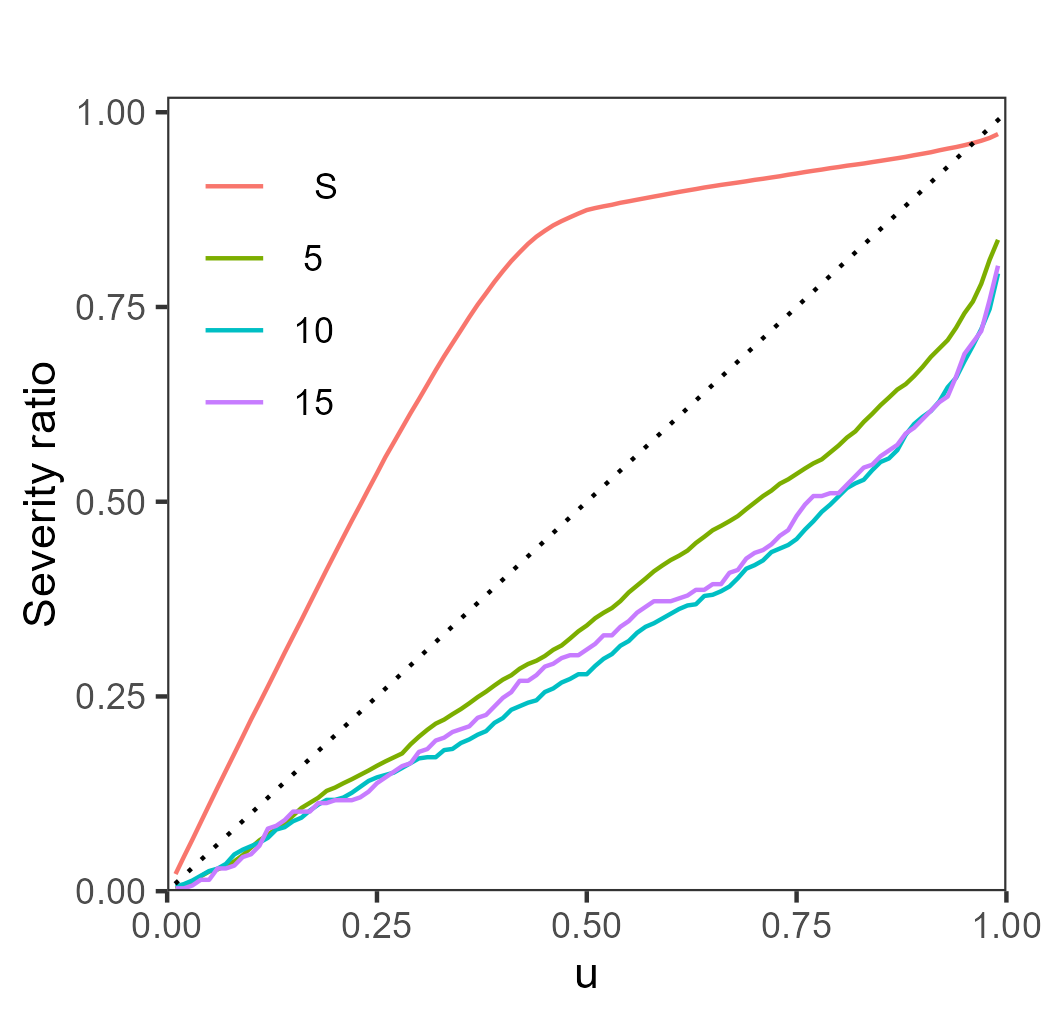}
    \includegraphics[width=0.3\textwidth]{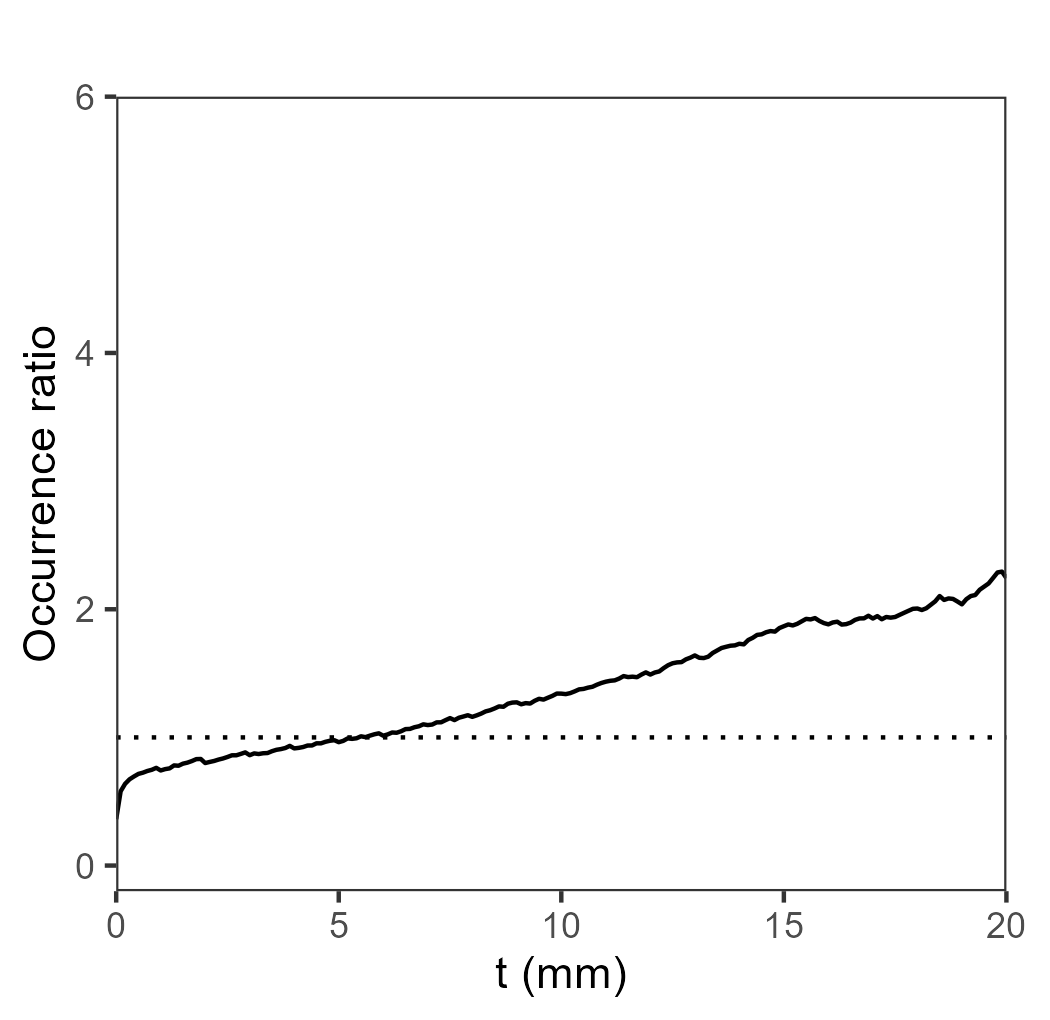}
    \includegraphics[width=0.3\textwidth]{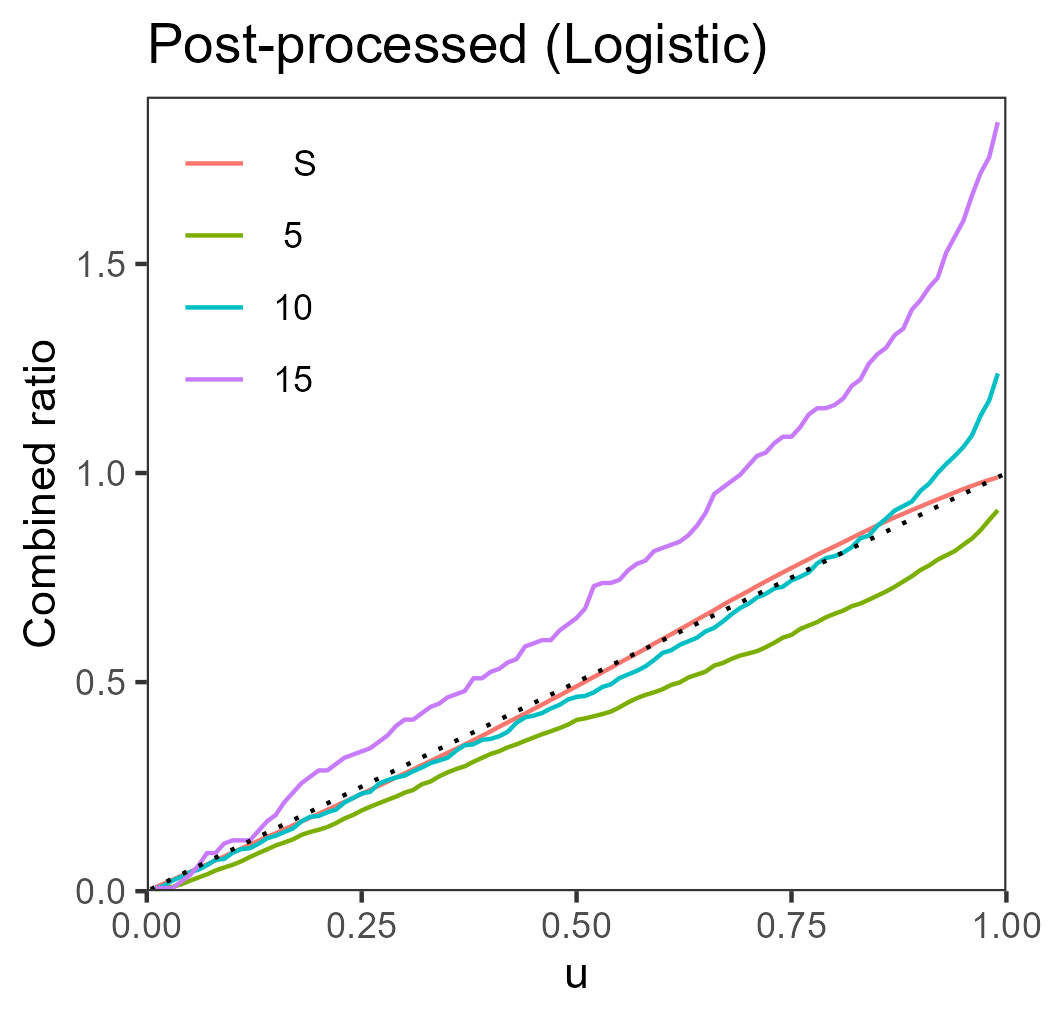}
    \includegraphics[width=0.3\textwidth]{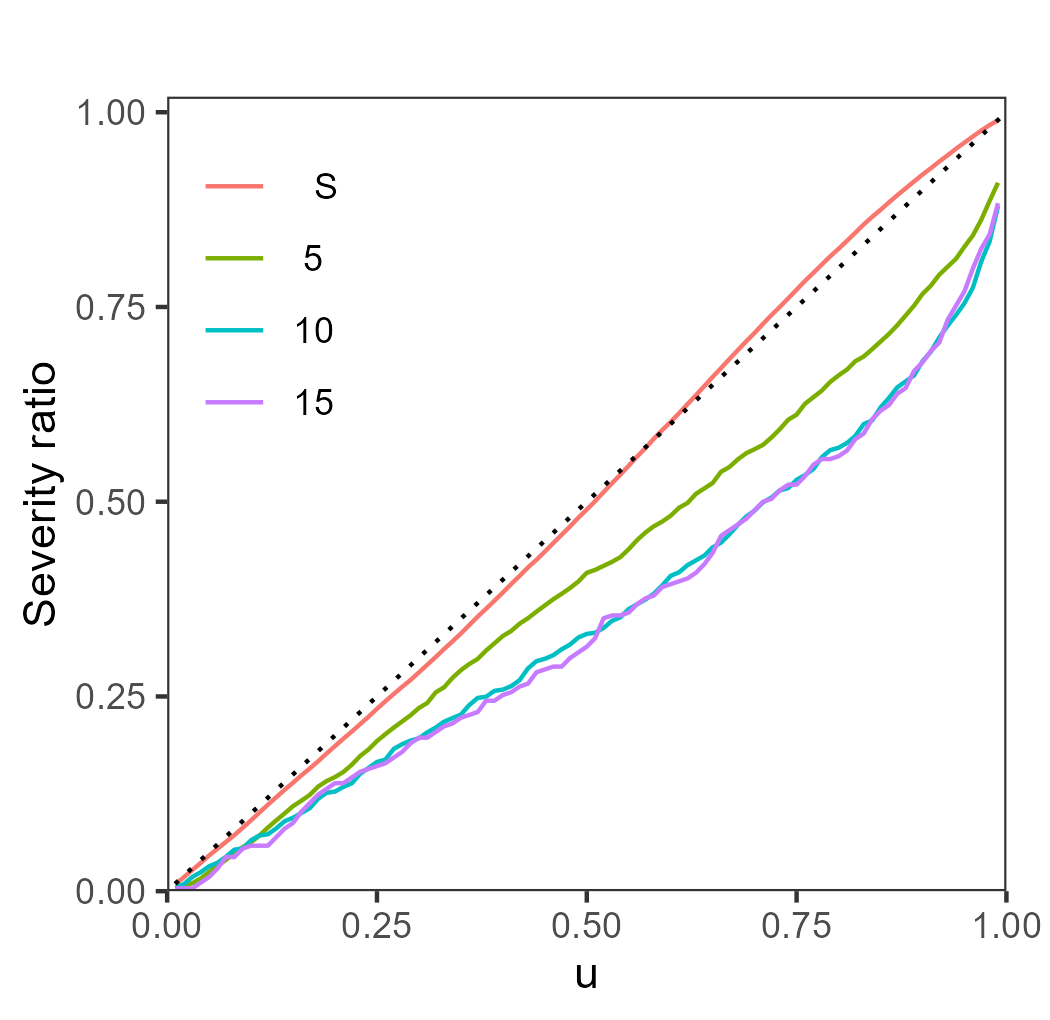}
    \includegraphics[width=0.3\textwidth]{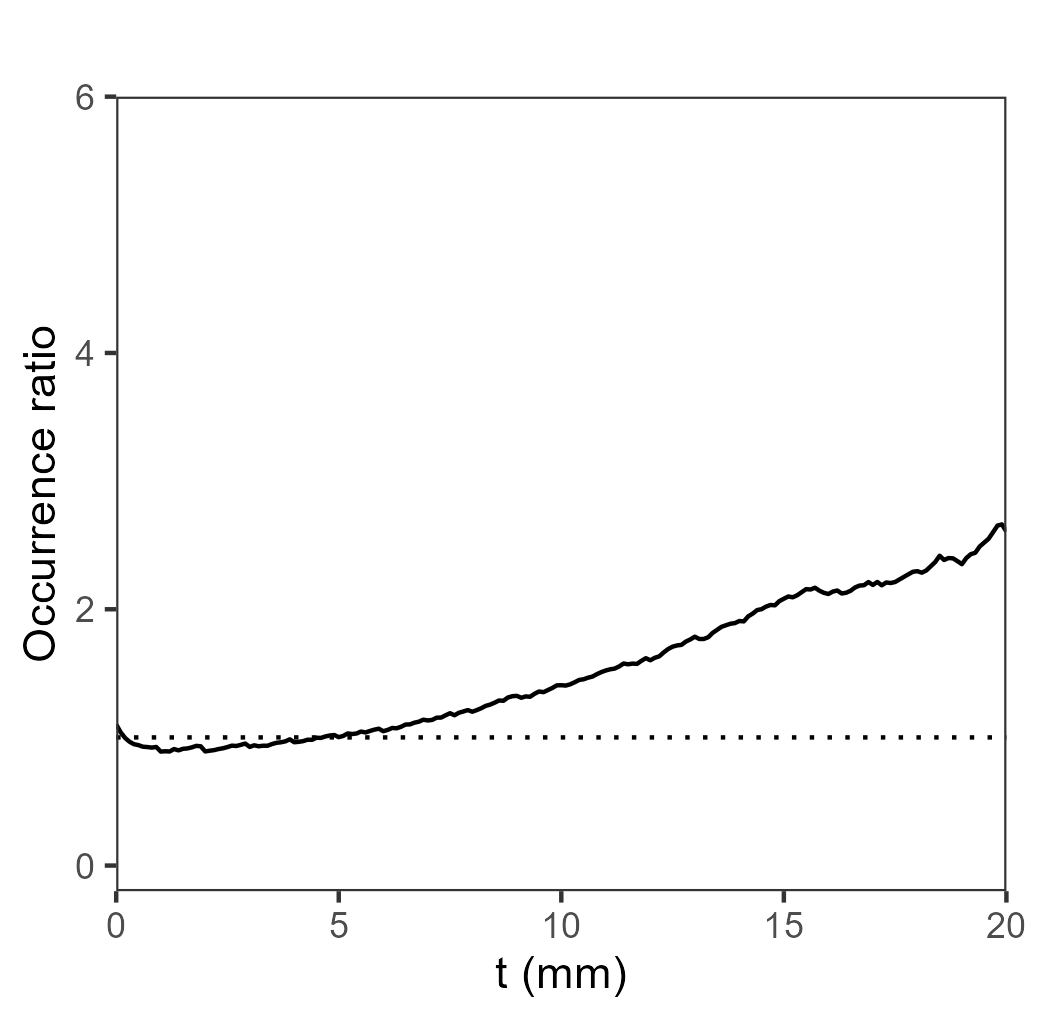}
    \includegraphics[width=0.3\textwidth]{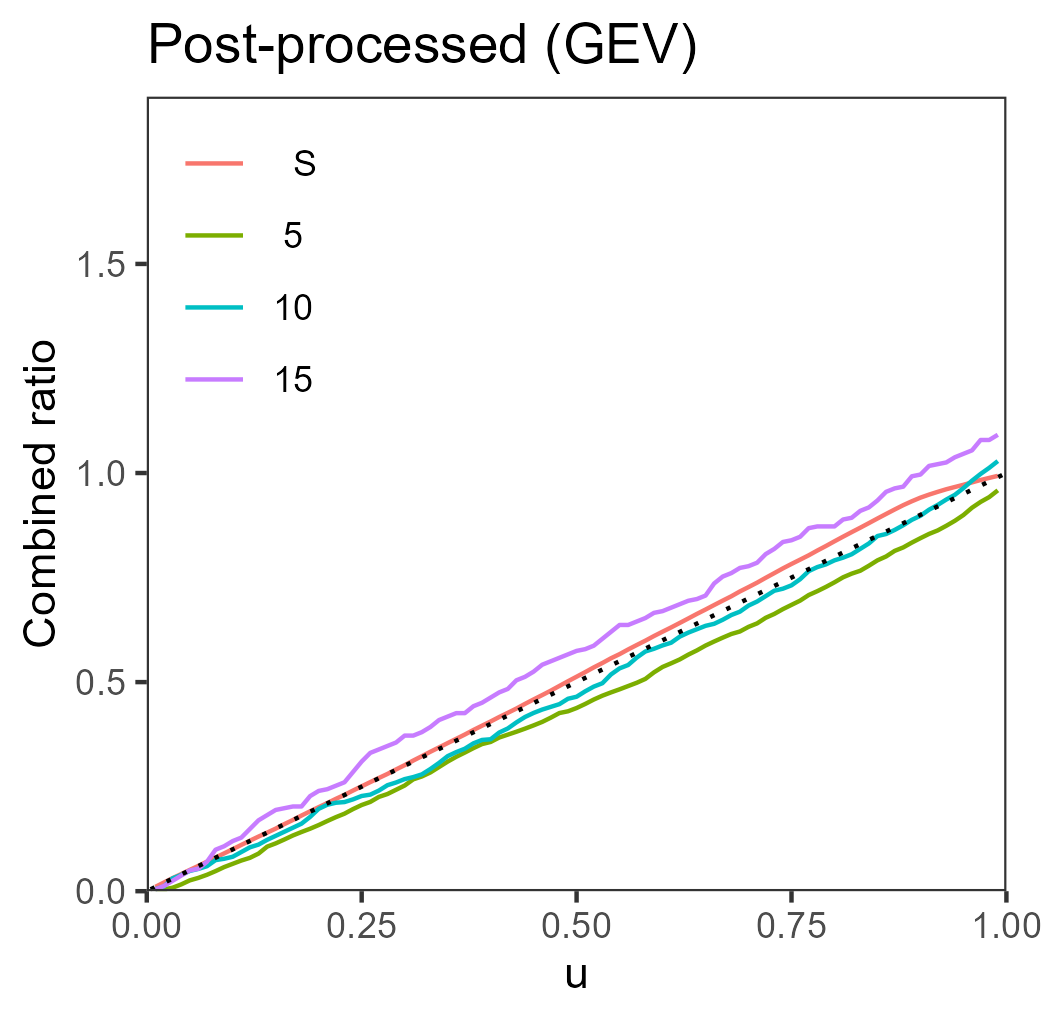}
    \includegraphics[width=0.3\textwidth]{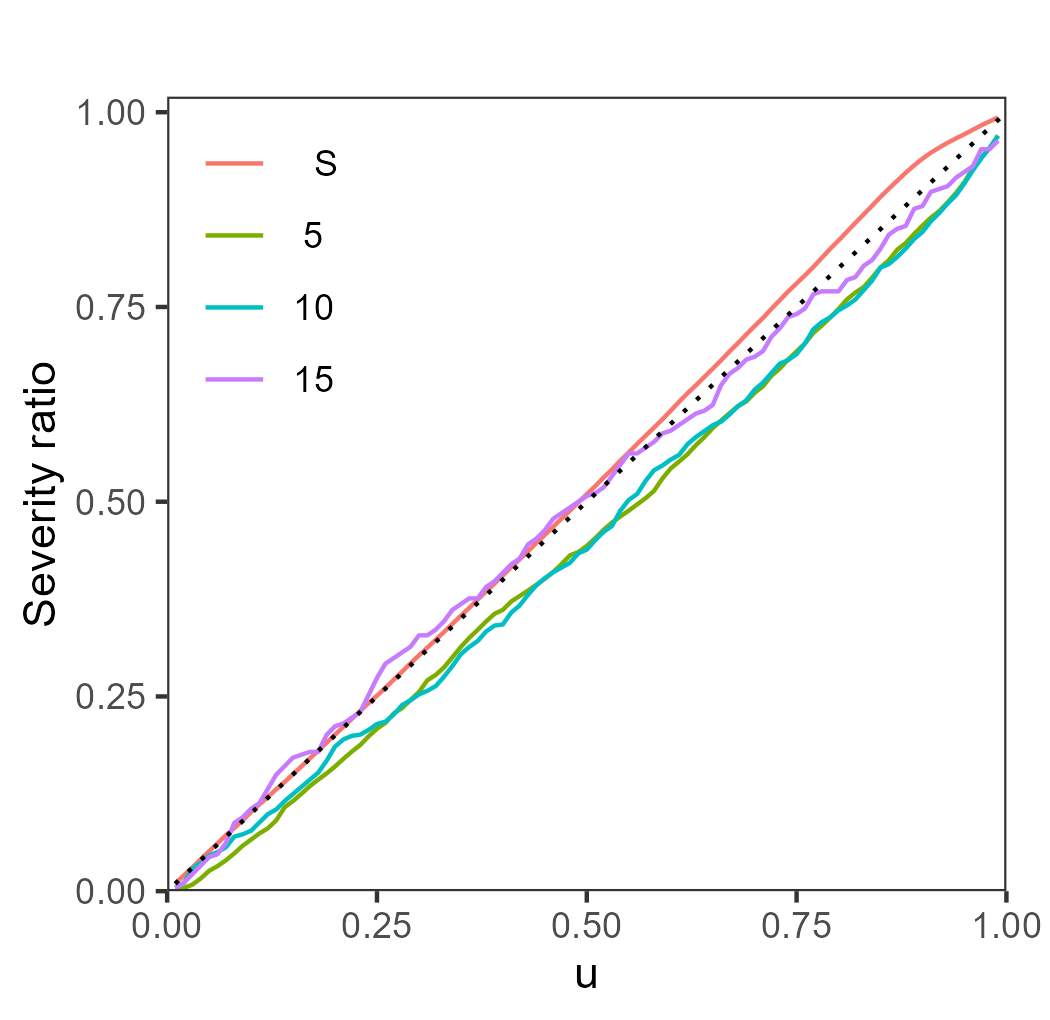}
    \includegraphics[width=0.3\textwidth]{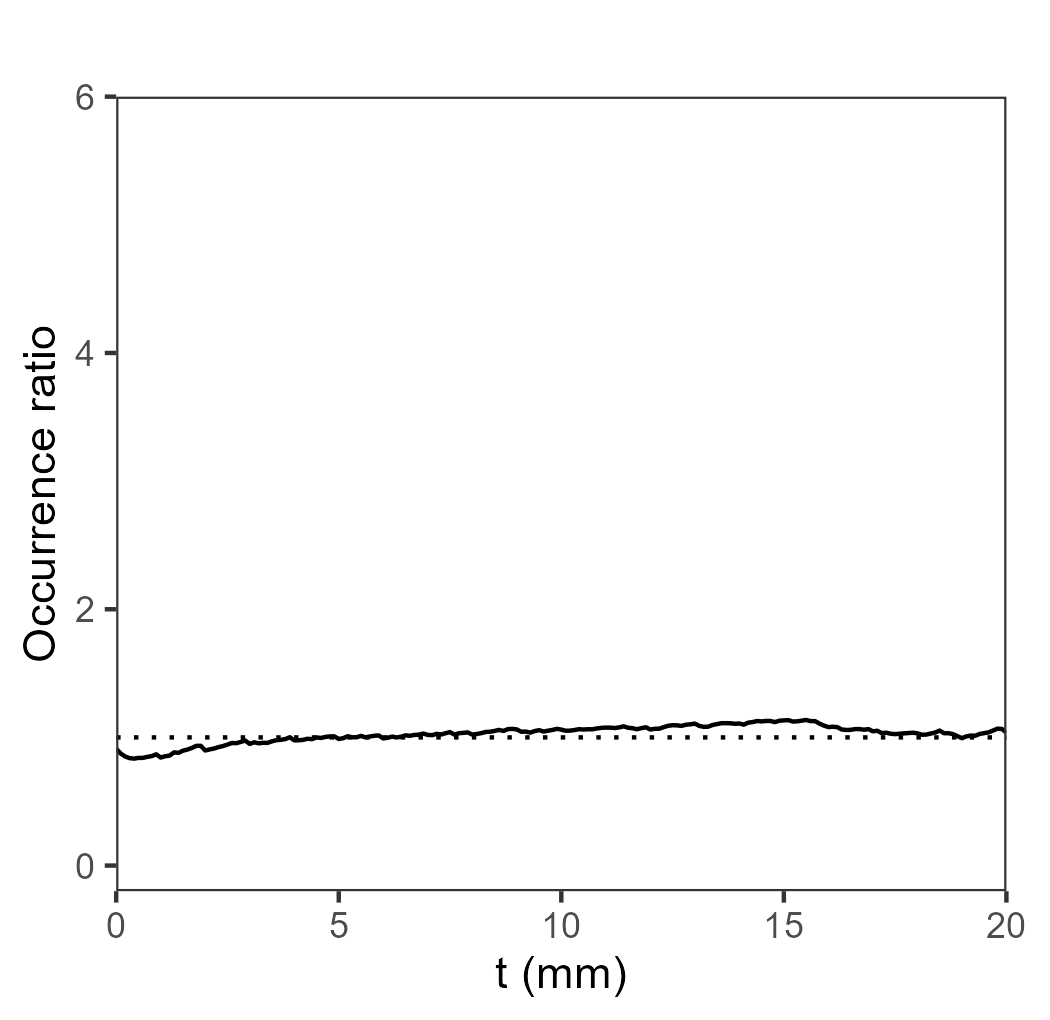}
    \caption{Probabilistic tail calibration plots for the IFS and statistically post-processed forecast distributions in Section~\ref{sec:casestudy}. Left: combined ratio \eqref{eq:emp_ratio}; middle: \emph{pp}-plot of $z_i^t$ for $i \in \mathcal{I}_t$; right: occurrence ratio \eqref{eq:ocratio} as a function of $t$. Results are shown for thresholds 5, 10, and \qty{15}{mm}, roughly corresponding to the 97$^{th}$, 99$^{th}$, and 99.5$^{th}$ percentiles of the observed precipitation measurements. The symbol `S' denotes the standard notion of probabilistic calibration ($t = -\infty$).}
    \label{fig:cs_pit}
\end{figure}

\section{Discussion}\label{sec:conclusion}

\cite{BrehmerStrokorb2019} prove that for any proper scoring rule, it is possible to find a forecast distribution that exhibits incorrect tail behavior, but receives an expected score that is arbitrarily close to that of the true distribution. We demonstrate that this negative result cannot be circumvented by evaluating forecasts using multiple scoring rules. As an alternative, we introduce a general framework of forecast tail calibration that allows the reliability of probabilistic forecasts to be assessed when interest is in extreme outcomes.

Tail calibration is generally not implied by standard calibration, and therefore provides additional information that is not available from existing checks for forecast calibration. This is illustrated in a case study on European precipitation forecasts, whereby statistical post-processing methods are found to produce calibrated forecasts despite issuing unreliable predictions for extreme outcomes. By introducing a method to evaluate forecasts for extreme outcomes, it becomes straightforward to identify deficiencies in predictions of these events, thereby facilitating the development of forecasting methods that can more reliably predict extreme outcomes. While we do not consider forecast re-calibration methods in this paper, tail re-calibration methods could be designed that modify forecast distributions so that they exhibit calibrated tails. For example, conformal prediction allows us to generate probabilistic forecasts with calibration guarantees \citep[see e.g.][]{VovkEtAl2022}, and there may be potential for such methods to be applied to extreme observations to construct forecasts with tail calibration guarantees.

Tail calibration can be assessed with respect to different information sets, though we imagine the unconditional case will be of most interest in practical applications. To assess probabilistic tail calibration, we introduce graphical diagnostic tools that not only demonstrate whether or not a forecast is tail calibrated, but also elucidate what types of biases occur in the tail; for example, whether the predictive distributions are too heavy- or light-tailed. Formal testing procedures for tail calibration are an interesting avenue for future work, where the crux will be to choose a data-dependent threshold that is representative of the asymptotic regime. Such tests could either be classical tests for fixed pre-specified sample sizes, or monitoring procedures as introduced by \cite{ArnoldHenziETAL2023}, who propose a sequential procedure for testing probabilistic calibration based on $e$-processes, recognizing that forecasting is often an inherently sequential task. 

Tail calibration may partially address the open question how best to compare probabilistic forecasts for extreme outcomes. While calibration is principally an absolute facet of probabilistic forecast performance, tail $\Bb$-calibration can be used for forecast comparison by finding the largest information set $\Bb$ on which forecasts are tail calibrated, with larger sets being better. Tail calibration with respect to a larger information set corresponds to a more informed forecast for extremes, which aligns with the idea that forecasts should be as sharp, or informative, as possible, subject to being calibrated \citep{GneitingEtAl2007}. \cite{zamo2021sequential} similarly suggest using this calibration-sharpness principle to select the most informative forecast among a set of reliable candidates when aggregating probabilistic forecasts, and such an approach could also be applied when forecasting extreme events. However, for miscalibrated forecasts and non-nested maximal information sets, this framework still does not allow for comparison between different forecasts.

When the unconditional distribution of the outcome variable is in the max-domain of attraction of a non-degenerate distribution (Assumption~\ref{ass:Ygpd}), probabilistic tail calibration is equivalent to correctly forecasting the shape and scale parameters of the GPD. Checks for probabilistic tail calibration therefore serve as a more general framework with which to validate tail models, since they are also applicable to distributions not satisfying Assumption~\ref{ass:Ygpd}. Future work may investigate whether notions of tail calibration can conversely be used to estimate quantities relevant for extreme value analysis, such as tail indices.

Finally, while we focus here on univariate extremes, future work could consider an extension to the multivariate case. In recent reviews of probabilistic forecasting, \citet{GneitingKatzfuss2014}, \citet{KneibEtAl2023}, and \citet{PetropoulosEtAl2022} all remark that developing more theoretically-principled techniques to generate and evaluate probabilistic forecasts for multivariate outcomes is one of the most pressing issues in the field. 
General notions of multivariate probabilistic calibration have been proposed \citep[e.g.][]{GneitingEtAl2008,AllenEtAl2023c}, but these typically project the multivariate forecasts and observations onto the univariate domain, and extremes in the transformed space may not be of practical interest. Instead, a more relevant definition of multivariate tail calibration may be obtained by leveraging results in multivariate extreme value theory.

\section*{Acknowledgments}

We would like to thank two anonymous referees for their constructive feedback. Johan Segers gratefully acknowledges helpful questions and suggestions from participants, in particular Anja Jan\ss{}en and Cl\'{e}ment Dombry, of the EXSTA Kick-Off Workshop at Universit\'e Paris Cité, Laboratoire MAP5, June 5--7, 2024. 

\bibliographystyle{chicago}

\bibliography{biblio.bib} 

\newpage 
\appendix

\section{Omitted proofs and calculations}
\label{app:proofs}

\begin{proof}[Proof that the misinformed forecaster in Example~\ref{ex:mis} does not satisfy \eqref{eq:tailBC:sev}.]
	On the event $Y > t$, the excess PIT of $Y-t$ by $F_t$ is
	\[
	\ZFt 
	= \frac{F(Y)-F(t)}{1-F(t)} 
	= \frac{e^{-\Delta_2 t} - e^{-\Delta_2 Y}}{e^{-\Delta_2 t}}
	= 1 - e^{-\Delta_2 (Y-t)}
	\]
	so that
	\begin{align*}
		\Q \bigl( \ZFt \le u, Y > t \bigr)
		&= \Q \bigl( t < Y \le t - \Delta_2^{-1} \log(1-u) \bigr) \\
		&= \E \bigl[ e^{-\Delta_1 t} - e^{-\Delta_1 (t - \Delta_2^{-1} \log(1-u))} \bigr] \\
		&= (1+\gamma t)^{-1/\gamma} - 
		\E \bigl[ \bigl\{ 1 + \gamma (t - \Delta_2^{-1}\log(1-u)) \bigr\}^{-1/\gamma} \bigr].
	\end{align*}
	For $u \in (0, 1)$, we thus find, by the dominated convergence theorem, as $t \to \infty$,
	\begin{align*}
		\frac{\Q \bigl( \ZFt \le u, Y > t \bigr)}{\Q( Y > t)}
		&= 1 - \E \left[
		\left(
		\frac{1 + \gamma (t - \Delta_2^{-1} \log(1-u))}{1 + \gamma t}
		\right)^{-1/\gamma}
		\right]
		\to 1-1=0,
	\end{align*}
	violating \eqref{eq:tailBC:sev}.
\end{proof}

\begin{proof}[Proof of Proposition~\ref{propo:tailequiv}]
	We check the convergence relations~\eqref{eq:tailBC:occ} and~\eqref{eq:tailBC:sev}. The first one is true by assumption. To verify the second one, let $u \in (0, 1)$ [for $u = 1$, condition~\eqref{eq:tailBC:sev} is trivial] and write $\Fbar(t) = 1-F(t)$, $\GB(t) = \Q(Y \le t \mid \Bb)$ and $\GBbar(t) = 1 - \GB(t)$. Eq.~\eqref{eq:tailBC3} implies that $x_F = x_Y$.
	Here, we use that $\Q(Y = x_Y) = 0$ and therefore also $\Q(Y = x_Y \mid \Bb) = 0$ almost surely; see Lemma~\ref{lem:endpoint}.
	For all $t \in \R$ such that $\Fbar(t) > 0$, we have
	\begin{align*}
		\ZFt < u
		&\iff \frac{F(Y) - F(t)}{\Fbar(t)} < u \\
		&\iff F(Y) < F(t) + u \, \Fbar(t) \\
		&\iff Y < F^{-1} \left( F(t) + u \, \Fbar(t) \right)
	\end{align*}
	where $F^{-1}(v) = \inf \{ y \in \R : F(y) \ge v \}$ for $v \in (0, 1]$. It follows that, for every $\eps \in (0, \min(u,1-u))$,
	\begin{align*}
		\Q \bigl( \ZFt \le u, \, Y > t \mid \Bb \bigr)
		&\le \Q \bigl( \ZFt < u+\eps, \, Y > t \mid \Bb \bigr) \\
		&= \Q \bigl( Y < F^{-1} \left( F(t) + (u+\eps) \, \Fbar(t) \right), \, Y > t \mid \Bb \bigr) \\
		&\le \Q \bigl( Y \le F^{-1} \left( F(t) + (u+\eps) \, \Fbar(t) \right), \, Y > t \mid \Bb \bigr) \\
		&= \GB \bigl( F^{-1} \left( F(t) + (u+\eps) \, \Fbar(t) \right) \bigr) - \GB(t).
	\end{align*}
	In the same way,
	\begin{align*}
		\Q \bigl( \ZFt \le u, \, Y > t \mid \Bb \bigr)
		&\ge \Q \bigl( \ZFt < u, \, Y > t \mid \Bb \bigr) \\
		&= \Q \bigl( Y < F^{-1} \left( F(t) + u \, \Fbar(t) \right), Y > t \mid \Bb \bigr) \\
		&\ge \Q \bigl( Y \le F^{-1} \left( F(t) + (u-\eps) \, \Fbar(t) \right) - \eps, \, Y > t \mid \Bb \bigr) \\
		&= \GB \bigl( F^{-1} \left( F(t) + (u-\eps) \, \Fbar(t) \right) \bigr) - \GB(t),
	\end{align*}
	where we used that $F$ is continuous by assumption and thus $F^{-1}$ is strictly increasing. Writing $y_{\pm}(u,t) = F^{-1}\left(F(t) + (u\pm \eps) \Fbar(t)\right)$, we have
	\begin{align*}
		\frac{\GB(y_\pm(u,t)) - \GB(t)}{\GBbar(t)}
		&= 1 - \frac{\GBbar(y_\pm(u,t))}{\GBbar(t)} \\
		&= 1 - \frac{\Fbar(y_\pm(u,t))}{\Fbar(t)} \cdot \frac{\GBbar(y_\pm(u,t))/\Fbar(y_\pm(u,t))}{\GBbar(t)/\Fbar(t)}.
	\end{align*}
	As $F(t) + (u \pm \eps) \Fbar(t)$ is less than $1$ and converges to $1$ as $t \to x_F$, we have $y_\pm(u,t) \to x_F$ as $t \to x_F$. Further, since $F$ is continuous, we have $F(F^{-1}(v)) = v$ for all $v \in (0, 1]$, and thus
	\[
	\Fbar(y_{\pm}(u,t)) 
	= 1 - \left(F(t) + (u\pm\eps) \Fbar(t)\right)
	= (1 - (u \pm \eps)) \Fbar(t).
	\]
	We find that
	\[
	\frac{\GB(y_\pm(u,t)) - \GB(t)}{\GBbar(t)}
	\to 1 - (1 - (u \pm \eps)) \cdot \frac{1}{1} 
	= u \pm \eps,
	\qquad \text{almost surely as $t \to x_Y$.}
	\]
	In view of the bounds earlier in the proof, we obtain
	\[
	u - \eps
	\le \liminf_{t \to x_Y}
	\frac{\Q\bigl( \ZFt \le u, \, Y > t \mid \Bb \bigr)}{\Q(Y > t \mid \Bb)}
	\le \limsup_{t \to x_Y}
	\frac{\Q\bigl( \ZFt \le u, \, Y > t \mid \Bb \bigr)}{\Q(Y > t \mid \Bb)}
	\le u+\eps
	\]
	almost surely. Setting $\eps = 1/k$ for sufficiently large integer $k$ and taking the intersection over all such $k$, we conclude that \eqref{eq:tailBC:sev} is true, as required.
\end{proof}

\begin{proof}[Proof of Lemma~\ref{lem:endpoint}]
	If $x_Y = +\infty$ there is nothing to show, so assume $x_Y \in \R$. We proceed by contraposition. Suppose $\Q(Y = x_Y) > 0$. 
	We will show that \eqref{eq:tailBC:sev} fails.

	Write $t_n = x_Y - 1/n$.
	We have $\1\{Y > t_n\} \to \1\{Y \ge x_Y\} = \1\{Y = x_Y\}$. By the conditional dominated convergence theorem \citep[e.g.][Exercise~8 on page~117]{kallenberg2002foundations}, $\Q(Y > t_n \mid \Bb) \to \Q(Y = x_Y \mid \Bb)$ as $n \to \infty$ a.s.
	Since $\E[\Q(Y = x_Y \mid \Bb)] = \Q(Y = x_Y) > 0$, we have
	$\Q(B) > 0$ where $B = \{ \Q(Y = x_Y \mid \Bb) > 0 \}$.

	Since $B \in \Bb$, we have
	\begin{align}
	\label{eq:PBb+}
	\E \left[ \Q \left( F(x_Y) < 1, Y = x_Y \mid \Bb \right) \1_B \right]
	&= \Q \left( \left\{ F(x_Y) < 1, Y = x_Y \right\} \cap B \right).
	\end{align}
	We consider two cases, whether $\Q \left( \left\{ F(x_Y) < 1, Y = x_Y \right\} \cap B \right)$ is positive or zero.
	
	\emph{1. Case $\Q \left( \left\{ F(x_Y) < 1, Y = x_Y \right\} \cap B \right) > 0$.}
	Eq.~\eqref{eq:PBb+} implies that, on $B$, with positive probability, $\Q \left( F(x_Y) < 1, Y = x_Y \mid \Bb \right) > 0$.
	Let $u \in (0, 1]$; note that we assume $u > 0$.
	On the event $\{F(t) < 1\}$, we have $\ZFt = (F(Y)-F(t))/(1-F(t))$ and thus
	\begin{align*}
		\ZFt \le u
		&\iff
		F(Y) \le (1-u) F(t) + u.
	\end{align*}
	If $F(x_Y) < 1$, then $F(x_Y) < (1-u) F(x_Y) + u$, since $u > 0$. Further, $Y \le x_Y$ almost surely and thus $F(Y) \le F(x_Y)$ almost surely. This implies that, on the event $\{F(x_Y) < 1\}$, we have $F(Y) < (1-u) F(x_Y) + u$ a.s. Since $F$ is continuous a.s., the inequality $F(Y) \le (1-u) F(t_n) + u$ is thus eventually true on the event $\{F(x_Y) < 1\}$, a.s.
	In view of the conditional dominated convergence theorem, we have
	\begin{align*}
		\Q \left( \ZFtn \le u, Y > t_n \mid \Bb \right)
		&\ge \Q \left( 
		\ZFtn \le u, F(x_Y) < 1, Y = x_Y \mid \Bb 
		\right) \\
		&= \Q \left( F(Y) \le (1-u) F(t_n) + u, \, F(x_Y) < 1, \, Y = x_Y \mid \Bb \right) \\
		&\to \Q \left( 
		F(x_Y) < 1, Y = x_Y \mid \Bb 
		\right), \qquad n \to \infty, \text{ a.s.}
	\end{align*}
	As explained at the beginning of this case, the limit is positive on $B$ with positive probability. Since the limit does not depend on $u$, Eq.~\eqref{eq:tailBC:sev} fails on an event with positive probability for $u > 0$ close to $0$.

	\emph{2. Case $\Q \left( \left\{ F(x_Y) < 1, Y = x_Y \right\} \cap B \right) = 0$.}
Eq.~\eqref{eq:PBb+} then implies
\[
	\E \left[ \Q \left(F(x_Y) < 1, Y = x_Y \mid \Bb\right) \1_B \right] = 0,
\]
so that $\Q(F(x_Y) < 1, Y = x_Y \mid \Bb) = 0$ and thus $\Q(F(x_Y) = 1, Y = x_Y \mid \Bb) = 1$ on $B$.
Let $u \in [0, 1)$; note that $u < 1$. By definition, $F_t(x) = 1$ for $x \ge 0$ in case $F(t) = 1$. The inequality $\ZFt \le u$ can thus hold only if $F(t) < 1$. 
On $B$, as $\{t_n < Y < x_Y\} \downarrow \varnothing$, we have
\begin{align*}
	\Q \left( \ZFtn \le u, \, Y > t_n \mid \Bb \right)
	&= \Q \left( \ZFtn \le u, \, F(t_n) < 1, \, Y > t_n \mid \Bb \right) \\
	&= \Q \left( F(Y) \le (1-u) F(t_n) + u, \, F(t_n) < 1, \, Y > t_n \mid \Bb \right) \\
	&= \Q \left( F(x_Y) \le (1-u) F(t_n) + u, \, F(t_n) < 1, \, Y = x_Y \mid \Bb \right) + o(1) \\
	&= \Q \left( 1 \le (1-u) F(t_n) + u, \, F(t_n) < 1, \, Y = x_Y \mid \Bb \right) + o(1) \\
	&= o(1),
\end{align*}
as the two requirements on $F(t_n)$ are contradictory. We find that on $B$, which has positive probability, $\Q \left( \ZFtn \le u, \, Y > t_n \mid \Bb \right)$ converges to zero for all $u < 1$. But then \eqref{eq:tailBC:sev} fails. 
\end{proof}

\begin{proof}[Details for Example~\ref{ex:12}]
    Let $t \ge a_+$, $x \ge 0$, and $u \in [0,1]$. Then,
    \begin{equation}\label{eq:ex:6_1}
    	1 - F(t) 
    	= \Q(\xi + a_{\tau} > t \mid \tau)
    	= e^{-(t-a_{\tau})}
    	= \frac{2 + \tau}{2} e^{-t},
    \end{equation}
    and thus $F_t(x) = (F(t+x)-F(t))/(1-F(t)) = 1-e^{-x}$,
    independently of $\tau$, since the factors $(2+\tau)/2$ in the numerator and denominator cancel out. Furthermore,
    \begin{align*}
    	\Q \bigl( \ZFt \le u, \, Y > t \bigr)
    	= \Q \bigl( 1 - e^{-(Y-t)} \le u, \, Y > t \bigr) 
    	= \Q \bigl( t < Y \le t-\log(1-u) \bigr) 
        = e^{-t} u.
    \end{align*}
    By taking expectations in \eqref{eq:ex:6_1}, we obtain $\E[1 - F(y)]  = e^{-y}$, since $\E[\tau] = 0$.
    It follows that $F$ is probabilistically tail calibrated, since
    $\Q\bigl(\ZFt \le u, \, Y > t\bigr)/\E[1-F(t)] = e^{-t}u/e^{-t} = u$.
    However, $F$ is not tail auto-calibrated, since $\Q(Y > t \mid \tau) = e^{-t}$, so that
    \[
    	\frac{\Q(Y > t \mid \tau)}{1 - F(t)}
    	= \frac{2}{2+\tau}
    	= \begin{dcases}
    		2/3 & \text{if $\tau = 1$,} \\
    		2 & \text{if $\tau = -1$,}
    	\end{dcases}
    \]
    which does not converge almost surely to $1$ as $t \to \infty$.
\end{proof}

The following lemma shows that the convergence in \eqref{eq:tailBCt} in Definition \ref{def:TAC_new} is uniform in $u$. This motivates the diagnostic plots in Section \ref{sec:diagnostics}.

\begin{lemma}
\label{lem:unifconv-u}
	The convergence relation~\eqref{eq:tailBCt} holds pointwise for all $u$ in a dense subset of $[0, 1]$ that includes $1$ if and only if it holds uniformly in $u \in [0, 1]$, i.e.,
	\[
		\sup_{u \in [0, 1]} \left|
			\frac{\Q\bigl(\ZFt \le u, Y > t \mid \Bb\bigr)}{\E[1-F(t)\mid \Bb]} - u
		\right|
		\to 0, \qquad \text{almost surely as $t \to x_Y$.}
	\]
    An analogous result holds for the convergence relation~\eqref{eq:tailBC:sev}.
\end{lemma}

\begin{proof}
	We show the claim for \eqref{eq:tailBCt}. The arguments for \eqref{eq:tailBC:sev} are identical except for the obvious modification of the definition of $R_t(u)$ below.

	Suppose \eqref{eq:tailBCt} holds for $u \in D$, where $D$ is a dense subset of $[0, 1]$ that includes $1$. Fix $\eps > 0$ and choose points $u_1,\ldots,u_k \in D$ such that $u_1 < \ldots < u_k = 1$ and moreover $u_1 \le \eps$ and $u_\ell - u_{\ell-1} \le \eps$ for all $\ell = 2, \ldots, k$. For $u \in [0, 1]$, write
	\[
		R_t(u) 
		= \frac{\Q\bigl(\ZFt \le u, Y > t \mid \Bb\bigr)}{\E[1-F(t)\mid \Bb]}.
	\]
	Note that $R_t(u)$ is monotone in $u \in [0, 1]$. For $u \in [0, 1]$, let $\ell(u) = \min \{ \ell = 1,\ldots,k : u \le u_\ell \}$. Then $u_{\ell(u)-1} \le u \le u_{\ell(u)}$ (with $u_0 = 0$) and thus
	\begin{align*}
		R_t(u) - u 
		&\le R_t(u_{\ell(u)}) - u_{\ell(u)-1}
		\le R_t(u_{\ell(u)}) - u_{\ell(u)} + \eps,\\
		R_t(u) - u
		&\ge R_t(u_{\ell(u)-1}) - u_{\ell(u)}
		\ge R_t(u_{\ell(u)-1}) - u_{\ell(u)-1} - \eps.
	\end{align*}
	It follows that
		$\sup_{u \in [0, 1]}
		\left| R_t(u) - u \right|
		\le \max_{\ell=1,\ldots,k}
		\left| R_t(u_\ell) - u_\ell \right| + \eps$.
	By assumption, $R_t(u_\ell) \to u_\ell$ almost surely as $t \to x_Y$.
	As a consequence, almost surely, 
 \[\limsup_{t \to x_Y} \sup_{u \in [0, 1]}
		\left| R_t(u) - u \right| 
		\le \limsup_{t \to x_Y} \max_{\ell=1,\ldots,k}
		\left| R_t(u_\ell) - u_\ell \right|
		+ \eps
		= \eps.
  \]
	Since $\eps > 0$ was arbitrary, we conclude $\limsup_{t \to x_Y} \sup_{u \in [0, 1]}
		\left| R_t(u) - u \right|
		= 0$, almost surely.
\end{proof}

\begin{proof}[Details for Example~\ref{ex:unfoc}]
Consider a generalization of the unfocused forecaster in \cite{GneitingEtAl2007}. Let $G$ be a distribution function supported on a (possibly unbounded) interval $[a,b]$ with continuous density $g > 0$ on $(a,b)$, let $Y \sim G$, let $\tau$ be independent of $Y$ and equal to $+1$ or $-1$ with probability $1/2$ and, given $\tau$, let $F(y) = \frac{1}{2} \{G(y)+G(y+\tau)\}$ for all real $y$. Then $Z_{F} = F(Y)$ is uniformly distributed on $[0, 1]$. Indeed, writing $G_\pm(y) = \frac{1}{2} \{G(y)+G(y\pm 1)\}$ and $z_\pm = G_\pm^{-1}(u)$, we have, for $u \in (0, 1)$,
\begin{align*}
    \frac{d}{du} \Q(Z_{F} \le u)
    &= \frac{d}{du} \frac{1}{2} \{ \Q(G_+(Y) \le u) + \Q(G_-(Y) \le u) \} \\
    &= \frac{d}{du} \frac{1}{2} \{ G(G_+^{-1}(u)) + G(G_-^{-1}(u)) \} \\
    &= \frac{1}{2}
    \left[ \frac{g(z_+)}{\frac{1}{2}\{g(z_+)+g(z_++1)\}}
    + \frac{g(z_-)}{\frac{1}{2}\{g(z_-)+g(z_--1)\}}
    \right].
\end{align*}
But since
\begin{align*}
    u &= G_+(z_+) = \frac{1}{2} (G(z_+) + G(z_++1)), \\
    u &= G_-(z_-) = \frac{1}{2} (G(z_-) + G(z_--1)) = \frac{1}{2} \{G(z_--1) + G((z_--1)+1)\},
\end{align*}
it follows that $z_--1 = z_+$, since $z \mapsto G(z) + G(z + 1)$, $z \in [a-1,b]$ is strictly increasing, and $z_+, z_- - 1 \in [a-1,b]$. Inserting this equality above yields $(d/du)\Q(Z_{F} \le u) = 1$, as required. 

The special case where $Y\sim \Unif(0, 1)$ yields, for $y \in [0, 1]$, 
\[
    F(y) = 
    \begin{dcases} 
        y/2 & \text{with probability $1/2$ (if $\tau = -1$),} \\
        (y+1)/2 & \text{with probability $1/2$ (if $\tau = +1$),}
    \end{dcases}
\]
and thus $Z_{F} = Y/2$ or $Z_{F} = (Y+1)/2$, both with probability $1/2$, so that $Z_{F}$ is indeed uniformly distributed on $(0, 1)$, too, and thus $F$ is probabilistically calibrated. In this case, the forecast $F = \frac{1}{2}(G + G(\,\cdot\,+\tau))$ is either the uniform distribution on $[-1,1]$ or the uniform distribution on $[0, 2]$. However, $F$ is not probabilistically tail calibrated. To see this, note that for $0 < t < y < 1$, we have
\[
    F_t(y-t) = \frac{F(y)-F(t)}{1-F(t)}
    = \begin{dcases}
        \frac{y/2-t/2}{1-t/2} = \frac{y-t}{2-t}
        = \frac{1-t}{2-t} \cdot \frac{y-t}{1-t}
        & \text{with probability $1/2$,} \\
        \frac{(y+1)/2 - (t+1)/2}{1-(t+1)/2} =
        \frac{y-t}{1-t} & \text{with probability $1/2$.}
    \end{dcases}
\]
For $0 < t < 1$, the distribution of $(Y - t)/(1-t) \mid Y > t$ is uniform on $[0, 1]$.
Recalling from~\eqref{eq:exPIT} the excess PIT $\ZFt = F_t(Y-t)$ given $Y > t$, we get, for $u \in (0,1)$,
\[
    \lim_{t \uparrow 1} \Q\bigl(\ZFt \le u \mid Y > t\bigr)
    = \frac{1}{2} + \frac{1}{2}u.
\] 
As a consequence, $F$ is not probabilistically tail calibrated.
\end{proof}

\begin{example}[More general version of Example \ref{ex:optimistic}]\label{ex:14b}

Let $G_1$ and $G_2$ be two random cdfs such that $G_1$ is stochastically smaller than $G_2$ almost surely, and let $\lambda$ be a deterministic cdf. We will specify relevant conditions on $G_1$, $G_2$, and $\lambda$ below, but for a concrete example, suppose that $\Delta \sim \Gamma(1/\gamma,1/\gamma)$ for some $\gamma > 0$ and define $G_1 = \Exp(\Delta)$, $G_2 = \Exp(\Delta/2)$, that is, $G_1(x) = 1 - e^{-\Delta x}$ for $x \ge 0$. Then $\E[G_1] = \GPD_{1,\gamma}$, that is, $\E[G_1(x)] = 1 - \E[e^{-\Delta x}] = 1 - (1+\gamma x)^{-1/\gamma}$ for $x \ge 0$, and we take $\lambda = \GPD_{1,\gamma/2}$.

The conditional distribution of $Y$ given $G_1, G_2$ is defined as
\begin{equation}
\label{eq:lamG1G2}
	G(y) = \lambda(y) G_1(y) + \left(1-\lambda(y)\right) G_2(y), 
	\quad y \in \R,
\end{equation}
and the random forecast is $F = G_1$. Note that $G$ is automatically increasing due to the stochastic ordering of $G_1$ and $G_2$. A stochastic construction yielding the conditional distribution at \eqref{eq:lamG1G2} goes as follows. Let $(U,L)$ be independent random variables, independent also of $(G_1, G_2)$, with $U$ uniform on $(0, 1)$ and $L$ with distribution function $\lambda$. Put $X_j = G_j^{-1}(U)$ for $j \in \{1, 2\}$, with $G_j^{-1}$ the quantile function of $G_j$. Conditionally on $(G_1, G_2)$, the random variables $X_1$ and $X_2$ have distribution functions $G_1$ and $G_2$ respectively, and $X_1 \le X_2$ almost surely since $G_1^{-1}(u) \le G_2^{-1}(u)$ for all $u \in (0, 1)$, as $G_1$ is stochastically smaller than $G_2$ almost surely. Define
\begin{equation}
\label{eq:YLX1X2}
	Y = \text{second largest of $(X_1, X_2, L)$}.
\end{equation}
For $y \in \R$, we have $Y \le y$ if and only if at least two out of three variables in $(X_1, X_2, L)$ are bounded by $y$. This yields the following event decomposition:
\[
	\{Y \le y, X_1 \le X_2\} 
	= \{L \le y, X_1 \le y, X_1 \le X_2\} 
	\cup \{L > y, X_1 \le X_2 \le y\}.
\]
Indeed, provided $X_1 \le X_2$, the event $Y \le y$ occurs in two cases: either $L \le y$, and then it is necessary and sufficient that $X_1 \le y$; or $L > y$, and then it is necessary and sufficient that $X_2 \le y$, forcing $X_1 \le y$. Since $\Q(X_1 \le X_2 \mid G_1, G_2) = 1$ almost surely, we obtain \eqref{eq:lamG1G2}:
\begin{align*}
	\Q(Y \le y \mid G_1, G_2)
	&= \Q(L \le y, X_1 \le y \mid G_1, G_2)
	+ \Q(L > y, X_2 \le y \mid G_1, G_2) \\
	&= \lambda(y) \, G_1(y) + \left(1-\lambda(y)\right) G_2(y) = G(y).
\end{align*}

The random forecast is $F = G_1$, the conditional distribution of $X_1$. Intuitively, since the marginal tail of $X_1$ is heavier than the one of $L$ and since $X_1 \le X_2$, the marginal tail of $Y$ is the same as the one of $X_1$, so that $F$ should be probabilistically tail calibrated. However, conditionally on $(G_1,G_2)$, the tail of $L$ is heavier than the one of $X_1$, and $F$ will not be tail auto-calibrated if the conditional tail of the minimum of $X_2$ and $L$ is not lighter than the one of $X_1$ almost surely. This intuition underlies the formal arguments in the following paragraph.

The conditional tail of $Y$ given $(G_1,G_2)$ can be computed as follows: we have $Y > y$ if and only if at least two out of three variables in $(X_1,X_2,L)$ exceed $y$. As a consequence,
\begin{align}
\nonumber
	\Q(Y > y \mid G_1, G_2)
	&= \Q \left(
		\{L > y, X_2 > y\} \cup \{L \le y, X_1 > y\} 
		\mid G_1, G_2
	\right) \\
\label{eq:QGgivenG1G2}
	&= \left(1-\lambda(y)\right) \left(1-G_2(y)\right)
	+ \lambda(y) \left(1-G_1(y)\right).
\end{align}
The random forecast $F = G_1$ is not tail auto-calibrated if, with positive probability, the conditional tail of $L \wedge X_2 := \min(L, X_2)$ is at least as heavy as the one as $X_1$, that is, if
\begin{equation}
\label{eq:minLX2heavy}
	\Q \left( 
		\limsup_{y \to \infty} 
		\frac{(1 - \lambda(y))(1-G_2(y))}{1-G_1(y)} 
		> 0 
	\right) > 0. 
\end{equation}
Indeed, if \eqref{eq:minLX2heavy} holds, then, by \eqref{eq:QGgivenG1G2}, we have
\[
	\frac{\Q(Y > y \mid G_1, G_2)}{1-F(y)}
	= \frac{(1 - \lambda(y))(1-G_2(y))}{1-G_1(y)} 
	+ \lambda(y),
\]
which does \emph{not} converge to $1$ almost surely as $y \to \infty$, since already $\lim_{y \to \infty} \lambda(y) = 1$. The concrete example given above satisfies \eqref{eq:minLX2heavy}. 

Next, we show that $F$ is probabilistically tail calibrated provided the marginal tail of $G_1$ is heavier than the one of $\lambda$. We assume
\begin{equation}
	\label{eq:lamG10}
	\lim_{y \to \infty} \frac{1-\lambda(y)}{\E[1-G_1(y)]} = 0.
\end{equation}
The excess probability integral transform $\ZFt$ of $Y - t$ given $Y > t$ by the forecast $F = G_1$ satisfies
\[
	1 - \ZFt 
	= 1 - F_t(Y - t) 
	= \frac{1 - G_1(Y)}{1-G_1(t)}.
\]
For $u \in [0, 1]$, we have by definition of $Y$ in \eqref{eq:YLX1X2}, since $X_1 \le X_2$ almost surely,
\[
	\Q( \ZFt \le u, Y > t)
	= \Q ( \ZFt \le u, L \le t < X_1) 
	+ \Q( \ZFt \le u, L > t, X_2 > t).
\]
The second term is bounded by $\Q(L > t) = 1 - \lambda(t)$, which is of smaller order than $\E[1-F(t)] = \E[1-G_1(t)]$ because of the assumption~\eqref{eq:lamG10}. On the event $\{L < t \le X_1\}$, we have $Y = X_1$ by the definition~\eqref{eq:YLX1X2}, so that, by independence of $L$ of everything else,
\[
	\Q \left( \ZFt \le u, L \le t < X_1 \right)
	= \lambda(t) \Q \left( 1 - \frac{1 - G_1(X_1)}{1-G_1(t)} \le u, X_1 > t \right)
	= \lambda(t) \, u \E[1-G_1(t)].
\]
The last identity follows from the fact that the conditional distribution of $X_1$ is $G_1$ together with the (implicit) assumption that $G_1$ is continuous almost surely. Together, we deduce that, for $u \in [0,1]$,
\[
	\lim_{t \to \infty} 
	\frac{\Q \left( \ZFt \le u, Y > t \right)}{\E[1-F(t)]} 
	=
	\lim_{t \to \infty}
	\left( 
	\lambda(t) u + \frac{\Q(\ZFt \le u, L > t, X_2 > t)}{\E[1-G_1(t)]}
	\right)
	= 
	u + 0 = u,
\]
meaning that $F$ is probabilistically tail calibrated.
\end{example}

\section{Scoring rules and extremes}\label{sec:scores}

A scoring rule is a function of the form 
\[
S: \Pp \times \R \to [-\infty, \infty],
\]
where $\Pp$ is a class of probability measures on $\R$. 
Scoring rules take a probabilistic forecast $F \in \Pp$ and an observation $y \in \R$ as inputs, and output a real (possibly infinite) value that measures the forecast's accuracy. It is assumed that $\bar{S}(F, G) = \E_{G}[S(F, Y)]$ exists for all $F,G \in \Pp$, that is, $S(F,\cdot)$ is assumed to be $\Pp$-quasiintegrable for all $F \in \Pp$. A scoring rule $S$ is called \emph{proper} (relative to $\Pp$) if
\begin{equation}\label{eq:propS}
\bar{S}(G, G) \leq \bar{S}(F, G) \quad \text{for all} \quad F, G \in \Pp.
\end{equation}
The scoring rule is strictly proper (relative to $\Pp$) if equality in \eqref{eq:propS} implies $F = G$. Hence, proper scoring rules encourage forecasters to issue what they truly believe will occur.

While proper scoring rules have become the standard approach with which to rank and compare competing forecast systems, they have undesirable properties when interest is on extreme events. \cite{TaillardatEtAl2023} demonstrate that the continuous ranked probability score (CRPS), arguably the most popular scoring rule for forecasts of real-valued outcomes, cannot distinguish between forecasts with different tail behavior. \cite{BrehmerStrokorb2019} generalize this result to all proper scoring rules. Specifically, let $S$ be a proper scoring rule relative to a convex class $\Pp$. Given $G \in \Pp$ and assuming that there exists a distribution $H \in \Pp$ with a heavier tail than $G$, then, for any $\epsilon > 0$, it is possible to construct a distribution $F \in \Pp$ that is not tail equivalent to $G$ but such that
\begin{equation*}
	\left| \bar{S}(F, G) - \bar{S}(G, G) \right| \leq \epsilon.
\end{equation*}
That is, the tail of the distribution can be modified while keeping the expected score $\epsilon$-close to its minimum. Recall that the upper endpoint $x^F$ of a distribution with cdf $F$ is defined as $x^F = \sup\{x \in \R \mid F(x) < 1\}$. For two distributions with cdfs $F$ and $G$, we say that $G$ has a \emph{heavier tail} than $F$ if $x^F < x^G$, or $x^F = x^G = x^*$ and $\lim_{x \to x^*}(1 - F(x))/(1-G(x)) = 0$. We denote this by $F <_t G$. The distributions are \emph{tail equivalent} if $x^F = x^G = x^*$ and $\lim_{x \to x^*}(1 - F(x))/(1-G(x)) \in (0,\infty)$.

That is, given a proper scoring rule $S$, it is possible to construct a forecast distribution $F$ whose expected score under $S$ is $\epsilon$-close to that of the true distribution $G$, despite $F$ and $G$ having different tail behavior. However, $F$'s expected score may not be $\epsilon$-close to $G$'s expected score under \emph{all} (or several) proper scoring rules. Hence, this result could potentially be circumvented by evaluating forecasts using multiple scoring rules, and determining whether the score difference is sufficiently close to zero for all scores. The following theorem shows that even if such a class would exist, it is ``all or nothing'', and would therefore not be interesting for forecast ranking. For the interpretation of the next theorem, it is helpful to keep in mind that for the class of kernel scores, which includes the CRPS, the divergence is symmetric, that is $\bar{S}(H,G) - \bar{S}(G,G) = \bar{S}(G,H) - \bar{S}(H,H)$, $G,H \in \Pp$ \citep[Section 5.1]{GneitingRaftery2007}.

\begin{theorem}\label{thm:19}
Let $\Ss$ be a collection of quasi-integrable, proper scoring rules on a convex set $\Pp$ of probability measures on $\R$. Let $G \in \Pp$ and assume that the set $\{ H \in \Pp : G <_t H\}$ is not empty. If 
\begin{equation}
\label{eq:infsup}
	\inf_{H \in \Pp : G <_t H }
	\sup_{S \in \Ss} 
	\left| \bar{S}(H, G) - \bar{S}(G, G) \right| 
	> 0
\end{equation}
then actually
\[
	\forall H \in \Pp, \; G <_t H :
	\sup_{S \in \Ss} \left| \bar{S}(G, H) - \bar{S}(H, H) \right| = +\infty.
\]
\end{theorem}

\begin{proof}
Let $\epsilon$ be equal to half the infimum in \eqref{eq:infsup}.
Let $H \in \Pp$ be such that $G <_t H$. Define $F_\lambda = \lambda H + (1-\lambda) G$ for $\lambda \in (0, 1)$. For all $\lambda \in (0, 1)$, we have $G <_t F_\lambda$, too. By assumption, there exists $S_\lambda \in \Ss$ such that $\left| 
		\bar{S}_\lambda(F_\lambda, G) 
		- \bar{S}_\lambda(G, G) 
	\right|
	\ge \epsilon$.
But, by the proof of Lemma~5.3 in \citet{BrehmerStrokorb2019}, we have
\[
	\left| 
		\bar{S}_\lambda(F_\lambda, G) 
		- \bar{S}_\lambda(G, G) 
	\right|
	\le \frac{\lambda}{1-\lambda}
	\left| 
		\bar{S}_\lambda(G, H) - \bar{S}_\lambda(H, H) 
	\right|,
\]
so that
\[
	\left| \bar{S}_\lambda(G, H) - \bar{S}_\lambda(H, H) \right|
	\ge \frac{1-\lambda}{\lambda} \epsilon.
\]
Since this is true for all $\lambda \in (0, 1)$, it follows that
\[
	\sup_{S \in \Ss} \left| \bar{S}(G, H) - \bar{S}(H, H) \right| = +\infty.
	\qedhere
\]
\end{proof}

Theorem \ref{thm:19} further substantiates the doubts concerning the adequacy of proper scoring rules for assessing tail behavior of probabilistic forecasts, and strengthens the results of \cite{BrehmerStrokorb2019}. A class of scoring rules would be uninformative when comparing two forecasters that do not have the same tail behavior as the true distribution. However, a key quality of proper scoring rules is that they allow different imperfect forecasts to be ranked and compared. If the assessment of tail behavior with proper scoring rules becomes a binary decision of whether or not the tail behavior is correctly captured, we see no gain in using proper scoring rules for this purpose.

The result of \citet{BrehmerStrokorb2019} also holds for max-functionals, and similarly, Theorem \ref{thm:19} also extends to this case. A functional $\mathrm{T} : \Pp \to \R$ is a \emph{max-functional} if it satisfies $\mathrm{T}(\lambda F_{1} + (1 - \lambda) F_{2}) = \max (\mathrm{T}(F_{1}), \mathrm{T}(F_{2}))$ for all $F_{1}, F_{2} \in \Pp$ and $\lambda \in (0, 1)$, with examples including the upper end point, the tail index, and the extreme value index. 

\begin{theorem}\label{thm:20}
Let $\Ss$ be a collection of quasi-integrable, proper scoring rules on a convex set $\Pp$ of probability measures on $\R$, and let $T:\Pp \to \R$ be a max-functional. Let $G \in \Pp$ and assume that the set $\{ H \in \Pp : T(G) < T(H)\}$ is not empty. If 
\begin{equation*}
	\inf_{H \in \Pp : T(G) < T(H) }
	\sup_{S \in \Ss} 
	\left| \bar{S}(H, G) - \bar{S}(G, G) \right| 
	> 0
\end{equation*}
then actually
\[
	\forall H \in \Pp, \; T(G) < T(H) :
	\sup_{S \in \Ss} \left| \bar{S}(G, H) - \bar{S}(H, H) \right| = +\infty.
\]
\end{theorem}

\begin{proof}
Let $H \in \Pp$ be such that $T(G) < T(H)$. Define $F_\lambda = \lambda H + (1-\lambda) G$ for $\lambda \in (0, 1)$. For all $\lambda \in (0, 1)$, we have $T(F_\lambda)= T(H) > T(G)$, since $T$ is a max-functional. The rest of the proof is completely analogous to that of Theorem~\ref{thm:19}.
\end{proof}

\section{Marginal tail calibration}\label{app:marginal}

\subsection{General construction of notions of (tail) calibration}
Auto-calibration is a strong requirement and generally hard to assess empirically. There are statistical tests for auto-calibration, but natural diagnostic tools that identify reasons for a lack of auto-calibration are not available. Therefore, one typically resorts to checking and improving necessary unconditional criteria for auto-calibration such as probabilistic and marginal calibration, with the former much more popular than the latter.

Both of these unconditional notions are a special case of the following approach, motivated by work by \cite{Bashaykh2022}. Choose a class of functions $(F,y) \mapsto h_\lambda(F,y)$ indexed by some parameter $\lambda$ that take as inputs a cdf $F$ and a real value $y$. The forecast $F$ is then called calibrated with respect to $(h_\lambda)_\lambda$ if 
\[
\E[h_\lambda(F,Y)] = 0 \quad \text{for all $\lambda$.}
\]
Probabilistic calibration uses $h_u(F,y) = \one\{F(y)\le u\}-u$, for $u \in [0,1]$, and marginal calibration arises with $h_z(F,y) = F(z) - \one\{y \le z\}$, for $z \in \mathbb{R}$. The following example of \citet{GneitingEtAl2007} shows that marginal calibration is also strictly weaker than auto-calibration, and neither implies nor is implied by probabilistic calibration. 

\begin{example}\label{ex:2}
    Consider the sign-reversed forecaster of \cite{GneitingEtAl2007}. Again, assume that $Y \mid \mu \sim N(\mu, 1)$, with $\mu \sim N(0, 1)$. The sign-reversed forecaster is  $F = N(-\mu, 1)$.
    Although this is clearly not a good forecast, taking the expectation over $\mu$ results in the unconditional distribution of the observation, meaning this forecast is marginally calibrated. However, reassuringly, this forecast is neither probabilistically calibrated nor auto-calibrated.
\end{example}

Alongside any unconditional notion of calibration, we can define $\Bb$-calibration with respect to $(h_\lambda)_\lambda$ as
\begin{equation}\label{eq:uncond_cal}
\E[ h_\lambda(F,Y) \mid \Bb]= 0 \quad \text{almost surely, for all $\lambda$,}
\end{equation}
for some $\sigma$-algebra $\Bb \subset \Aa$. If the class of functions $(h_\lambda)_\lambda$ is chosen such that it identifies distributions, that is, if for any deterministic cdf $G$ it holds that
\[
	\left( \forall \lambda : \E [h_\lambda(G,Y)] = 0 \right)
	\quad \text{if and only if} \quad Y \sim G,
\]
then auto-calibration implies $\Bb$-calibration with respect to $(h_\lambda)_\lambda$ for $\Bb \subseteq \sigma(F)$. In some instances, conditional notions of calibration have also been defined for families of $\sigma$-algebras $(\Bb_\lambda)_\lambda$ indexed by $\lambda$, that is, $F$ is calibrated with respect to $(h_\lambda)_\lambda$ and $(\Bb_\lambda)_\lambda$ if
\[
\E[ h_\lambda(F,Y) \mid \Bb_\lambda]= 0 \quad \text{almost surely, for all $\lambda$.}
\]
Examples include threshold-calibration with $h_z(F,y) = F(z) - \one\{y \le z\}$, $\Bb_z = \sigma(F(z))$, $z \in \R$ and quantile calibration with $h_u(F,y) = \one\{F(y)\le u\}-u$, $\Bb_u = \sigma(F^{-1}(u))$, $u \in [0,1]$. 

Starting from a notion of calibration as in \eqref{eq:uncond_cal}, a corresponding notion of tail calibration can be derived. For example, one can request that 
\[
\lim_{t \to x_Y} \sup_{\lambda}|\E[ h_\lambda (F_t,Y-t) \mid Y >t]| = 0,
\]
and $\lim_{t \to x_Y} \Q(Y>t)/\E[1-F(t)] = 1$. In the main paper, we have followed this path starting from probabilistic calibration. In the next section, we consider the corresponding notion of tail calibration that arises from marginal calibration. Lemma \ref{lem:unifconv-u} motivates the use of the supremum in the last display.

\subsection{The example of marginal tail calibration}

\begin{defin}[Marginal tail calibration]
\label{defin:mtc}
Let $x_Y = \sup \{ x \in \R : \Q(Y \le x) < 1 \}$ be the upper endpoint of $Y$. The forecast $F$ is \emph{marginally tail calibrated} for $Y$ if $\E[1-F(t)] > 0$ for all $t < x_Y$ and it holds that
\begin{equation}\label{eq:tailMC1}
        \lim_{t \to x_Y} \sup_{x \ge 0} \left|
        \frac{\Q(t < Y \le t+x)}{\E[1-F(t)]}
        -
        \E[F_t(x) \mid Y > t]
    \right|
    = 0.
\end{equation}
\end{defin}

This notion of marginal tail calibration could be further refined by defining conditional versions similarly to Definition \ref{def:TAC_new}, where this framework has been pursued starting from probabilistic calibration.

As in the definition of $\Bb$-tail calibration at \eqref{eq:tailBCt}, the single requirement \eqref{eq:tailMC1} is equivalent to two conditions together that separately concern the occurrence and severity of excesses over the threshold $t$: For all $u \in [0,1]$,
\begin{equation}\label{eq:tailMC2}
    \lim_{t \to x_Y}
    \sup_{x \ge 0} \left|
        \Q(Y-t \le x \mid Y>t)
        - \E [F_t(x) \mid Y>t]
    \right| = 0, \quad \lim_{t \to x_Y} \frac{\Q(Y > t)}{\E[1-F(t)]} = 1.
\end{equation}

Indeed, by taking the limit in \eqref{eq:tailMC1} as $x \to \infty$, the difference $|\Q(Y > t)/\E[1-F(t)]-1|$ is bounded by the supremum. This gives the second part of \eqref{eq:tailMC2}; the first part then follows too.

Akin to Proposition~\ref{propo:tailequiv}, for non-random probabilistic forecasts $F$, it is straightforward to deduce that marginal tail calibration is equivalent to marginal tail equivalence.

\begin{propo}
	\label{propo:tailequiv_mtc}
	If the continuous probabilistic forecast $F$ is non-random, then $F$ is marginally tail calibrated if and only if $x_F = x_Y$ and $\lim_{t\to x_Y} \Q(Y > t)/(1-F(t)) = 1$.
\end{propo}
    
\begin{proof} 
    Showing that the second condition in \eqref{eq:tailMC2} holds is trivial. For the first condition, write $G(t) = \Q(Y \le t)$ and note that
    \[
        G_t(x)
        = \frac{G(t+x)-G(t)}{1-G(t)}
        = 1 - \frac{1-F(t+x)}{1-F(t)} \point \frac{(1-G(t+x))/(1-F(t+x))}{(1-G(t))/(1-F(t))}.
    \]
    The second fraction is arbitrarily close to $1$ for large $t$ and $x \ge 0$ such that $t+x < x_F = x_G$, and $G_t(x)$ can be made arbitrarily close to $F_t(x) = 1 - (1-F(t+x))/(1-F(t))$.
\end{proof}

\begin{remark}
    We know that for non-random $F$, (standard) marginal and probabilistic calibration are equivalent. Propositions~\ref{propo:tailequiv} and \ref{propo:tailequiv_mtc} show that marginal and probabilistic tail calibration are also equivalent in this case, and both correspond to marginal tail equivalence.
\end{remark}

We can also derive implications for marginal tail calibration. The misinformed forecaster shows that ordinary marginal calibration does not imply marginal tail calibration.

\begin{example}[Marginal calibration does not imply marginal tail calibration]\label{ex:misinformed} 
    The \emph{misinformed} forecast in Example~\ref{ex:mis} is marginally calibrated because the unconditional distribution of $Y$ and the expectation of $F$ are both $\GPD_{1,\gamma}$. However, the forecast is not marginally tail calibrated because the first condition in \eqref{eq:tailMC2} is violated: For any $t > 0$ and $x > 0$, we have on the one hand
    \[
    	\E [F_{t}(x) \mid Y > t]
    	= \E[F_{t}(x)] 
    	= \E[1 - e^{-\Delta_2 x}]
    	= 1 - (1+\gamma x)^{-1/\gamma}, 
    \]
    independently of the value of $t$, while on the other hand
    \begin{align*}
    	\Q(Y-t\leq x \mid Y>t) 
    	&= \frac{(1+\gamma t)^{-1/\gamma}-(1+\gamma(t+x))^{-1/\gamma}}{(1+\gamma t)^{-1/\gamma}} \\
    	&= 1 - \left(1 + \gamma \frac{x}{1+\gamma t}\right)^{-1/\gamma}
    	\to 1 - 1 = 0, \qquad t \to \infty.
    \end{align*}
\end{example}

The tail unfocused forecaster in Example~\ref{ex:12} is marginally tail calibrated, which shows that marginal tail calibration (and joint probabilistic and marginal tail calibration) is a strictly weaker requirement than tail auto-calibration. A second example of this type is given by the forecaster in Example~\ref{ex:optimistic}, see Example~\ref{ex:14c} below.

\begin{example}[Example \ref{ex:12} continued]\label{ex:24}
We consider the same setting as in Example \ref{ex:12} and show that the tail unfocused forecaster $F$ is marginally tail calibrated. The second condition in \eqref{eq:tailMC2} holds since $F$ is probabilistically tail calibrated. The first condition holds since $F_t(x) = 1-e^{-x} = \Q(Y \le x+t \mid Y > t)$ for $t \ge a_+$ and $x \ge 0$. 
\end{example}

\begin{example}[Examples \ref{ex:optimistic} and \ref{ex:14b} continued]\label{ex:14c}
Consider the forecaster $F=G_1$ and outcome $Y$ as defined in Example \ref{ex:14b}. Under condition \eqref{eq:lamG10}, $F$ is marginally tail calibrated. Indeed, the second condition in \eqref{eq:tailMC2} holds since $F$ is probabilistically tail calibrated. Furthermore, for $t \in \R$ and $x \ge 0$, we have, by \eqref{eq:QGgivenG1G2} and the tower property,
\begin{align*}
	\lefteqn{\Q(Y > t+x \mid Y > t) - \E[1-F_t(x) \mid Y > t]} \\
	&= \frac{1}{\Q(Y > t)} \left\{
		\Q(Y > t+x) - 
		\E \left[ 
			\frac{1-G_1(t+x)}{1-G_1(t)} \, \1(Y > t)
		\right]
	\right\} \\
	&= \frac{\E[1-G_1(t)]}{\Q(Y > t)}
	\frac{1}{\E[1-G_1(t)]} \left\{
		\Q(Y > t+x)
		-
		\E \left[ 
		\frac{1-G_1(t+x)}{1-G_1(t)} \, \Q(Y > t \mid G_1,G_2)
		\right]
	\right\} \\
	&= \frac{\E[1-G_1(t)]}{\Q(Y > t)}
	\Bigg\{
		\frac{\Q(Y > t+x)}{\E[1-G_1(t+x)]}
		\frac{\E[1-G_1(t+x)]}{\E[1-G_1(t)]}\\
		& \qquad \qquad \qquad- \frac{1-\lambda(t)}{\E[1-G_1(t)]}
		\E \left[ (1-G_2(t)) \frac{1-G_1(t+x)}{1-G_1(t)} \right]  - \lambda(t) \frac{\E[1-G_1(t+x)]}{\E[1-G_1(t)]}
	\Bigg\}.
\end{align*}
The factor before the curly braces tends to one since the second condition in \eqref{eq:tailMC2} holds; inside the curly braces, the middle term converges to zero as $t \to \infty$ uniformly in $x \ge 0$ by \eqref{eq:lamG10}, while the first and the third term cancel out as $t \to \infty$ and uniformly in $x \ge 0$.
Indeed, the difference between the first and third term is bounded uniformly in $x \ge 0$ by
\begin{align*}
	\sup_{x \ge 0} &\left|
		\frac{\Q(Y > t+x)}{\E[1-G_1(t+x)]} \cdot
		\frac{\E[1-G_1(t+x)]}{\E[1-G_1(t)]}
		- 
		\lambda(t) \cdot \frac{\E[1-G_1(t+x)]}{\E[1-G_1(t)]}
	\right| \\
	&\le
	\sup_{x \ge 0} \left|
		\frac{\Q(Y > t+x)}{\E[1-G_1(t+x)]}
		- \lambda(t)
	\right| \cdot
	\frac{\E[1-G_1(t+x)]}{\E[1-G_1(t)]} \le
	\sup_{x \ge 0} \left|
	\frac{\Q(Y > t+x)}{\E[1-G_1(t+x)]}
	- \lambda(t)
	\right| \\
	&\le
	\sup_{x \ge 0} \left|
		\frac{\Q(Y > t+x)}{\E[1-G_1(t+x)]}
		- 1
	\right|
	+ \left| \lambda(t) - 1 \right|,
\end{align*}
and the right-hand side converges to $0$ as $t \to \infty$ by the second condition in \eqref{eq:tailMC2} and the fact that $\lambda$ is a cdf.
\end{example}

\begin{propo}
\begin{enumerate}[label=(\alph*)]
\item If $\Q(Y > t \mid F)/\Q(Y> t)$ is asymptotically bounded in $L^\infty$ for $t \to x_Y$, tail auto-calibration implies marginal tail calibration. The converse is false. 
\item Marginal tail calibration does not imply marginal calibration or vice versa.
\end{enumerate}
\end{propo}
\begin{proof}
Concerning part (a), we have
\begin{align*}
\sup_{x \ge 0} &\left| \Q(Y-t \le x \mid Y>t) - \E [F_t(x) \mid Y > t]\right| \\&\le \E \left[ \sup_{x \ge 0} \left| \frac{\Q(Y-t \le x, Y > t \mid F)}{\Q(Y > t)} -  F_t(x)\frac{\Q(Y > t \mid F)}{\Q(Y > t)}\right| \right]\\
&= \E \left[ \frac{\Q(Y > t \mid F)}{\Q(Y > t)}\sup_{x \ge 0} \left| \frac{\Q(Y-t \le x, Y > t \mid F)}{\Q(Y > t\mid F)} -  F_t(x)\right| \right]\\
    &= \E \left[ \frac{\Q(Y > t \mid F)}{\Q(Y > t)} \sup_{u \in [0,1]} \left| \frac{\Q(\ZFt \le u, Y > t \mid F)}{\Q(Y > t\mid F)} -  u 
    \right| \right].
\end{align*}
By Lemma \ref{lem:unifconv-u}, the supremum on the right-hand side converges to zero almost surely as $t \to x_Y$. Since it is bounded above by one, it also converges to zero in $L^1$. By assumption, the first factor is asymptotically bounded in $L^\infty$, which yields the first claim of part~(a).

For the converse and for part~(b), see Examples~\ref{ex:misinformed}, \ref{ex:24}, and~\ref{ex:14c}. 
\end{proof}

We conjecture that marginal and probabilistic tail calibration generally do not imply each other. However, the following result shows that some creativity will be required to construct counter-examples.

\begin{propo}[Sufficient condition for both probabilistic and marginal tail calibration]
\label{pro:suffptcmtc2}
	Suppose that Assumptions~\ref{ass:Ygpd} and~\ref{ass:Fgpd} are fulfilled and that $\lim_{t \to x_Y} \Q(Y > t) / \E[1 - F(t)] = 1$. The following conditions are equivalent:
	\begin{enumerate}[label=(\alph*)]
	\item $\xi = \eta$ and $\lim_{t \to x_Y} \sigma_F(t)/\sigma_Y(t) = 1$;
	\item $F$ is probabilistically tail calibrated for $Y$;
	\item $F$ is marginally tail calibrated for $Y$.
	\end{enumerate}
\end{propo}

\begin{proof}
By Lemma~\ref{lem:unifconv} below, the convergence statements in \ref{ass:Ygpd} and \ref{ass:Fgpd} are valid uniformly in $x \ge 0$. We show first that (a) and (b) are equivalent, and then that (a) and (c) are equivalent.

\noindent \emph{Proof that (a) and (b) are equivalent.}
For $\eps > 0$ and $t < x_Y$, define the event
\[
	A(t,\eps)
	= \left\{ \sup_{x \ge 0} \left| 1-F_t(x) - \left(1+\eta x/\sigma_F(t)\right)_+^{-1/\eta} \right| \le \eps \right\}.
\]
By Assumption~\ref{ass:Fgpd} and Lemma~\ref{lem:unifconv}, we have $\lim_{t \to x_F} \Q(A(t, \eps) \mid Y > t) = 1$. Since $\ZFt = F_t(Y-t)$, we have, on the event $A(t,\eps) \cap \{Y > t\}$, the bound
\[
	\left|
		\ZFt - \left\{1 - \left(1+\eta (Y-t)/\sigma_F(t)\right)_+^{-1/\eta}\right\}
	\right| \le \eps.
\]
We deduce that, on $A(t, \eps) \cap \{Y > t\}$, we have the implications
\begin{align*}
	1 - \left\{1+\eta (Y-t)/\sigma_F(t)\right\}_+^{-1/\eta}
	\le u - \eps
	&\implies
	\ZFt \le u \\
	&\implies 
	1 - \left\{1+\eta (Y-t)/\sigma_F(t)\right\}_+^{-1/\eta}
	\le u + \eps.
\end{align*}
For $0 < p < 1$ and $v \ge 0$, we have the equivalence
\[ 
	1 - (1+\eta v)_+^{-1/\eta} \le p
	\iff
	v \le \frac{(1-p)^{-\eta} - 1}{\eta},
\]
to be interpreted as $1 - \exp(-v) \le p \iff v \le -\log(1-p)$ in case $\eta = 0$. Since $\lim_{t \to x_Y} \Q(A(t,\eps)^c \mid Y > t) = 0$, it follows that, for $0 < u < 1$ and $0 < \eps < \min(u, 1-u)$,
\begin{multline*}
	\Q\left( Y \le t + \sigma_F(t) \frac{(1-u+\eps)^{-\eta}-1}{\eta} \mid Y > t \right)
	+ o(1)
	\le \Q\left(\ZFt \le u \mid Y > t\right) + o(1) \\
	\le \Q\left( Y \le t + \sigma_F(t) \frac{(1-u-\eps)^{-\eta}-1}{\eta} \mid Y > t \right)
	+ o(1)
\end{multline*}
as $t \to x_Y$.
Because of Assumption~\ref{ass:Ygpd} and Lemma~\ref{lem:unifconv}, we have thus, as $t \to x_Y$,
\begin{multline*}
	1 - \left(1 + \xi \frac{\sigma_F(t)}{\sigma_Y(t)} \frac{(1-u+\eps)^{-\eta}-1}{\eta} \right)_+^{-1/\xi} + o(1)
	\le \Q\left(\ZFt \le u \mid Y > t\right) + o(1) \\
	\le
	1 - \left(1 + \xi \frac{\sigma_F(t)}{\sigma_Y(t)} \frac{(1-u-\eps)^{-\eta}-1}{\eta} \right)_+^{-1/\xi} + o(1).
\end{multline*}

On the one hand, if $\xi = \eta$ and $\lim_{t \to x_Y} \sigma_F(t)/\sigma_Y(t) = 1$, then the left- and right-hand sides converge to $u-\eps$ and $u+\eps$, respectively, and since $\eps \in (0, \min(u,1-u))$ can be taken arbitrarily small, this implies $\lim_{t \to x_Y} \Q(\ZFt \le u \mid Y > t) = u$, so that $F$ is probabilistically tail calibrated.

On the other hand, suppose $F$ is probabilistically tail calibrated. Then the term in the middle of the display above converges to $u = 1 - (1+\xi \frac{(1-u)^{-\xi}-1}{\xi})_+^{-1/\xi}$. By monotonicity, it follows that
\[
	\limsup_{t \to x_Y}
	\frac{\sigma_F(t)}{\sigma_Y(t)} \frac{(1-u+\eps)^{-\eta}-1}{\eta}
	\le 
	\frac{(1-u)^{-\xi}-1}{\xi}
	\le
	\liminf_{t \to x_Y}
	\frac{\sigma_F(t)}{\sigma_Y(t)} \frac{(1-u-\eps)^{-\eta}-1}{\eta}.
\]
Since this is true for all $u \in (0, 1)$ and all $\eps \in (0, \min(u,1-u))$, we can decrease $\eps$ to zero to find that, for all $u \in (0,1)$,
\[
	\limsup_{t \to x_Y}
	\frac{\sigma_F(t)}{\sigma_Y(t)} \frac{(1-u)^{-\eta}-1}{\eta}
	\le 
	\frac{(1-u)^{-\xi}-1}{\xi}
	\le
	\liminf_{t \to x_Y}
	\frac{\sigma_F(t)}{\sigma_Y(t)} \frac{(1-u)^{-\eta}-1}{\eta},
\]
and thus, actually
\[
	\lim_{t \to x_Y}
	\frac{\sigma_F(t)}{\sigma_Y(t)} \frac{(1-u)^{-\eta}-1}{\eta}
	=
	\frac{(1-u)^{-\xi}-1}{\xi},
\]
which is in turn equivalent to
\[
	\lim_{t \to x_Y}
	\frac{\sigma_F(t)}{\sigma_Y(t)} 
	=
	\frac{\left((1-u)^{-\xi}-1\right)/\xi}{\left((1-u)^{-\eta}-1\right)/\eta}.
\]
The left-hand side does not depend on $u \in (0, 1)$.
The only way for the right-hand side to be constant in $u \in (0, 1)$, too, is when $\xi = \eta$, and then the ratio is equal to $1$.

\noindent\emph{Proof that (a) and (c) are equivalent.}
As $t \to x_Y = x_F$, we have both
\[
	\sup_{x \ge 0} \left| \Q(Y > t + x \mid Y > t) - \left(1+\xi x/\sigma_Y(t)\right)_+^{-1/\xi} \right|,
	\to 0
\]
and, for every $\eps > 0$,
\begin{align*}
	\lefteqn{
		\sup_{x \ge 0} \left| \E[1-F_t(x) \mid Y>t] - \left(1+\eta x/\sigma_F(t)\right)_+^{-1/\eta} \right| 
	} \\
	&\le
	\E\left[ \sup_{x \ge 0} \left| 1-F_t(x) - \left(1+\eta x/\sigma_F(t)\right)_+^{-1/\eta} \right| \mid Y>t \right] \\
	&\le \eps + \Q \left( \sup_{x \ge 0} \left| 1-F_t(x) - \left(1+\eta x/\sigma_F(t)\right)_+^{-1/\eta} \right| > \eps \mid Y>t \right) 
	\to \eps,
\end{align*}
since the supremum inside the expectation on the second line is bounded by one. Since $\eps > 0$ can be taken arbitrarily small, we find that marginal tail calibration is equivalent to
\[
	\lim_{t \to x_Y} \sup_{x \ge 0}
	\left|
		\left(1+\xi x/\sigma_Y(t)\right)_+^{-1/\xi}
		-
		\left(1+\eta x/\sigma_F(t)\right)_+^{-1/\eta}
	\right|
	= 0.
\]
Writing $y = x/\sigma_F(t)$ and taking the supremum over $y \ge 0$ yields the equivalent statement
\[
	\lim_{t \to x_Y} \sup_{y \ge 0}
	\left|
	\left(1+\xi \frac{\sigma_F(t)}{\sigma_Y(t)} y\right)_+^{-1/\xi}
	-
	\left(1+\eta y\right)_+^{-1/\eta}
	\right|
	= 0.
\]
This is a statement about convergence of distribution functions, which is equivalent to convergence of their quantile functions (since the quantile functions are continuous), i.e., for all $u \in (0,1)$, 
\[
	\lim_{t \to x_Y} \frac{\sigma_Y(t)}{\sigma_F(t)}
	\frac{(1-u)^{-\xi}-1}{\xi} = \frac{(1-u)^{-\eta}-1}{\eta}.
\]
The proof of the equivalence with (a) now proceeds as above.
\end{proof}

\begin{lemma}[Uniform convergence]
\label{lem:unifconv}
The convergences in Equations~\eqref{eq:Ygpd} and~\eqref{eq:Fgpd} are necessarily valid uniformly in $x \ge 0$: Equation~\eqref{eq:Ygpd} implies
\begin{equation}
	\label{eq:Ygpdunif}
	\lim_{t \to x_Y} \sup_{x \ge 0}
	\left|
		\Q \left( Y > t + \sigma_Y(t)x \mid Y > t \right)
		- (1 + \xi x)_+^{-1/\xi}
	\right|
	= 0
\end{equation}
while Equation~\eqref{eq:Fgpd} implies, for all $\eps > 0$,
\begin{equation}
	\label{eq:Fgpdunif}
	\lim_{t \to x_F} \Q \left(
		\sup_{x \ge 0}
		\left|
		1 - F_t(\sigma_F(t)x) - (1+\eta x)_+^{-1/\eta}
		\right|
		> \eps
	\, \Bigg| \,
	Y > t
	\right)
	= 0.
\end{equation}
\end{lemma}

\begin{proof}
Let us prove that \eqref{eq:Fgpd} implies \eqref{eq:Fgpdunif}; the proof that \eqref{eq:Ygpd} implies \eqref{eq:Ygpdunif} is similar but easier. Fix integer $n \ge 2/\eps$ and let $0 = x_0 < x_1 < \ldots < x_{n-1} < x_n = \infty$ be such that
\[
	\forall k \in \{0,\ldots,n\}, \qquad
	(1+\eta x_k)_+^{-1/\eta} = 1 - k/n.
\]
Apply \ref{ass:Fgpd} to each $x_k$ individually and with $\eps$ replaced by $\eps/2$: considering the event
\[
	A_k(t) = \left\{ 
		\left| 
			1 - F_t(\sigma_F(t)x_k) 
			- (1 + \eta x_k)_+^{-1/\eta} 
		\right| 
		> \eps/2
	\right\},
\]
we have
\[
	\forall k \in \{0,1,\ldots,n\}, \qquad
	\lim_{t \to x_F}
	\Q \left( A_k(t) \mid Y > t \right)
	= 0.
\]
(Note that the cases $k = 0$ and $k = n$ are trivial, since $F$ is a continuous distribution function.)
By the union bound applied to the event $A(t) = A_0(t) \cup \cdots \cup A_n(t)$, we deduce that
\[
	\Q \left( 
		A(t) \mid Y > t
	\right)
	\le 
	\sum_{k=0}^n \Q \left( A_k(t) \mid Y > t \right)
	\to 0, \qquad t \to x_F.
\]
For arbitrary $y \ge 0$, let $k(y) = \max \{ k = 0,1,\ldots,n-1 : x_k \le y \}$, so that $x_{k(y)} \le y < x_{k(y)+1}$. Then on the event $A(t)$, we have, by monotonicity of a distribution function,
\begin{align*}
	1 - F_t(\sigma_F(t)y)
	&\le 1 - F_t(\sigma_F(t)x_{k(y)}) \le (1+\eta x_{k(y)})_+^{-1/\eta} + \frac{\eps}{2} \le (1+\eta y)_+^{-1/\eta} + \frac{1}{n} + \frac{\eps}{2},
\end{align*}
and
\begin{align*}
	1 - F_t(\sigma_F(t)y) \ge 1 - F_t(\sigma_F(t) x_{k(y)+1}) \ge (1+\eta x_{k(y)+1})^{-1/\eta} - \frac{\eps}{2} \ge (1+\eta y)_+^{-1/\eta} - \frac{1}{n} - \frac{\eps}{2}.
\end{align*}
The two inequalities imply
\[
	\left|
		1 - F_t(\sigma_F(t)y)
		-
		(1+\eta y)_+^{-1/\eta}
	\right|
	\le \frac{1}{n} + \frac{\eps}{2}
	\le  \frac{\eps}{2} +  \frac{\eps}{2}
	= \eps.
\]
On the event $A(t)$, the previous inequality is true for every $y \ge 0$. We conclude that the uniform convergence statement~\eqref{eq:Fgpdunif} holds.
\end{proof}

\section{Normal distribution simulation study}\label{sec:normsim}

In the following, we consider forecasts that are well-known in the forecast verification literature, and evaluate them with respect to their probabilistic tail calibration. Suppose observations are drawn from $Y \mid \mu \sim N(\mu, 1)$, where $\mu \sim N(0, 1)$. The unconditional distribution of $Y$ is $N(0, 2)$, and $\mu$ represents a source of information that may or may not be available to the forecasters. 

Following \cite{GneitingEtAl2007}, consider four different forecasts for $Y$: the \textit{ideal} forecaster, $F_{\mathrm{id}} = N(\mu, 1)$; the \textit{climatological} forecaster, $F_{\mathrm{cl}} = N(0, 2)$; the \textit{unfocused} forecaster, $F_{\mathrm{un}} = \frac{1}{2} N (\mu, 1) + \frac{1}{2} N(\mu + \tau, 1)$, where $\tau$ is $-1$ or $1$ with equal probability, independently of $\mu$; and the \textit{sign-reversed} forecaster, $F_{\mathrm{sr}} = N (-\mu, 1)$. The unfocused and sign-reversed forecasters were introduced in Examples \ref{ex:1} and \ref{ex:2}, respectively.

We draw $n = 10^{6}$ independent realizations of $\mu$ and $\tau$ to obtain the forecasts. Conditionally on $\mu$, the outcome $Y$ is then drawn for each instance. Figure \ref{fig:ss_ptc_reldiag_norm} displays the probabilistic tail calibration of the four forecasters for five thresholds. The climatological and ideal forecasters are auto-calibrated, and are therefore probabilistically calibrated and probabilistically tail calibrated. The unfocused forecaster is probabilistically calibrated but not probabilistically tail calibrated, while the sign-reversed forecaster is neither probabilistically nor probabilistically tail calibrated (though it is marginally calibrated). Similarly to in Example \ref{ex:exponential}, we could assess conditional notions of calibration, as described in Section \ref{sec:btail}, allowing us to confirm that the ideal forecaster is $\sigma(\mu)$-tail calibrated, whereas the climatological forecaster is not. 

\begin{figure}
    \centering
    \includegraphics[width=0.3\textwidth]{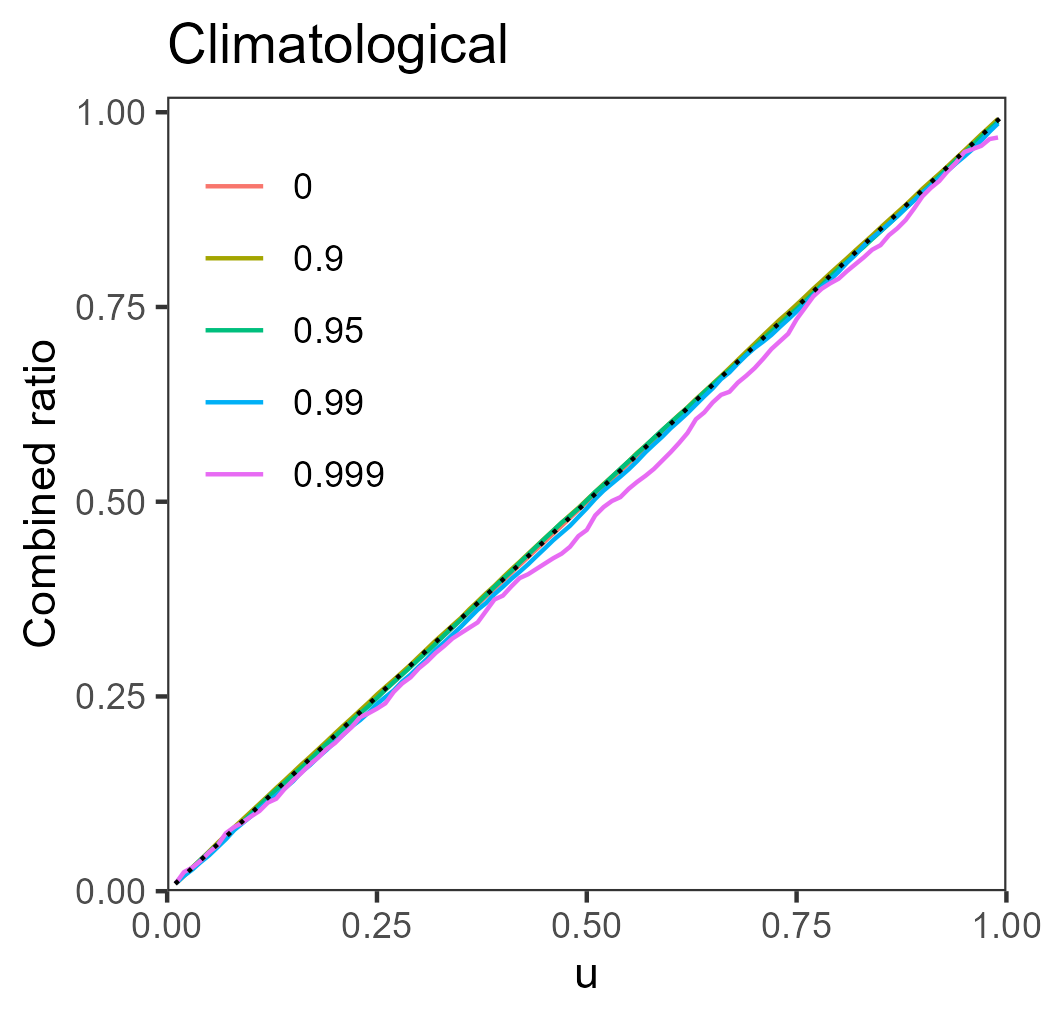}
    \includegraphics[width=0.3\textwidth]{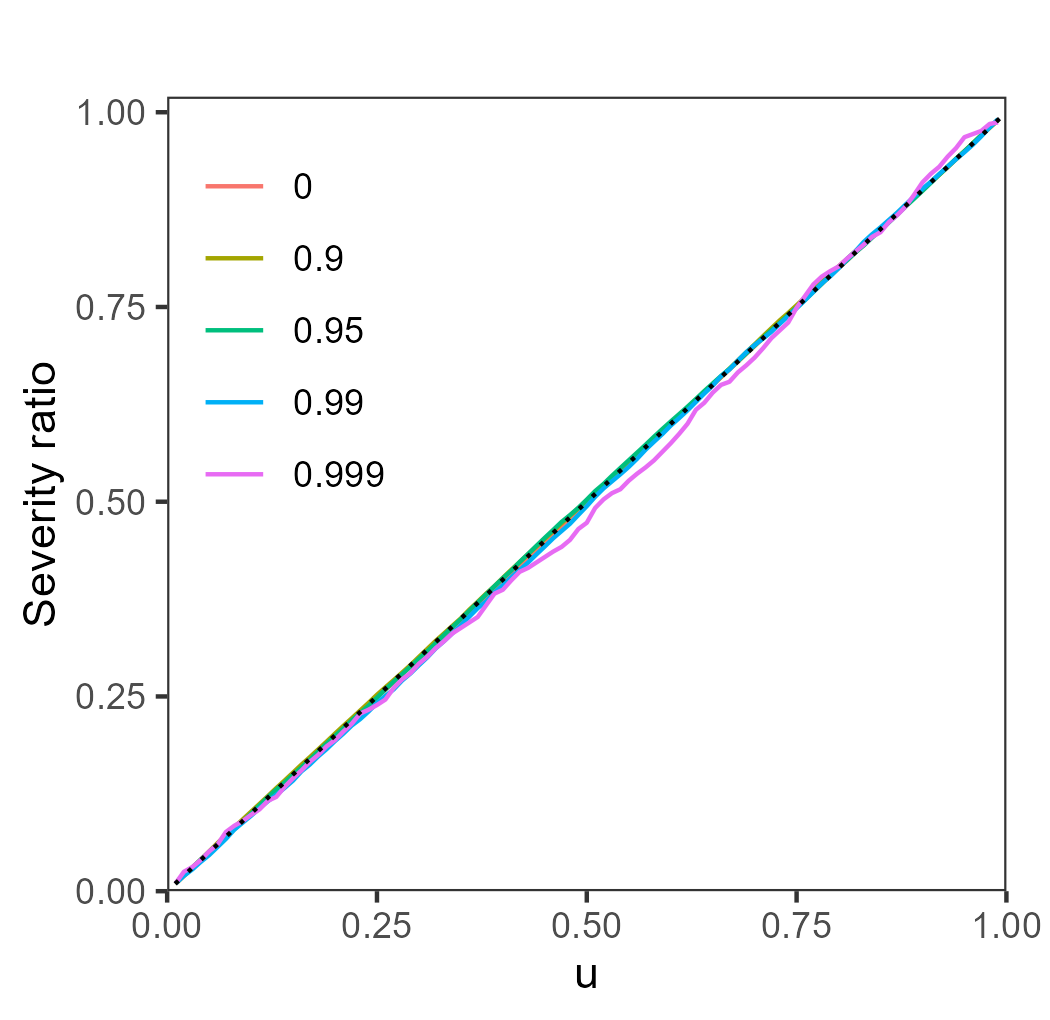}
    \includegraphics[width=0.3\textwidth]{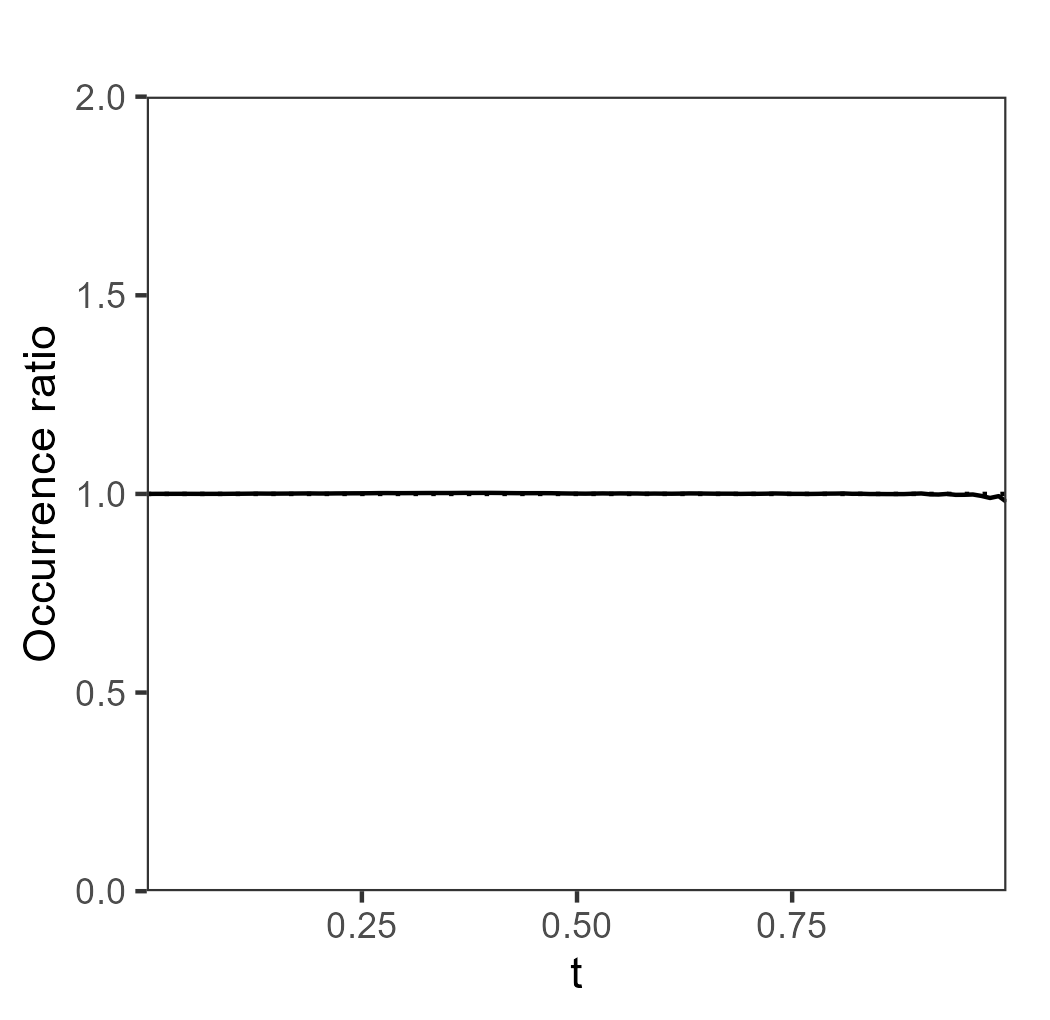}
    \includegraphics[width=0.3\textwidth]{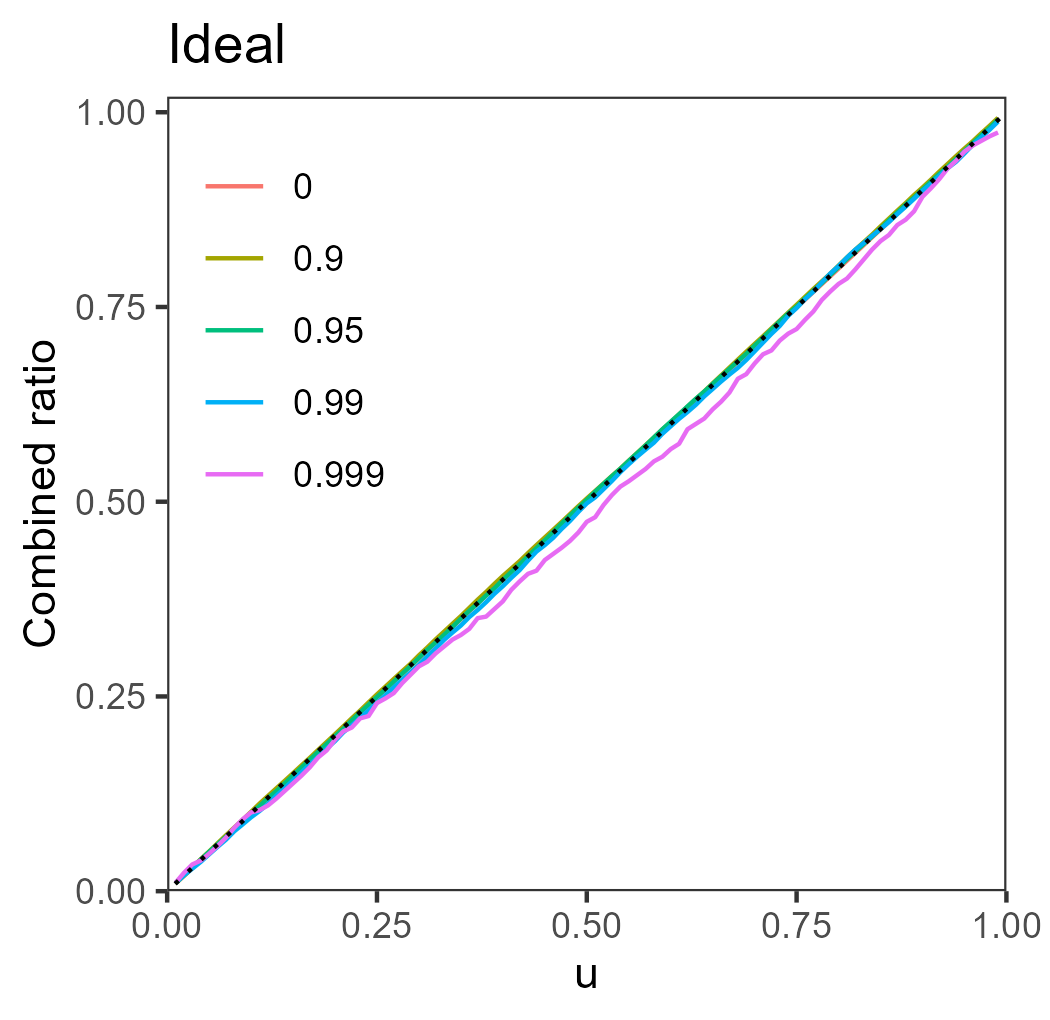}
    \includegraphics[width=0.3\textwidth]{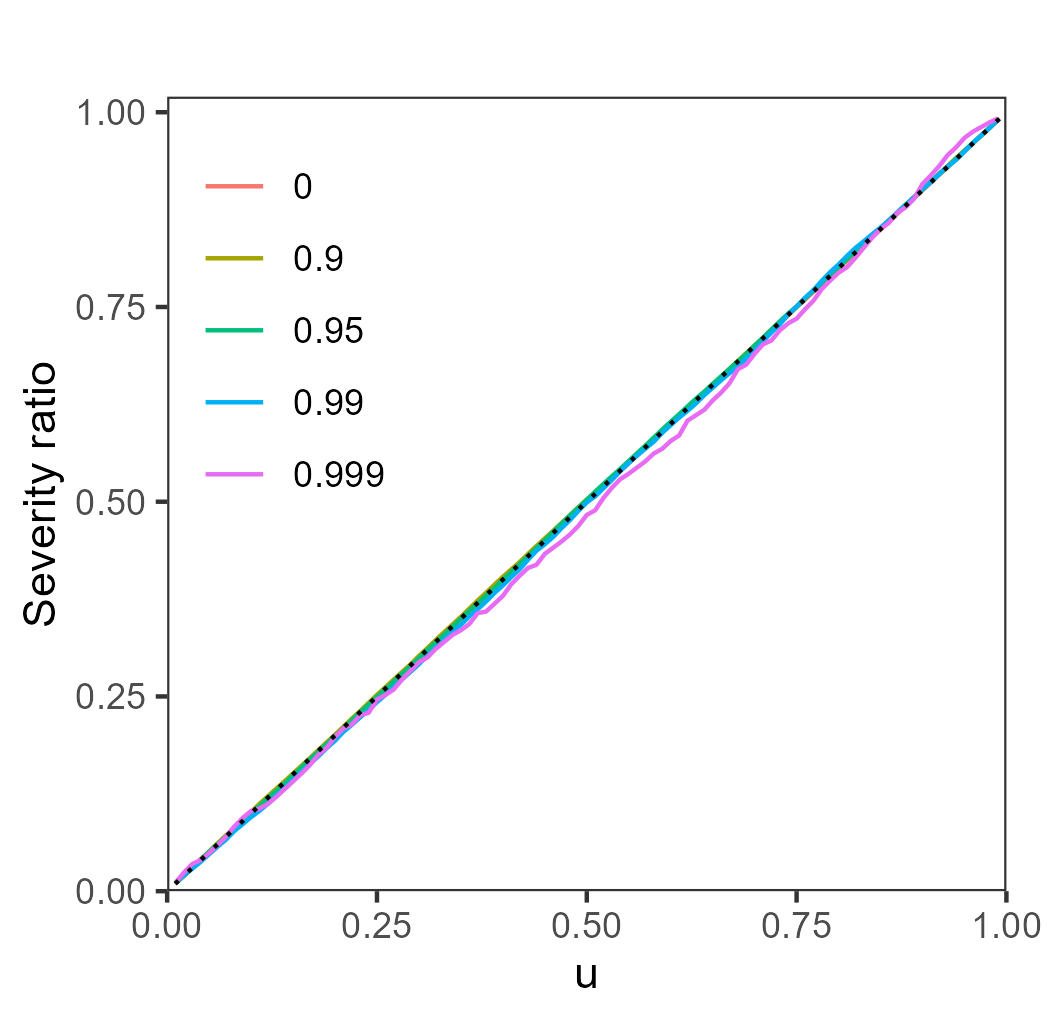}
    \includegraphics[width=0.3\textwidth]{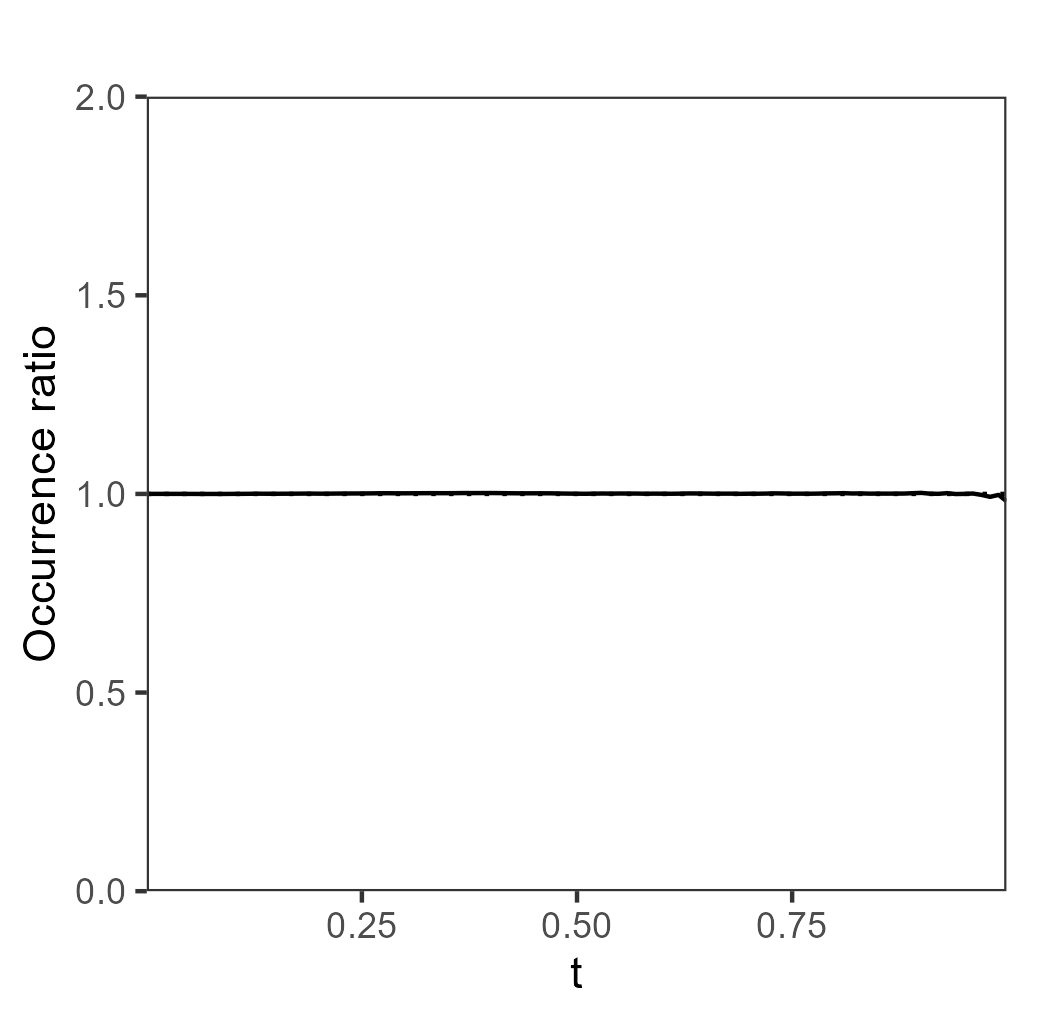}
    \includegraphics[width=0.3\textwidth]{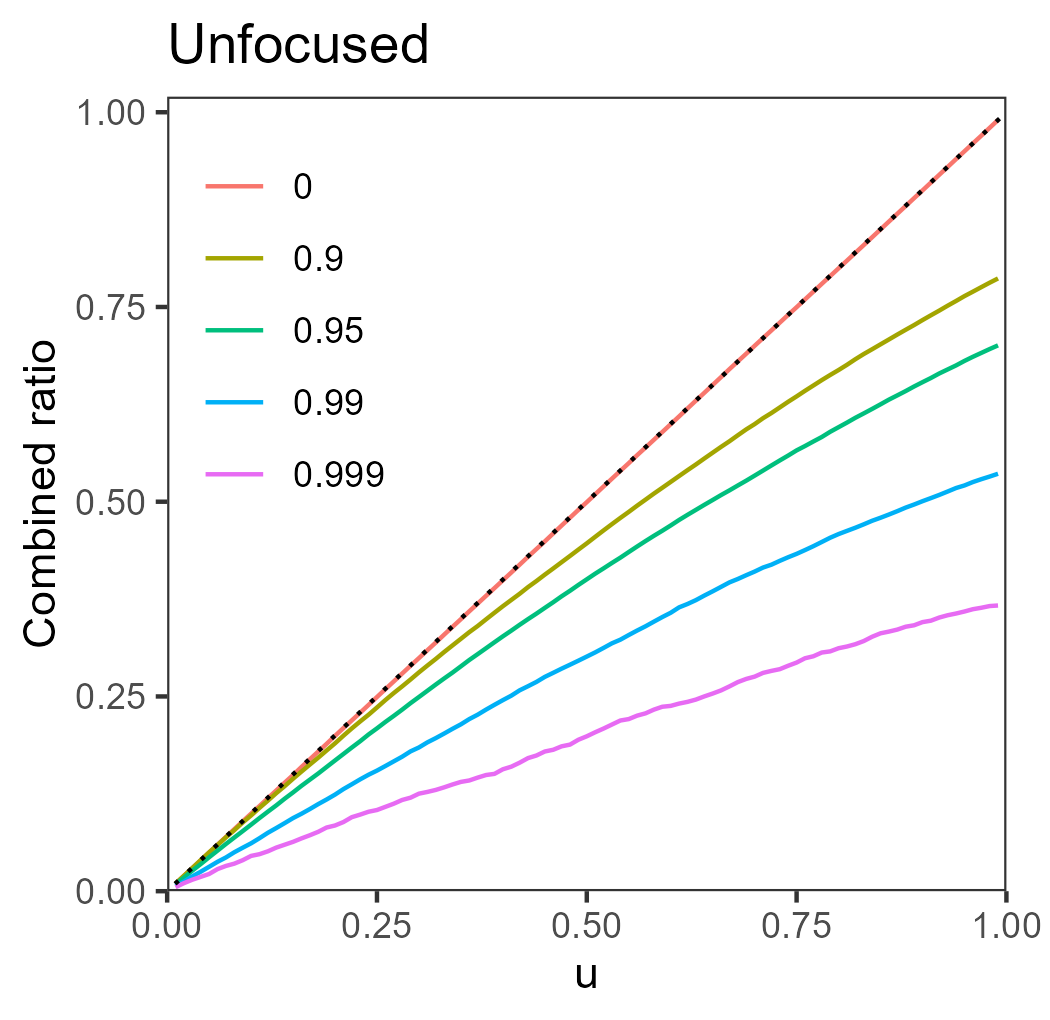}
    \includegraphics[width=0.3\textwidth]{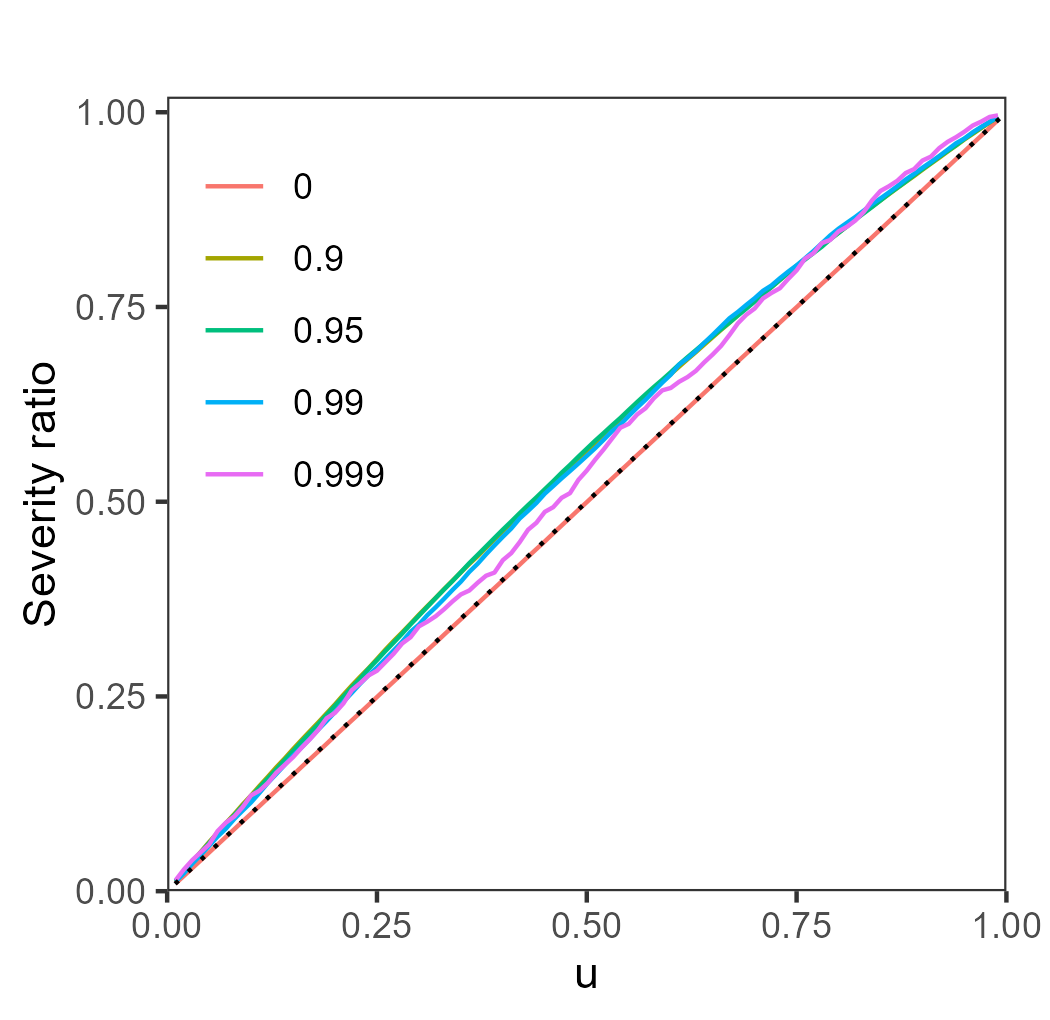}
    \includegraphics[width=0.3\textwidth]{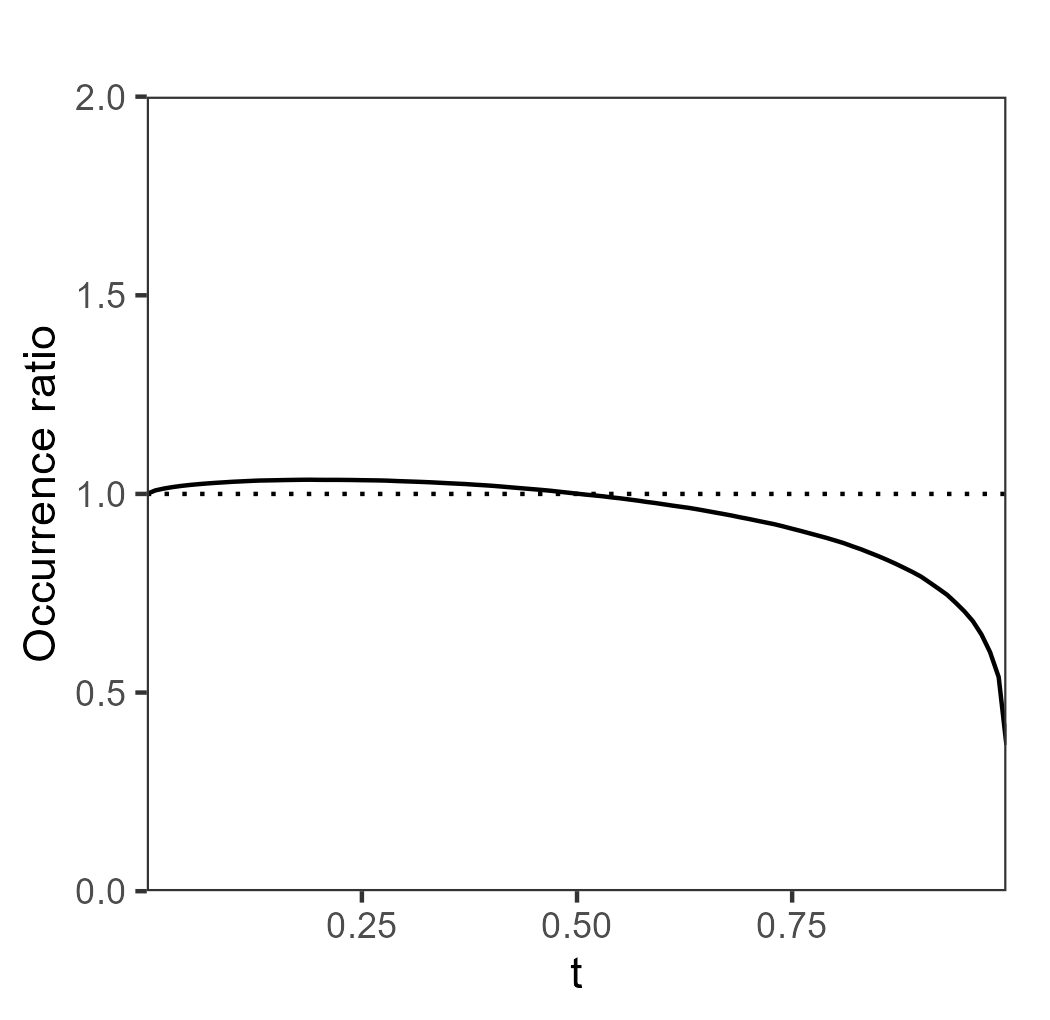}
    \includegraphics[width=0.3\textwidth]{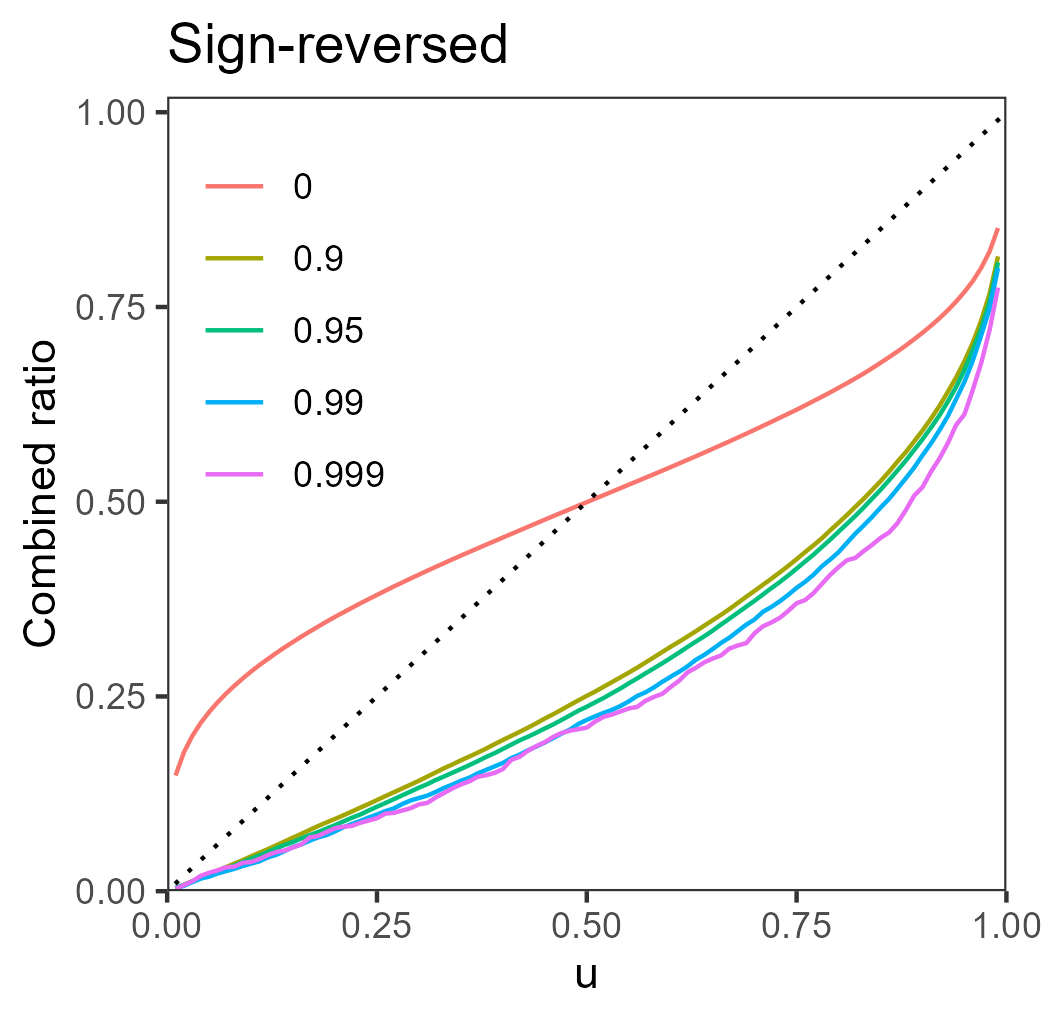}
    \includegraphics[width=0.3\textwidth]{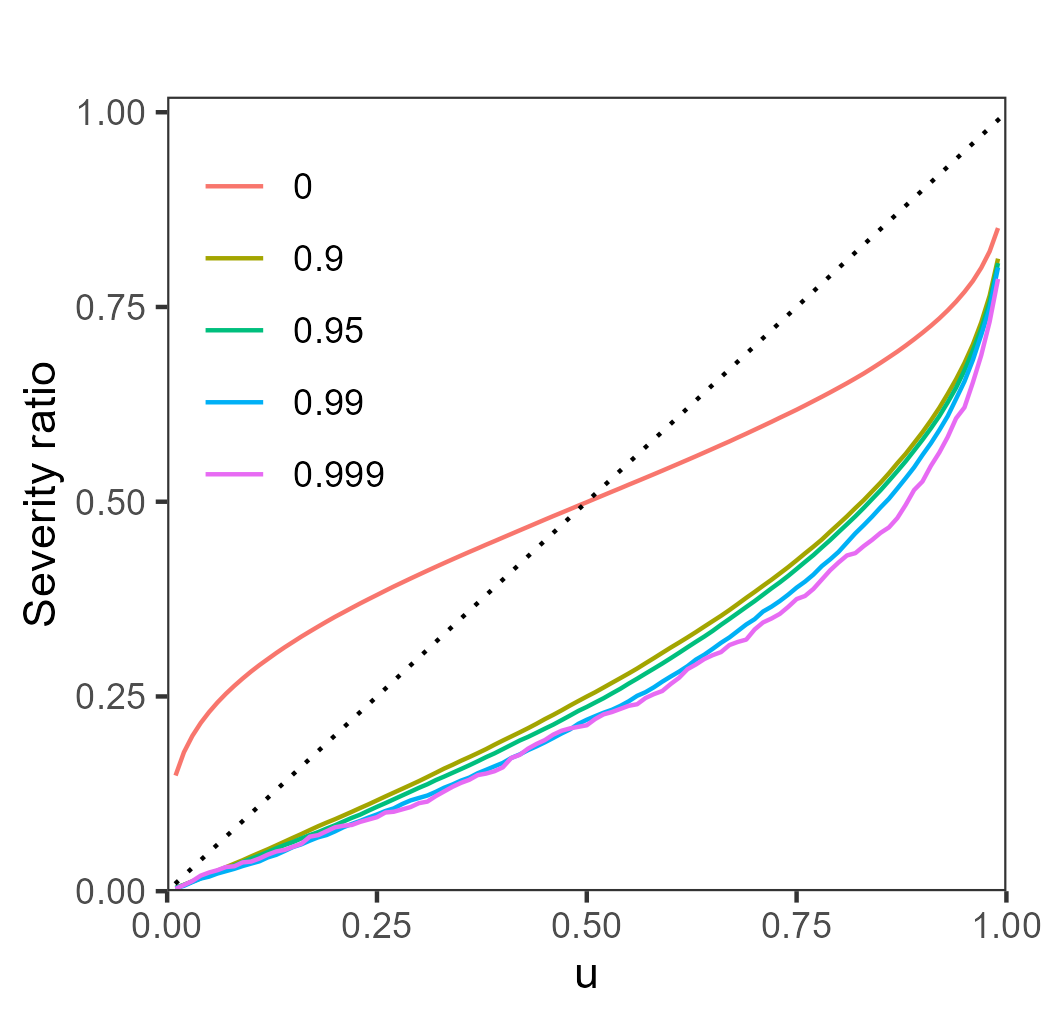}
    \includegraphics[width=0.3\textwidth]{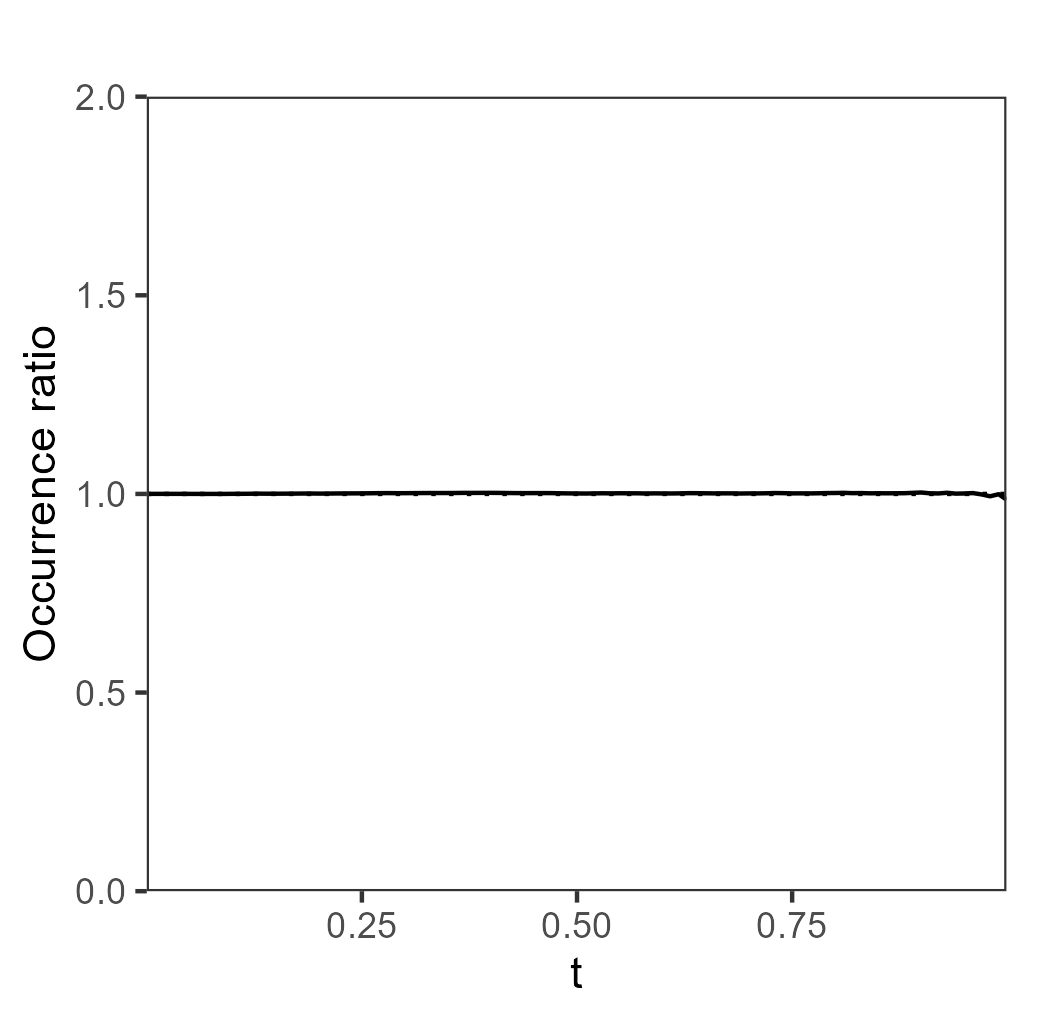}
    \caption{Probabilistic tail calibration diagnostic plots for the climatological, ideal, unfocused, and sign-reversed forecasters in Section \ref{sec:normsim}. Results are shown for five thresholds, expressed in terms of quantiles $\alpha$ of the $10^{6}$ observations.}
    \label{fig:ss_ptc_reldiag_norm}
\end{figure}

\begin{figure}
    \centering
    \includegraphics[width=0.3\textwidth]{Figures/sim_ex_com_1e6_cl.png}
    \includegraphics[width=0.3\textwidth]{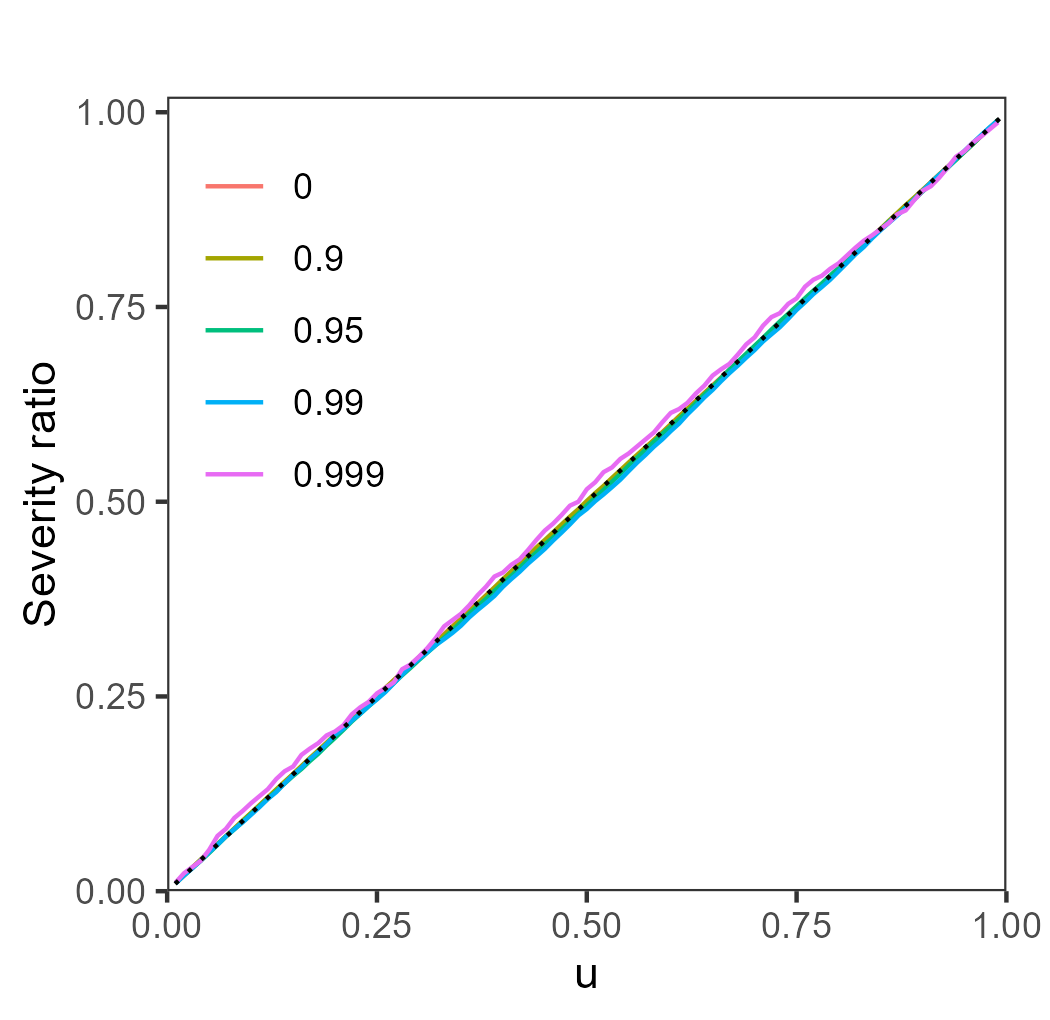}
    \includegraphics[width=0.3\textwidth]{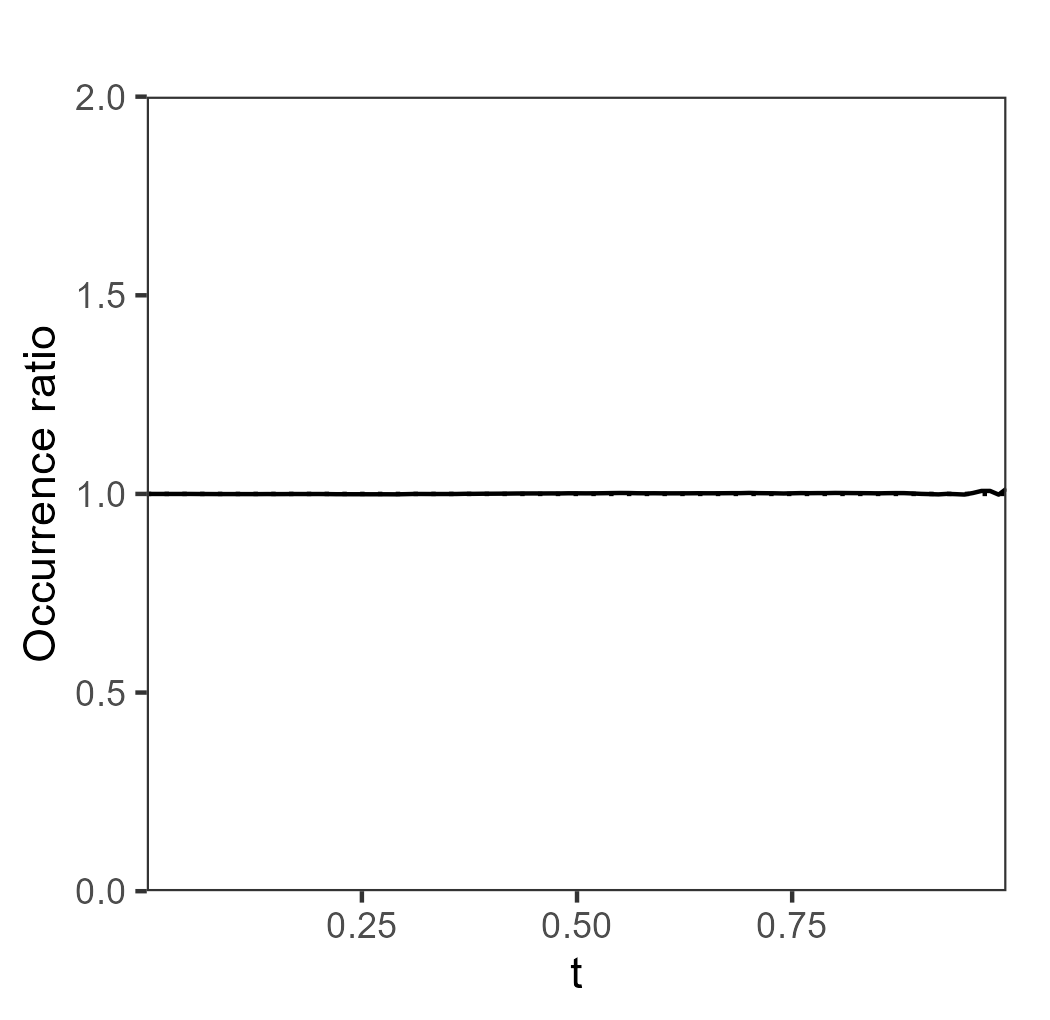}
    \includegraphics[width=0.3\textwidth]{Figures/sim_ex_com_1e6_id.png}
    \includegraphics[width=0.3\textwidth]{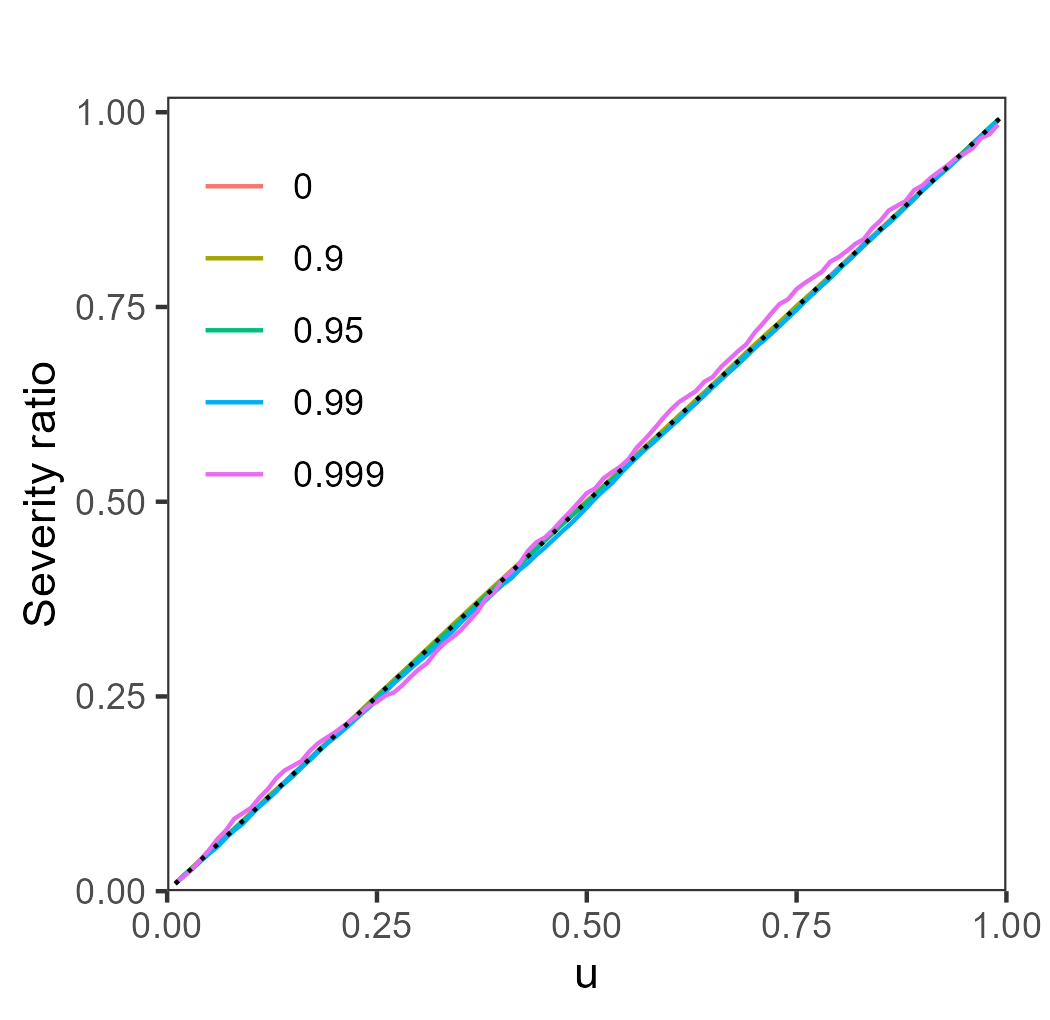}
    \includegraphics[width=0.3\textwidth]{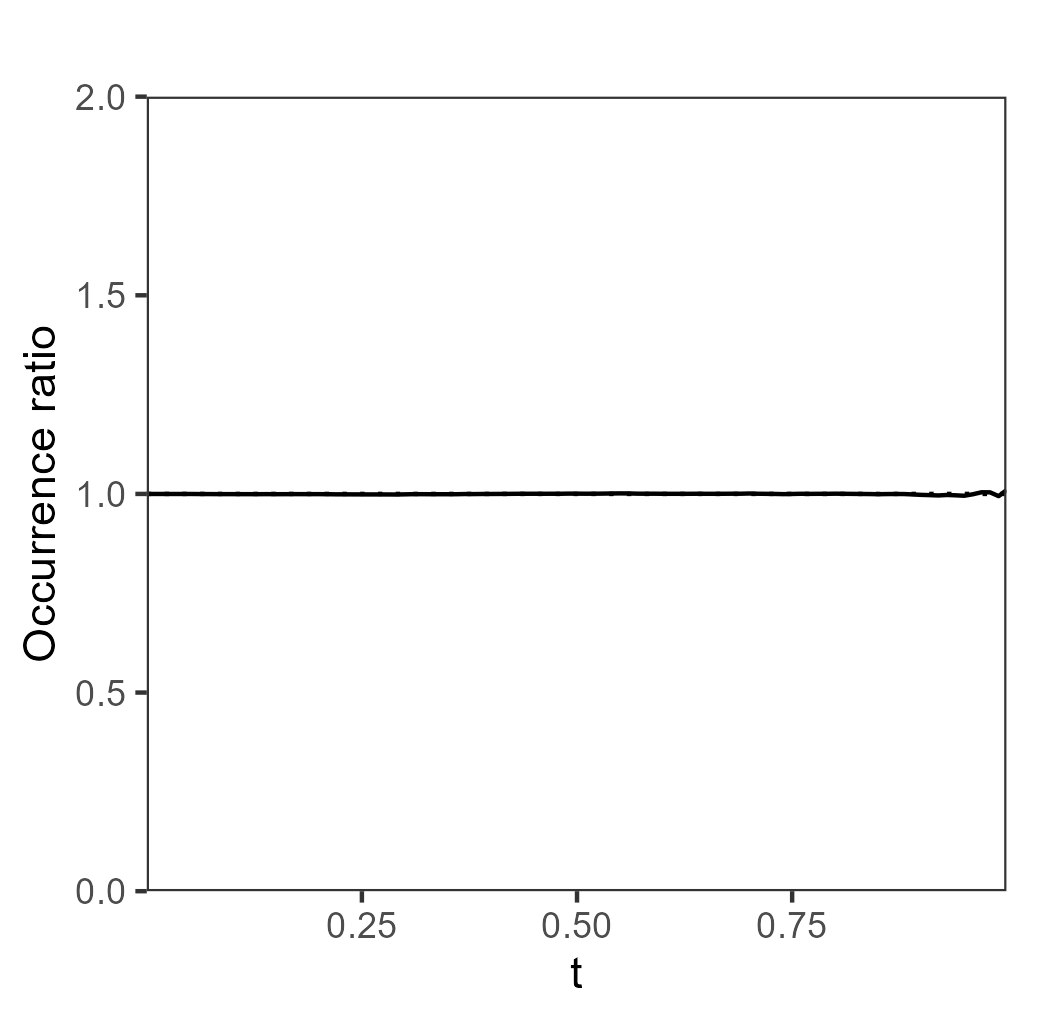}
    \includegraphics[width=0.3\textwidth]{Figures/sim_ex_com_1e6_ex.png}
    \includegraphics[width=0.3\textwidth]{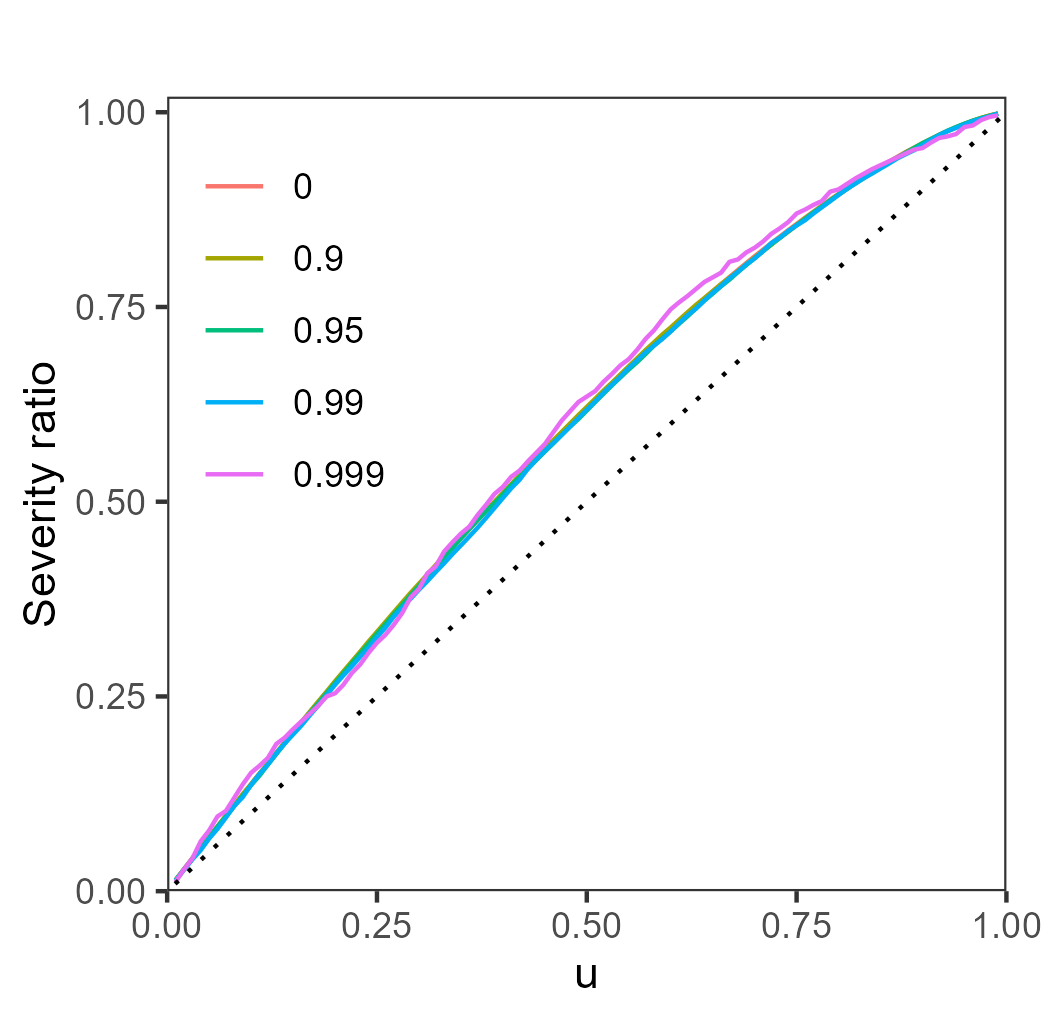}
    \includegraphics[width=0.3\textwidth]{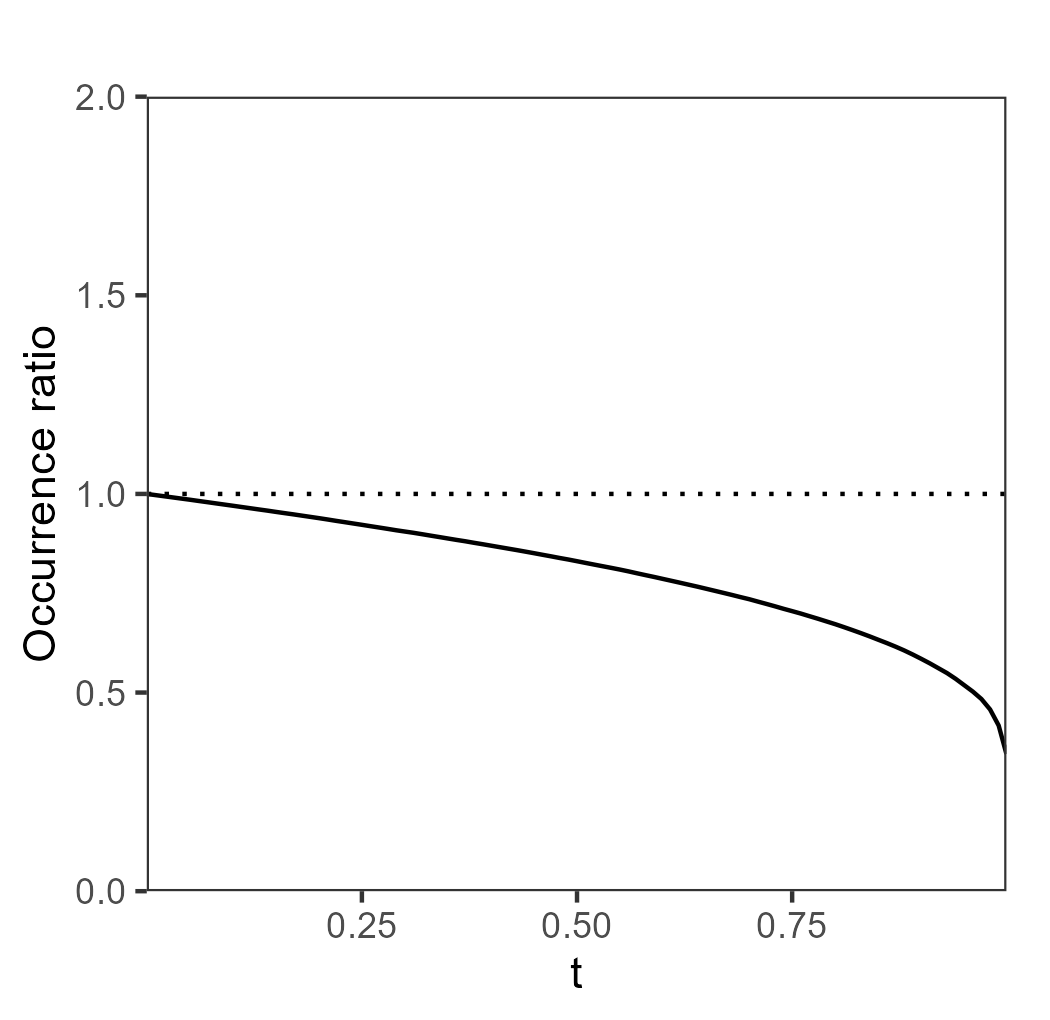}
    \caption{Probabilistic tail calibration plots for the climatological, ideal, and extremist forecasters in Example \ref{ex:exponential}. Results are shown for five thresholds, expressed in terms of quantiles $\alpha$ of the $10^{6}$ observations.}
    \label{fig:ss_ptc_reldiag2}
\end{figure}

\section{Additional case study results}\label{app:cs_extra}

The results in Section \ref{sec:casestudy} are shown for fixed thresholds of 5, 10, and 15mm. While the region we consider is fairly homogeneous, one could argue that the impact associated with a given precipitation amount will differ depending on the location. In practice, the tail calibration diagnostic tools can also be implemented using variable thresholds. We additionally consider the case when the threshold is chosen to be an extreme quantile of the local climatology, which therefore changes for each location considered. Figure~\ref{fig:cs_pit_qu} displays the corresponding tail calibration plots. Forecasts for the occurrence of extreme events are more reliable when a variable threshold is used. However, in both cases, the excess PIT values are not uniformly distributed for any of the three forecasting methods. In particular, the post-processed forecasts appear probabilistically calibrated but not probabilistically tail calibrated, clearly exhibiting a tail that is too light.

\begin{figure}
    \centering
    \includegraphics[width=0.3\textwidth]{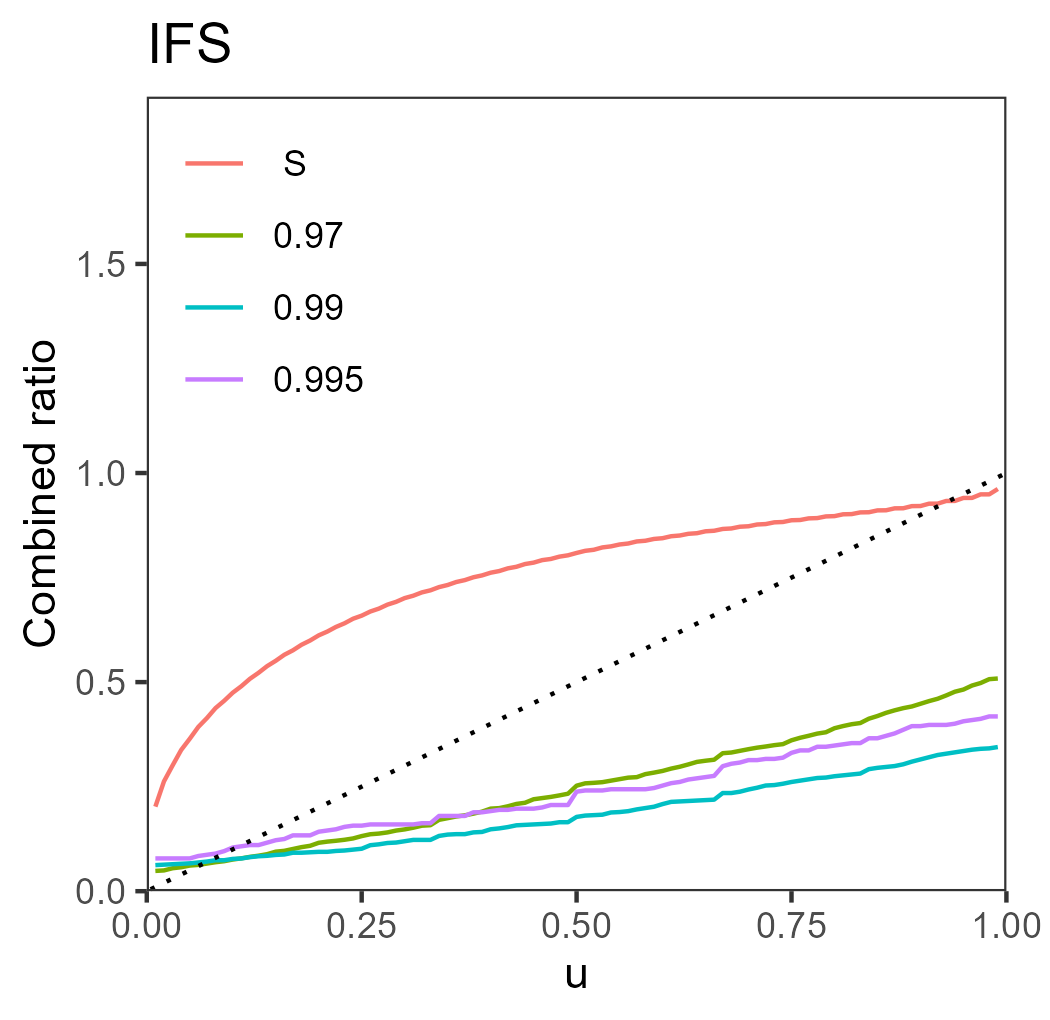}
    \includegraphics[width=0.3\textwidth]{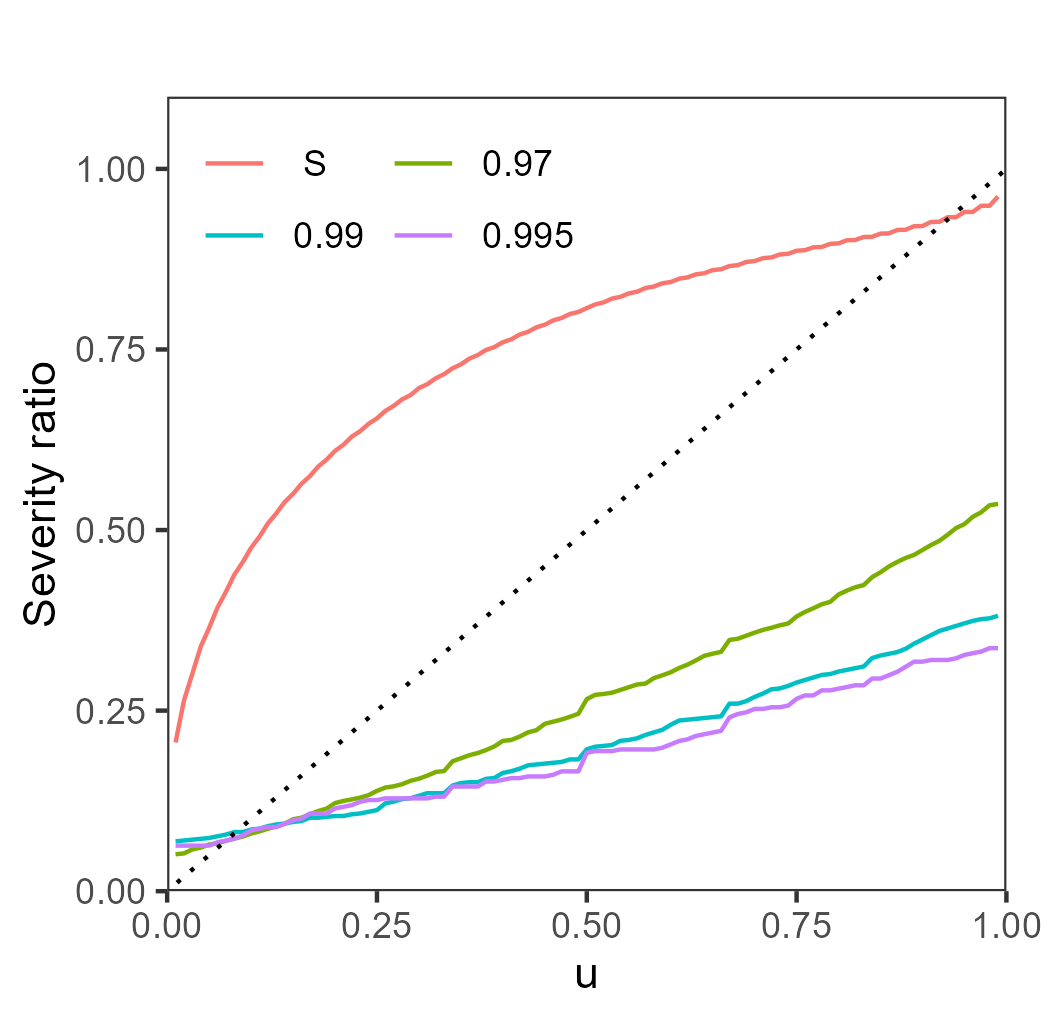}
    \includegraphics[width=0.3\textwidth]{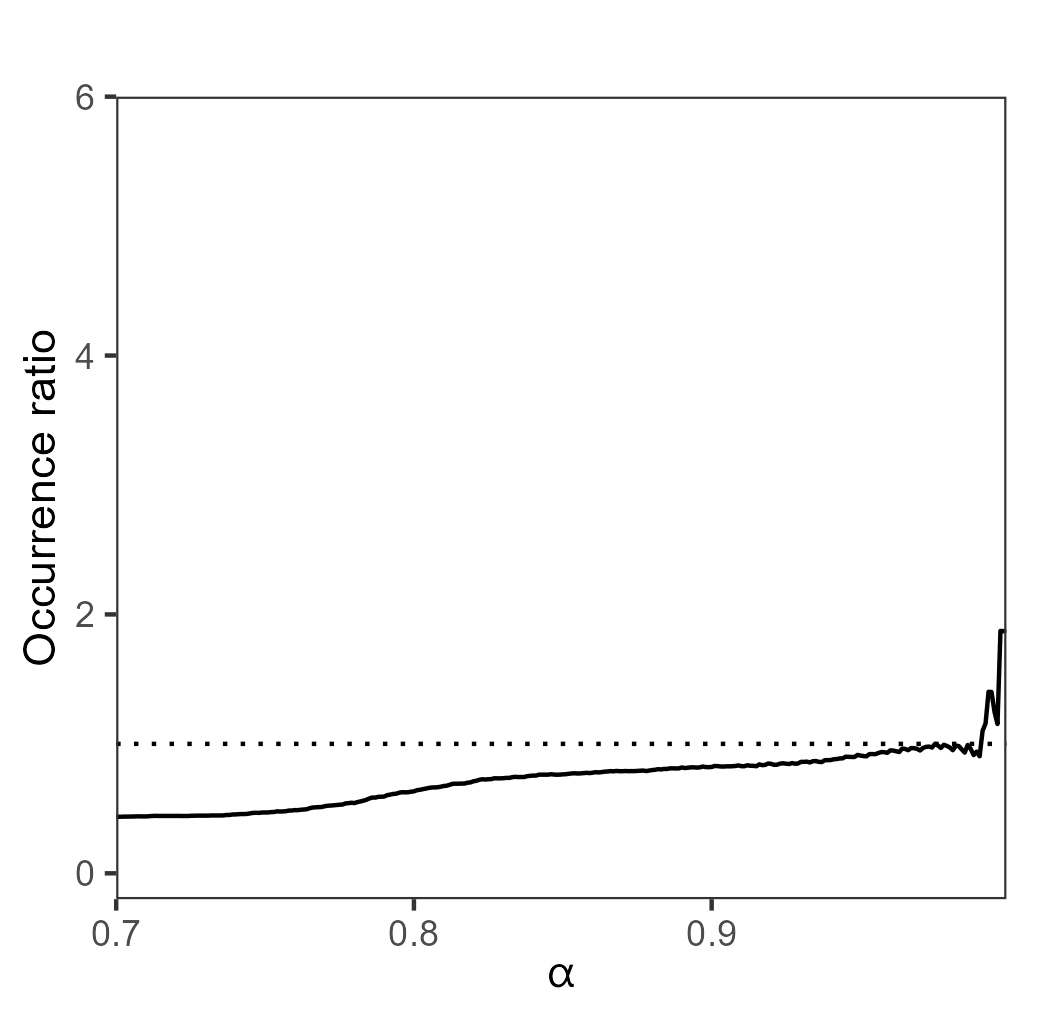}
    \includegraphics[width=0.3\textwidth]{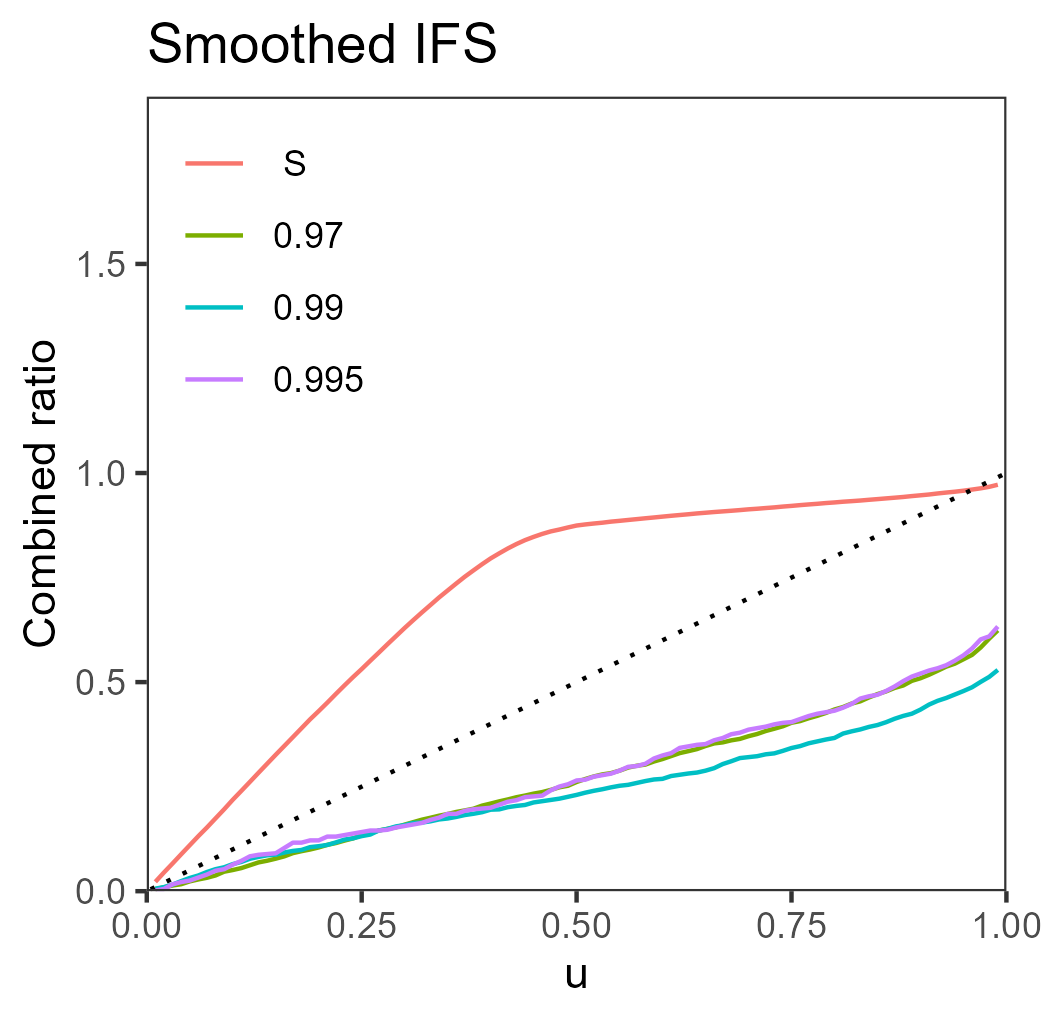}
    \includegraphics[width=0.3\textwidth]{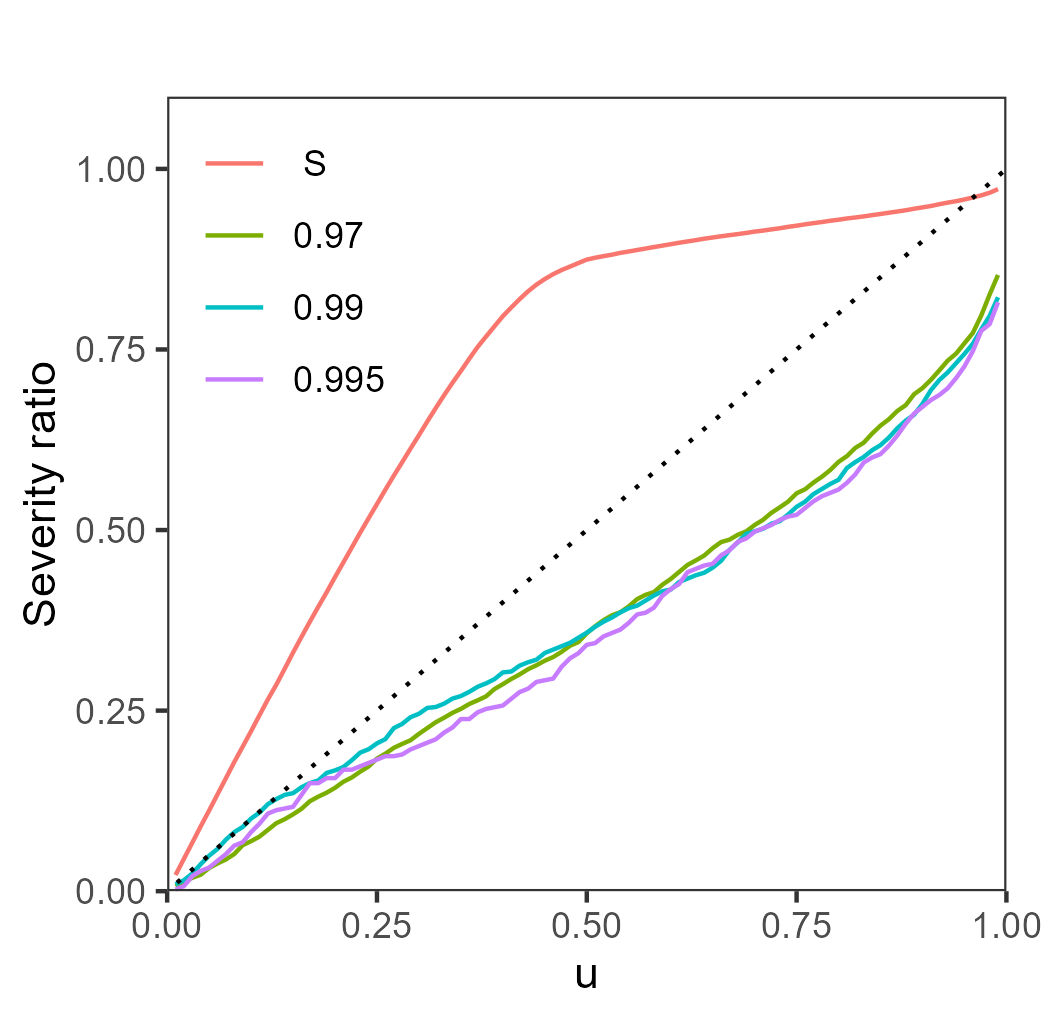}
    \includegraphics[width=0.3\textwidth]{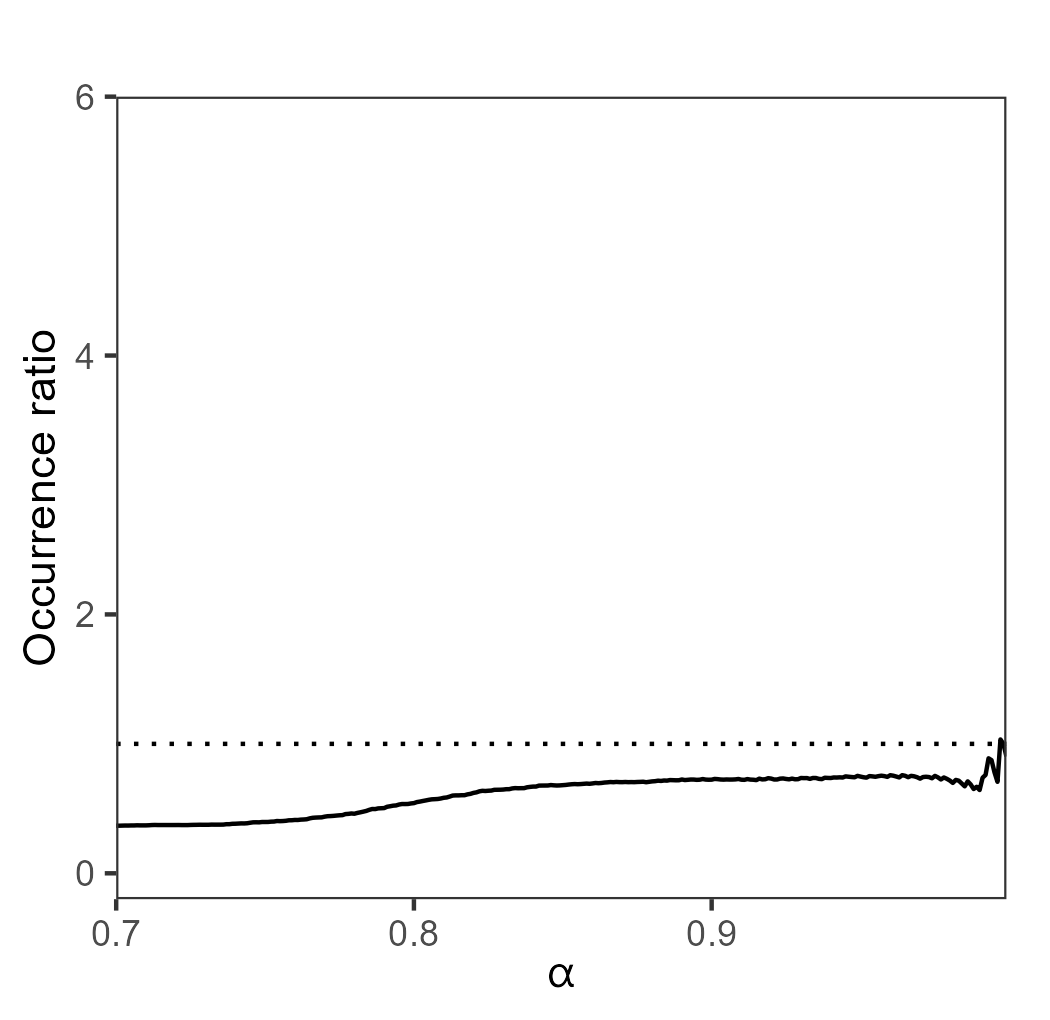}
    \includegraphics[width=0.3\textwidth]{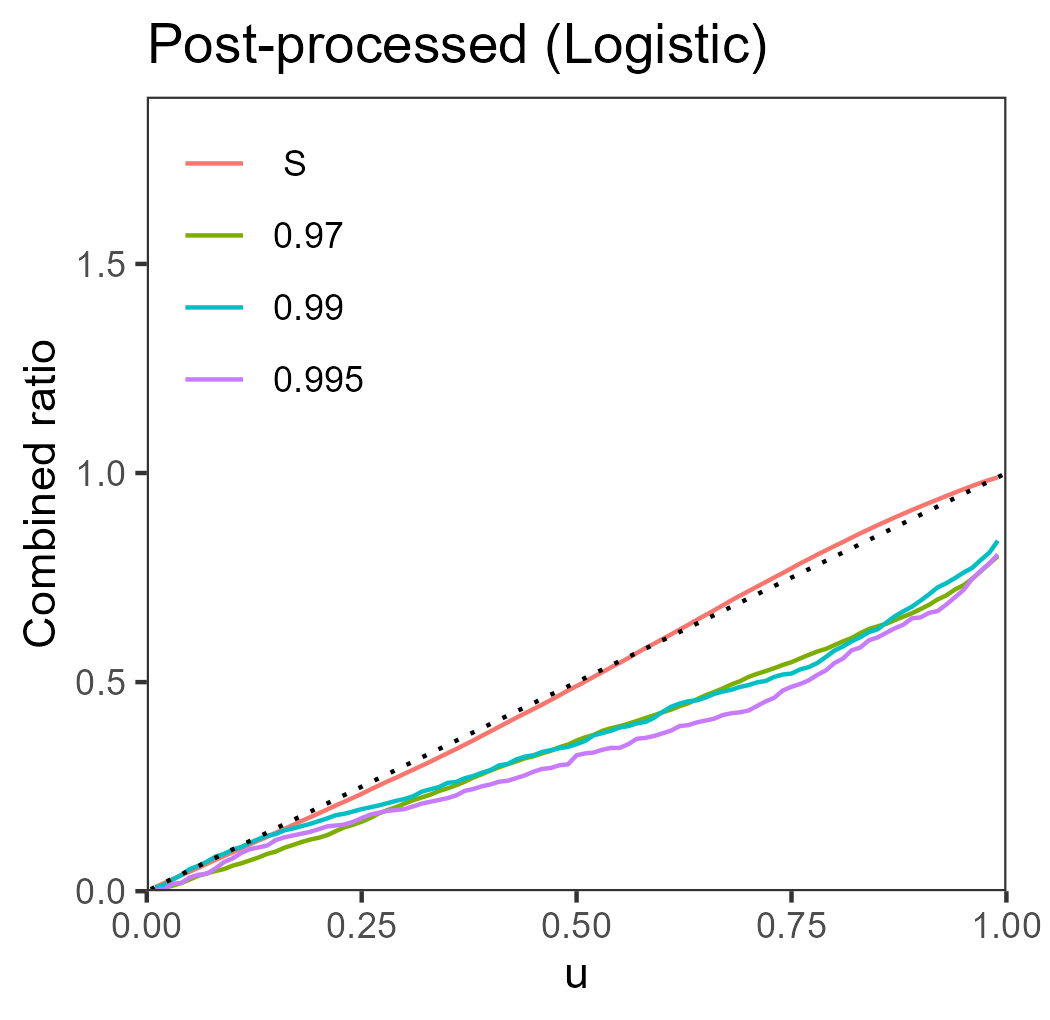}
    \includegraphics[width=0.3\textwidth]{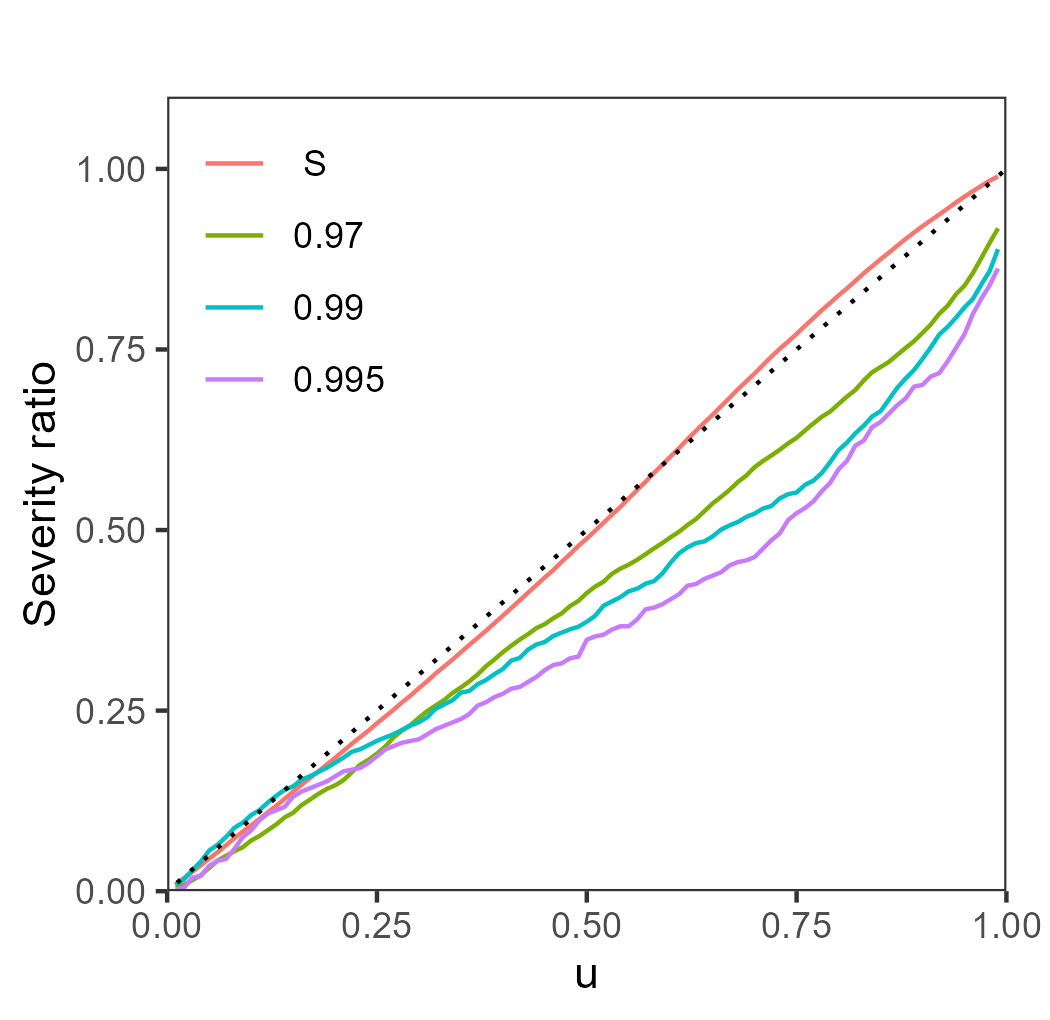}
    \includegraphics[width=0.3\textwidth]{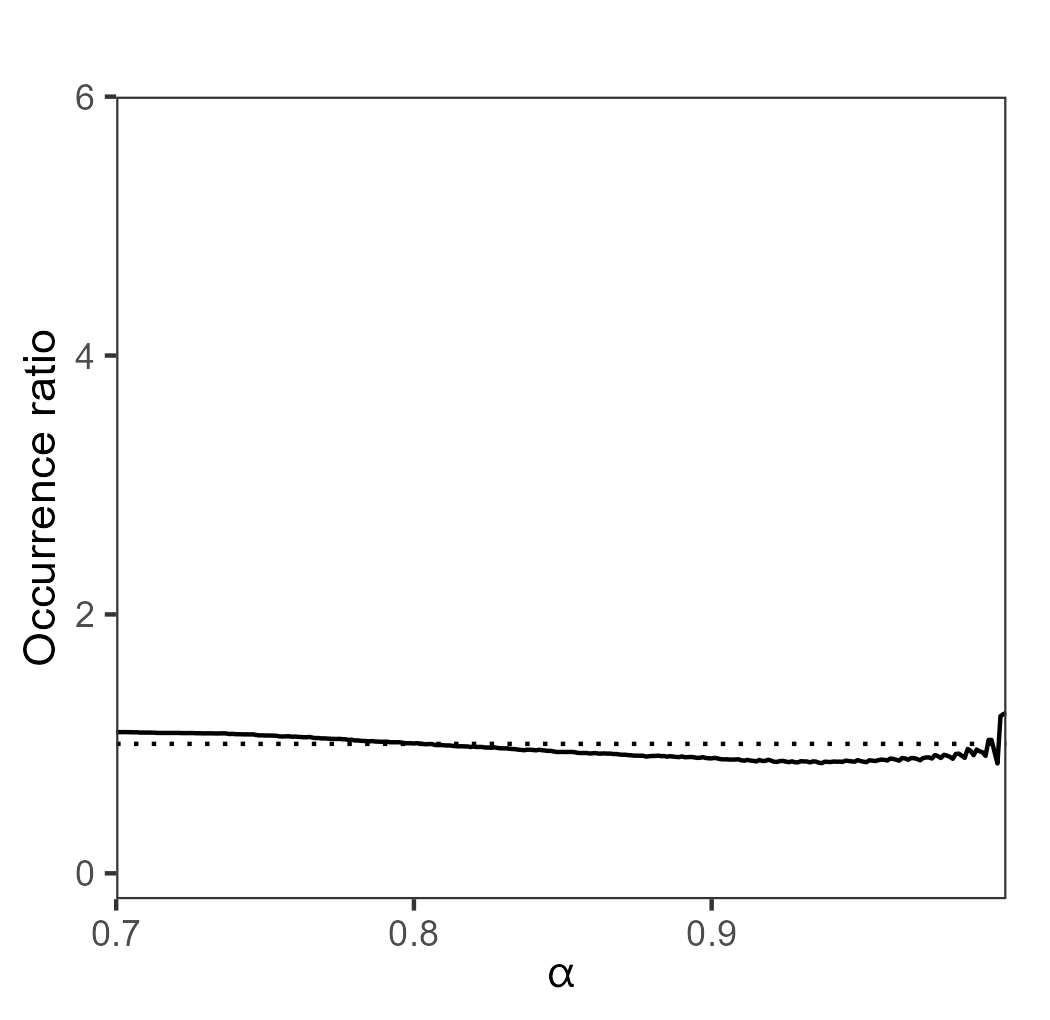}
    \includegraphics[width=0.3\textwidth]{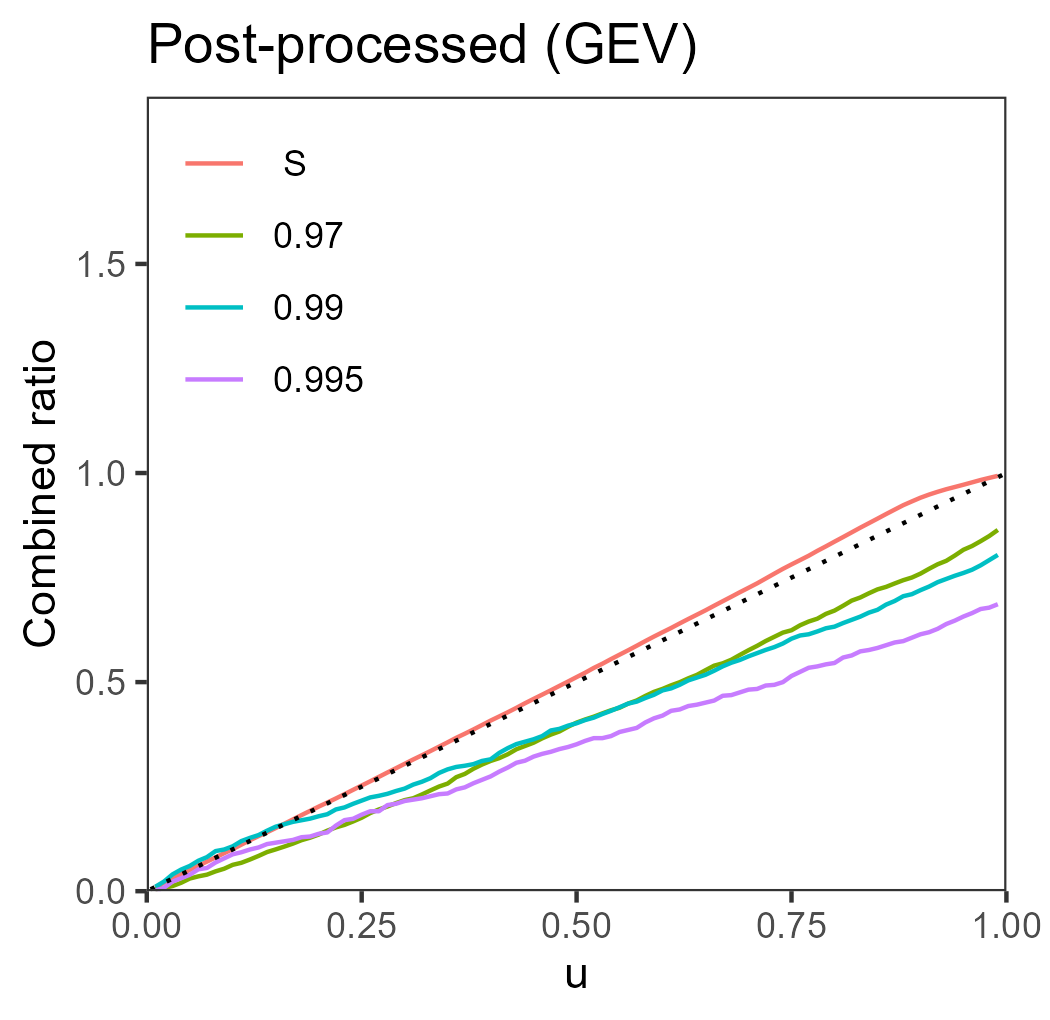}
    \includegraphics[width=0.3\textwidth]{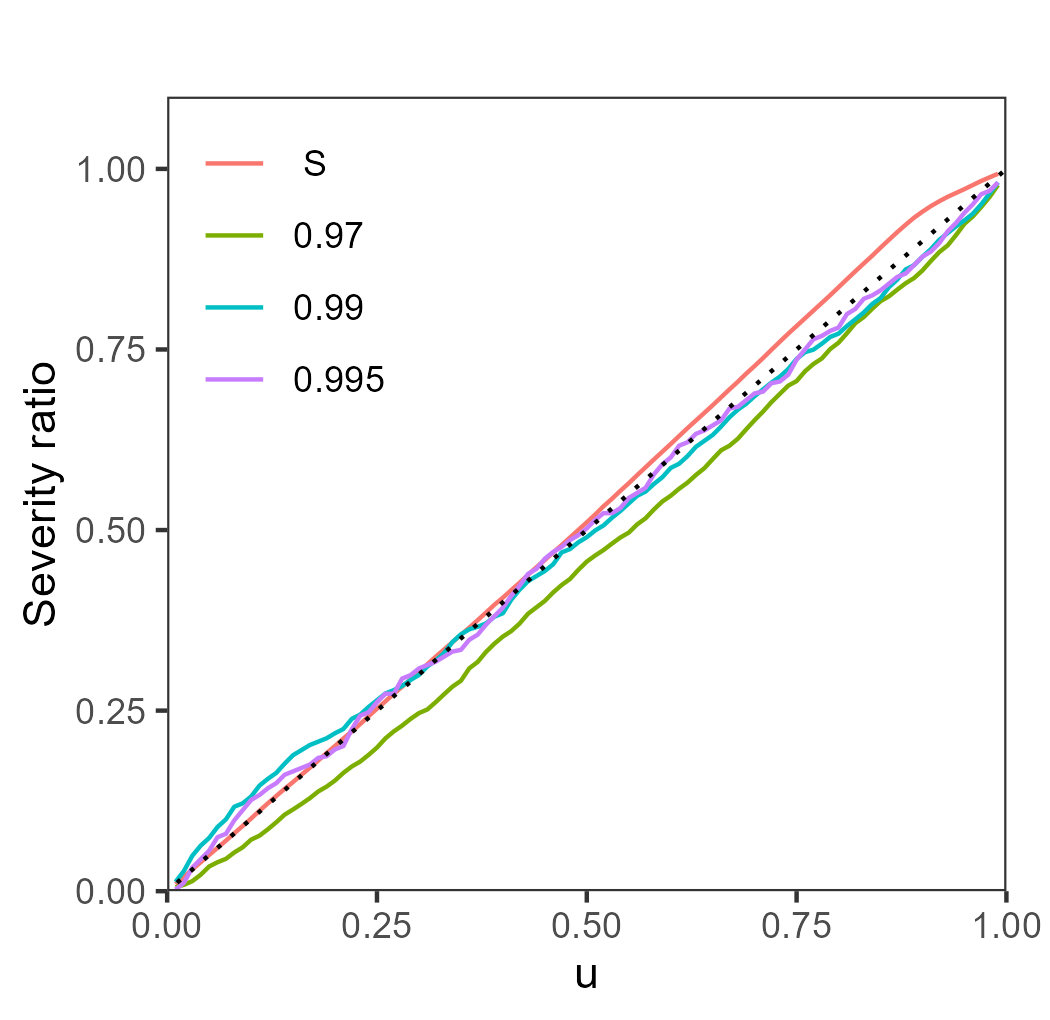}
    \includegraphics[width=0.3\textwidth]{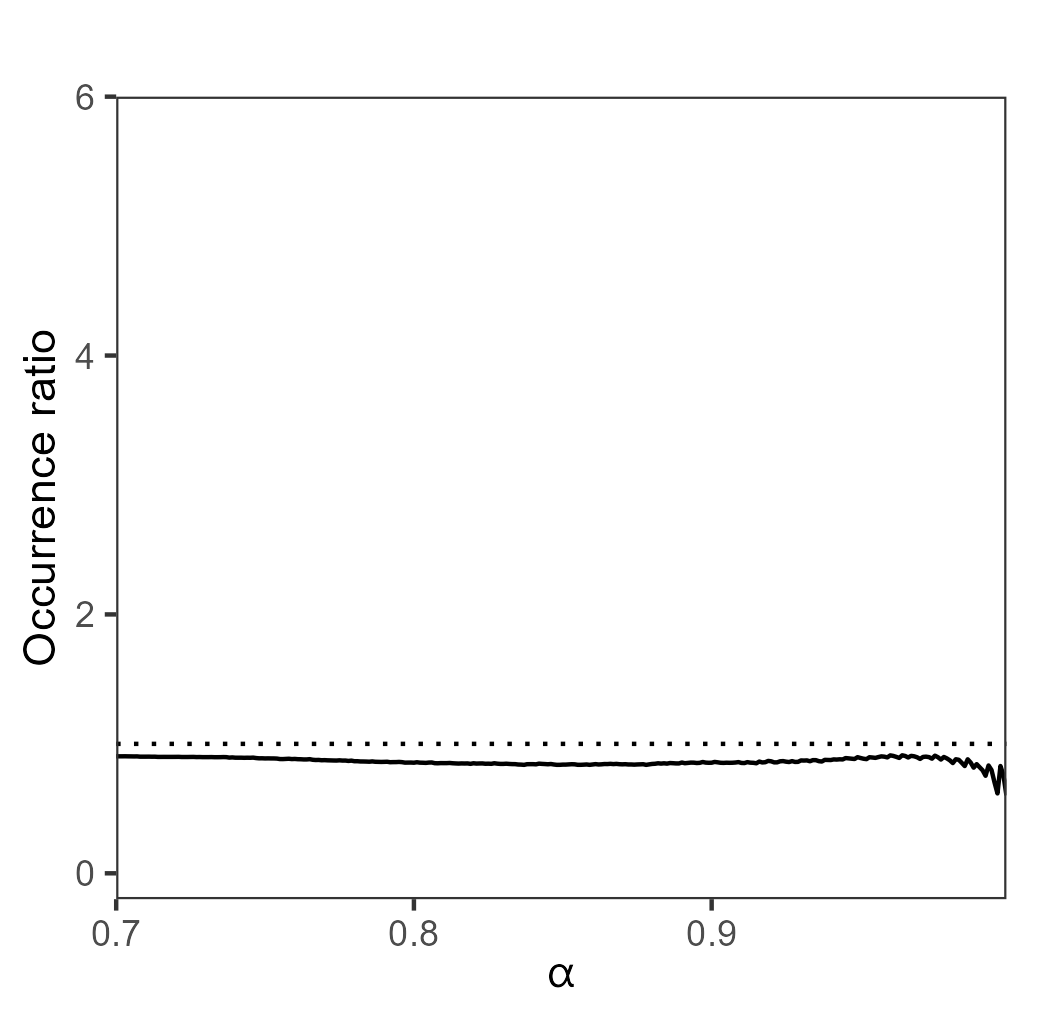}
    \caption{As in Figure \ref{fig:cs_pit} of the main manuscript, but implemented with thresholds that are quantiles of the precipitation distribution at each station; $\alpha$ denotes the quantile level.}
    \label{fig:cs_pit_qu}
\end{figure}

\section{Confidence intervals}\label{app:cs_confint}

Since the occurrence, severity and combined ratios in Equations~\eqref{eq:tailBC:occ}, \eqref{eq:tailBC:sev} and \eqref{eq:tailBCt} in the paper, respectively, are all ratios of sample means, we can use the central limit theorem and the delta method to derive pointwise confidence intervals for these quantities in practice. In the following, we suppose that we observe i.i.d. forecast-observation pairs $(F_i, Y_i)$ for $i \in \{1,\ldots,n\}$; $(F, Y)$ denotes a generic such pair.

\subsection{Occurrence ratio}

For one or more real thresholds $t$, suppose we want to check whether the \emph{occurrence ratio},
\begin{equation*}
	\frac{\sum_{i=1}^{n} \one\{ Y_{i} > t\}}{\sum_{i=1}^{n} (1 - F_{i}(t))},
\end{equation*}
is ``close to one''. Let $G(t) = \Q(Y \le t)$ be the true cdf of $Y$, and define empirical versions as
\begin{align*}
    \widehat{G}_n(t) &= \frac{1}{n} \sum_{i=1}^n \one\{ Y_i \le t \}, \\
    \widehat{F}_n(t) &= \frac{1}{n} \sum_{i=1}^n F_i(t).
\end{align*}
By the multivariate central limit theorem, we have the convergence in distribution
\begin{equation}
\label{eq:FmultiCLT}
    \sqrt{n} \left( 
        \widehat{G}_n(t) - G(t), \, 
        \widehat{F}_n(t) - \expec[F(t)] 
    \right)'
    \xrightarrow{d}
    N_2(0, \Sigma(t)), \qquad n \to \infty,
\end{equation}
where the apostrophe denotes matrix transposition, with the limit distribution being bivariate Gaussian with zero mean vector and covariance matrix
\[
    \Sigma(t) = \begin{pmatrix}
        G(t) \left( 1 - G(t) \right) &
        \cov \left[ \one\{Y \le t\}, F(t) \right] \\
        \cov \left[ \one\{Y \le t\}, F(t) \right] & \var[F(t)]
    \end{pmatrix}.
\]
The same limit holds if we use survival functions in Eq.~\eqref{eq:FmultiCLT}.

By the delta method applied to the function $g(x, y) = x/y$, with gradient $\nabla g(x, y) = (1/y, -x/y^2)$, we find
\begin{equation}
    \label{eq:occratCLT}
    \sqrt{n} \left(
        \frac{1-\widehat{G}_n(t)}{1-\widehat{F}_n(t)} - 
        \frac{1-G(t)}{\expec[1-F(t)]}
    \right)
    \xrightarrow{d}
    N(0, \sigma^2(t)), \qquad n \to \infty,
\end{equation}
with limiting variance
\[
    \sigma^2(t)
    = v(t)' \Sigma(t) v(t)
    \quad \text{with} \quad
    v(t)' = \nabla g \left( 1-G(t), \expec[1-F(t)] \right).
\]
All quantities appearing in the expression $\sigma^2_t$ can be estimated by the obvious sample quantities, leading to a plug-in estimate $\hat{\sigma}^2_{n,t}$. This yields the possibility for (classical) statistical inference based on the normal approximation in the usual way, including the construction of point-wise confidence intervals for $(1 - G(t)) / \expec[1-F(t)]$, and hypothesis tests for $H_0 : 1 - G(t) = \expec[1 - F(t)]$, for example.

\subsection{Combined and severity ratios}

Now, for one or more real thresholds $t$, suppose that we want to check whether the \emph{combined ratio},
\begin{equation}
	\label{eq:combined}
	\frac{\sum_{i =1}^n \one\{\Zit \leq u, Y_i > t\}}{ \sum_{i = 1}^{n} (1 - F_{i}(t))},
\end{equation}
and the \emph{severity ratio},
\begin{equation}
	\label{eq:severity}
	\frac{\sum_{i =1}^n \one\{\Zit \leq u, Y_i > t\}}{\sum_{i = 1}^{n} \one\{Y_i > t\}},
\end{equation}
are ``close to the diagonal''. Define
\begin{align*}
    \widehat{A}_n(t,u) &= \frac{1}{n} \sum_{i=1}^n \one\{ \Zit \le u, Y_i > t \}, \\
    \widehat{B}_n(t) &= \frac{1}{n} \sum_{i=1}^n (1-F_i(t)),\\
    \widehat{C}_n(t) &= \frac{1}{n} \sum_{i=1}^n \one\{ Y_i > t \}.
\end{align*}
By the multivariate central limit theorem, we have that
\[
    \sqrt{n} \left( 
        \widehat{A}_n(t,u) - \Q\bigl(\ZFt \le u, Y > t \bigr), \, 
        \widehat{B}_n(t) - \expec[1-F(t)] 
    \right)'
    \xrightarrow{d}
    N_2(0, \Sigma(t,u)),
\]
as $n \to \infty$. The limit distribution is again bivariate Gaussian with zero mean vector, and with covariance matrix $\Sigma(t,u)$ equal to
\[
     \begin{pmatrix}
        \Q\bigl(\ZFt \le u, Y > t \bigr) \left( 1 - \Q\bigl(\ZFt \le u, Y > t \bigr)\right) &
        \cov \left[ \one\{\ZFt \le u, Y > t \}, 1 - F(t) \right] \\
        \cov \left[ \one\{\ZFt \le u, Y > t \}, 1 - F(t) \right] & \var[1-F(t)]
    \end{pmatrix}.
\]
By the same arguments, 
\[
    \sqrt{n} \left( 
        \widehat{A}_n(t,u) - \Q\bigl(\ZFt \le u, Y > t \bigr), \, 
        \widehat{C}_n(t) - \Q(Y > t) 
    \right)'
    \xrightarrow{d}
    N_2\left(0, \tilde{\Sigma}(t,u)\right),
\]
as $n \to \infty$, where $\tilde{\Sigma}(t,u)$ is
\[
    \begin{pmatrix}
        \Q\bigl(\ZFt \le u, Y > t \bigr) \left( 1 - \Q\bigl(\ZFt \le u, Y > t \bigr)\right) &
        \Q\bigl(\ZFt \le u, Y > t \bigr)(1-\Q(Y > t))\\
        \Q\bigl(\ZFt \le u, Y > t \bigr)(1-\Q(Y > t)) & \Q(Y > t)(1-\Q(Y > t))
    \end{pmatrix}.
\]

As before, consider the delta method applied to the function $g(x, y) = x/y$, with gradient $\nabla g(x, y) = (1/y, -x/y^2)$. This gives
\begin{equation}
    \label{eq:combratCLT}
    \sqrt{n} \left(
        \frac{\widehat{A}_{n}(t, u)}{\widehat{B}_n(t)} - 
        \frac{\Q\bigl(\ZFt \le u, Y > t \bigr)}{\expec[1-F(t)]}
    \right)
    \xrightarrow{d}
    N(0, \sigma^2(t)), \qquad n \to \infty,
\end{equation}
with limiting variance
\[
    \sigma^2(t,u)
    = v(t,u)' \Sigma(t,u) v(t,u)
    \quad \text{with} \quad
    v(t,u)' = \nabla g \left(\Q\bigl(\ZFt \le u, Y > t \bigr), \expec[1-F(t)] \right),
\]
and
\begin{equation}
    \label{eq:sevratCLT}
    \sqrt{n} \left(
        \frac{\widehat{A}_{n}(t, u)}{\widehat{C}_n(t)} - 
        \frac{\Q\bigl(\ZFt \le u, Y > t \bigr)}{\Q(Y > t)}
    \right)
    \xrightarrow{d}
    N(0, \tilde{\sigma}^2(t)), \qquad n \to \infty,
\end{equation}
with limiting variance
\[
    \tilde{\sigma}^2(t,u)
    = \tilde{v}(t,u)' \tilde{\Sigma}(t,u) \tilde{v}(t,u)
    \quad \text{with} \quad
    \tilde{v}(t,u)' = \nabla g \left(\Q\bigl(\ZFt \le u, Y > t \bigr), \Q(Y > t) \right),
\]
which is
\[
    \tilde{\sigma}^2(t,u) = \frac{\Q\bigl(\ZFt \le u, Y > t \bigr)}{\Q(Y > t)^3}\Q\bigl(\ZFt > u, Y > t \bigr).
\]

By replacing the quantities defining $\sigma^2(t,u)$ and $\tilde{\sigma}^2(t,u)$ with obvious sample estimates, we can obtain plug-in estimates $\hat{\sigma}^2_{n,t,u}$ and $\tilde{\sigma}^2_{n,t,u}$, respectively. This again facilitates the construction of classical normal confidence intervals and hypothesis tests, allowing us to test whether the combined ratio or severity ratio is close (pointwise) to the diagonal for a given $t$.

Figure \ref{fig:cs_pit_ci} displays the occurrence, severity, and combined ratios for the four forecasting methods employed in the case study, along with pointwise 95\% confidence intervals derived from the central limit theorem and delta method, as described above.

\begin{figure}
    \centering
    \includegraphics[width=0.3\textwidth]{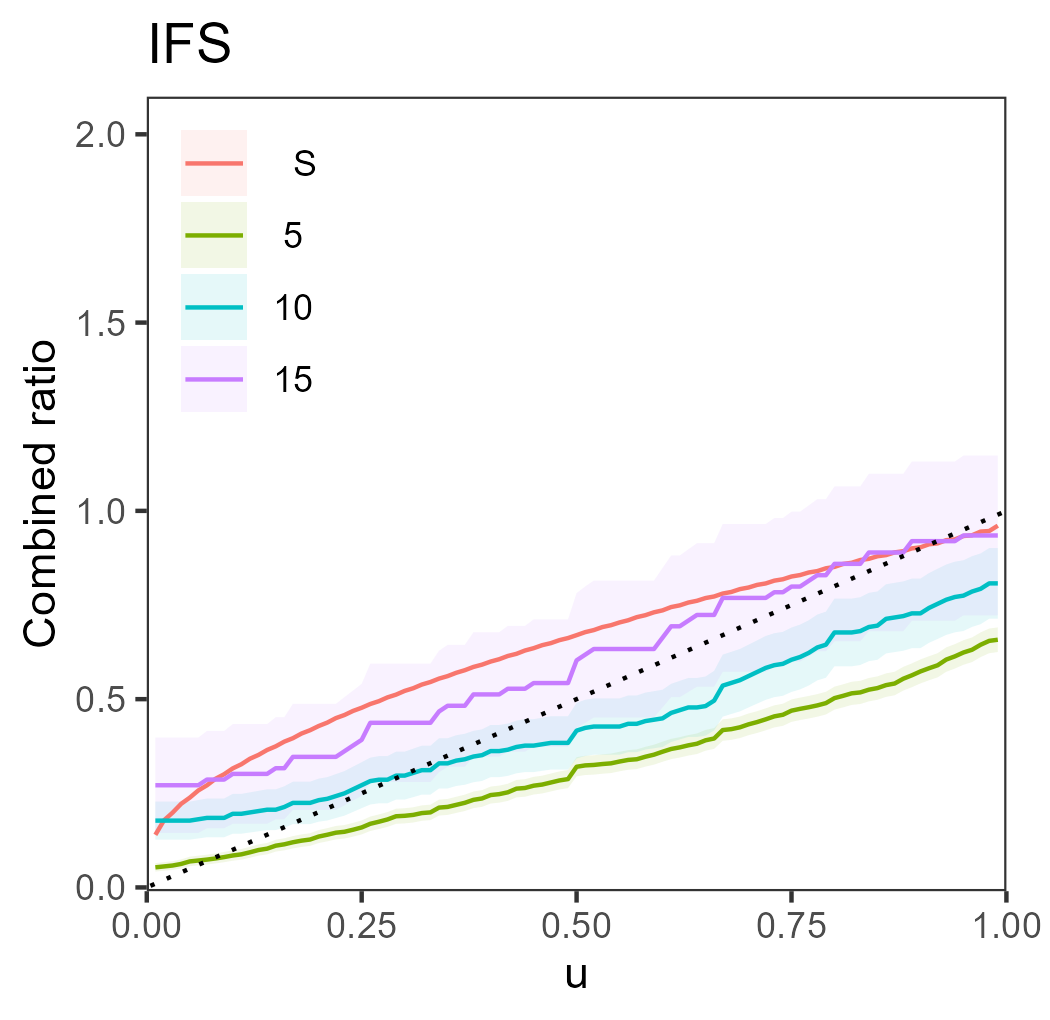}
    \includegraphics[width=0.3\textwidth]{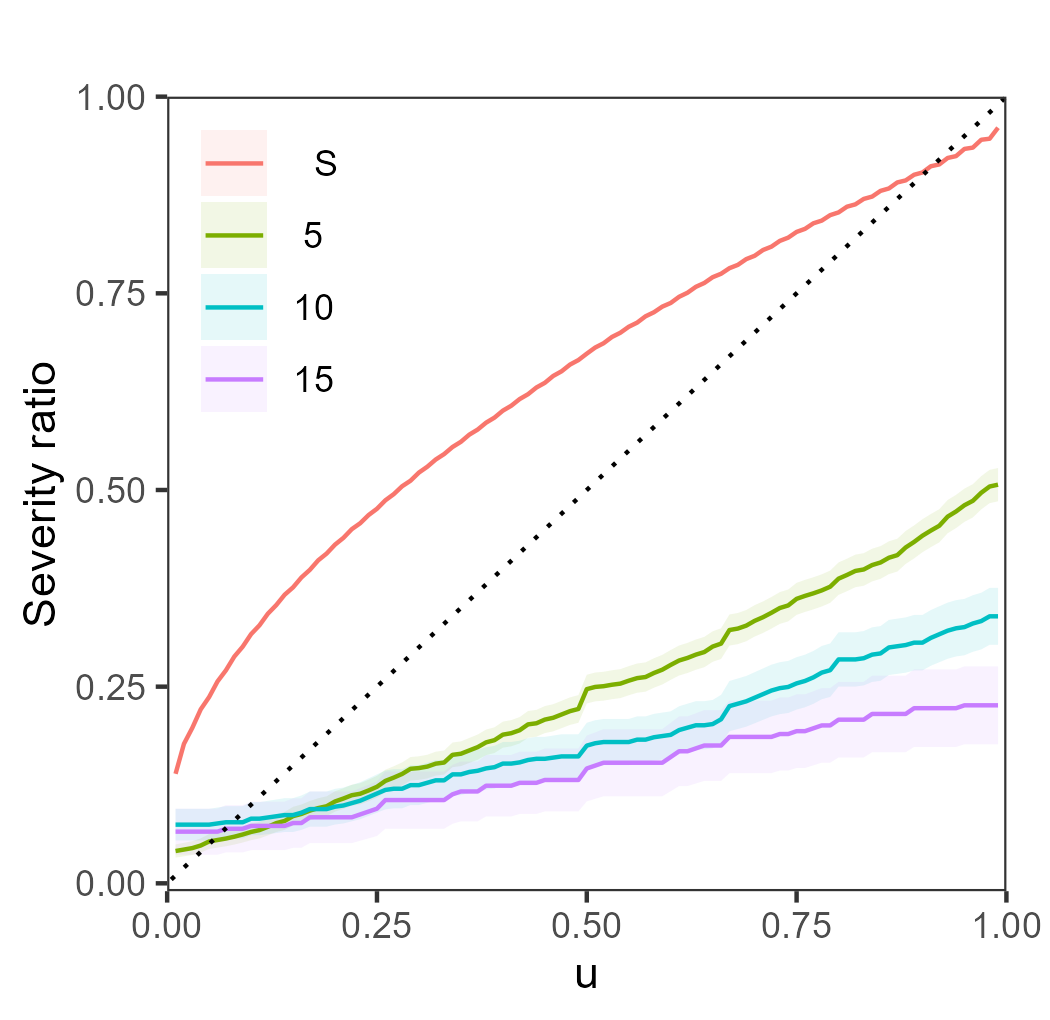}
    \includegraphics[width=0.3\textwidth]{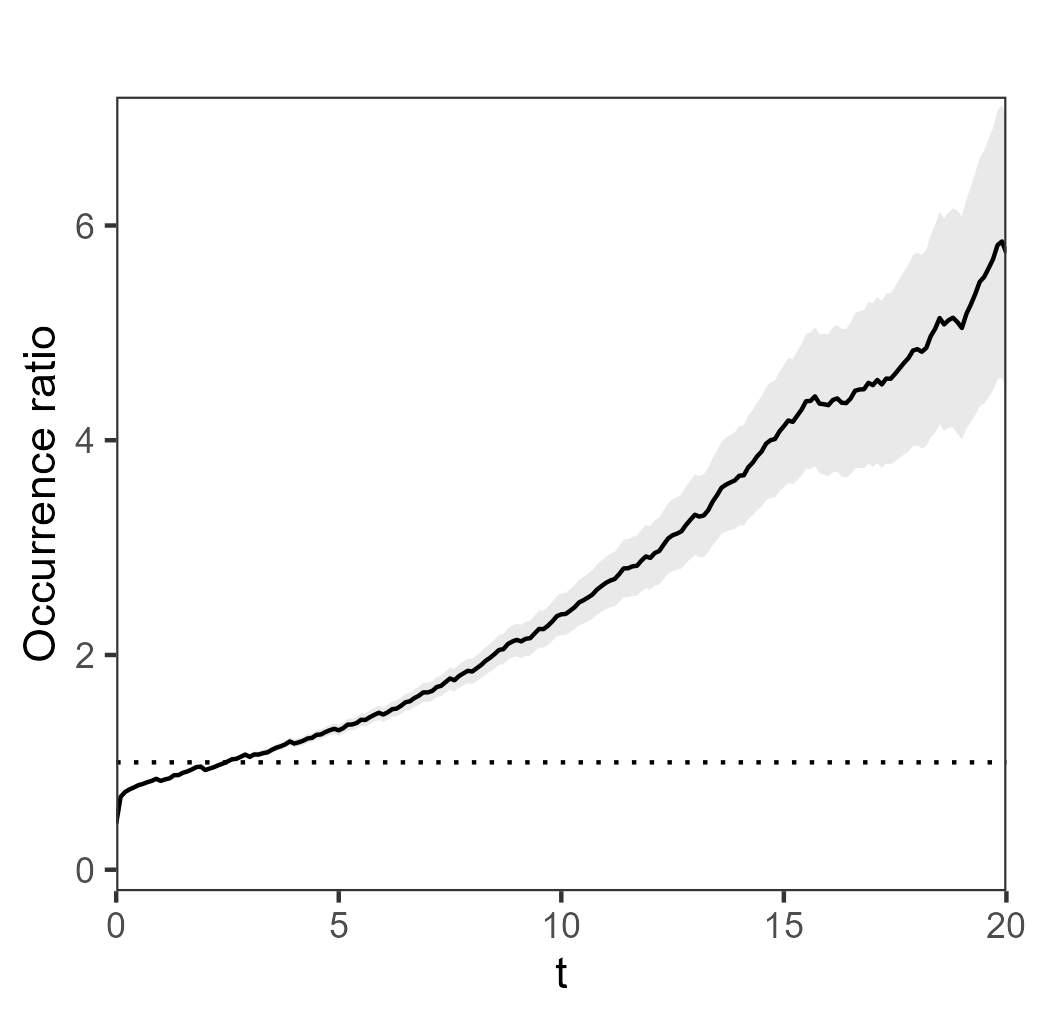}
    \newline
    \includegraphics[width=0.3\textwidth]{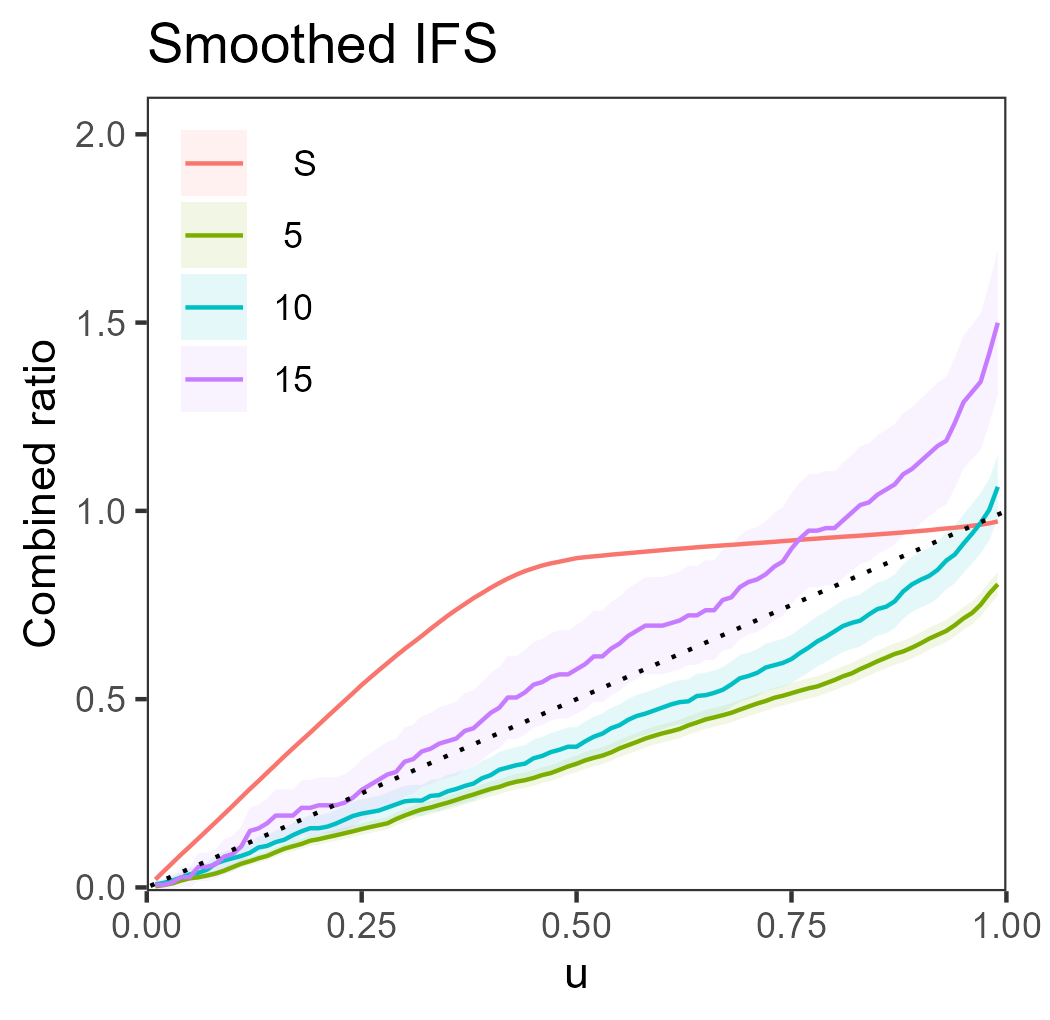}
    \includegraphics[width=0.3\textwidth]{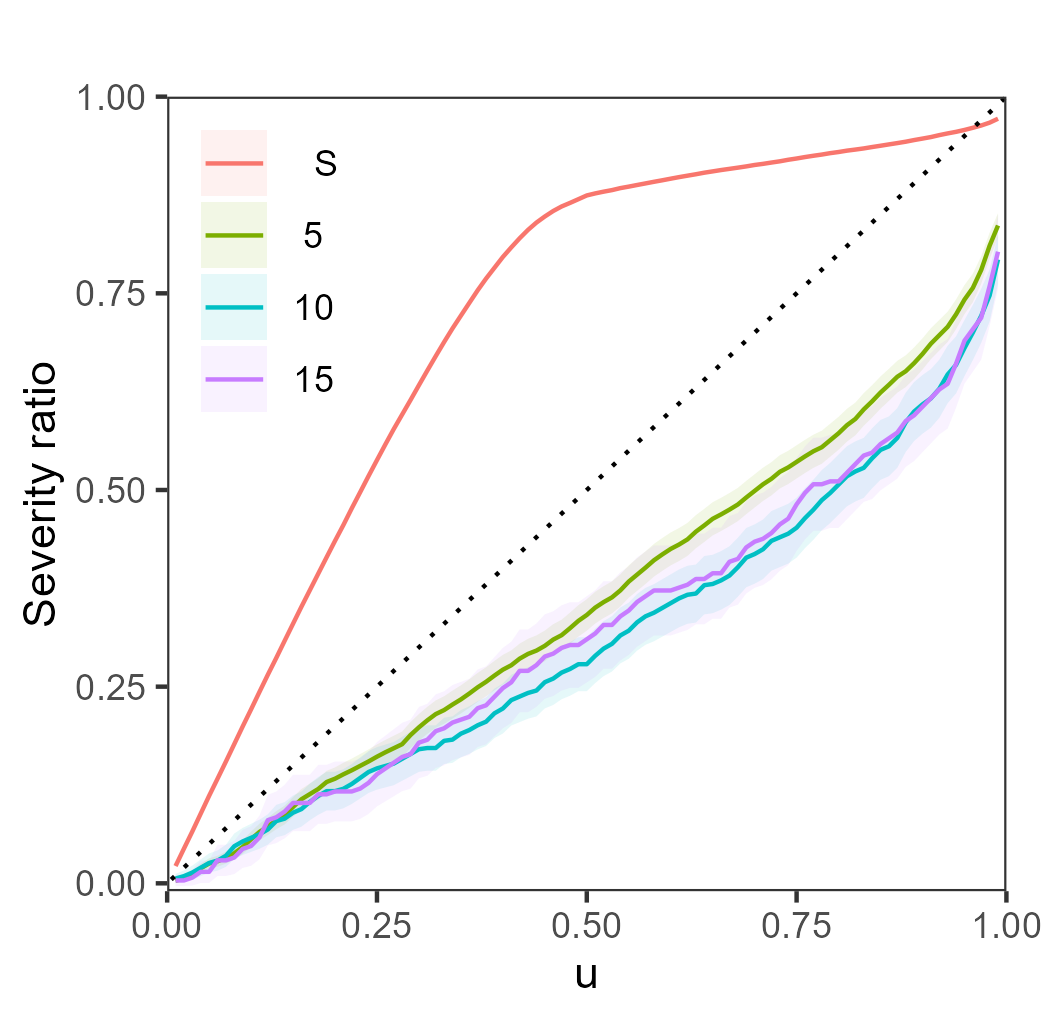}
    \includegraphics[width=0.3\textwidth]{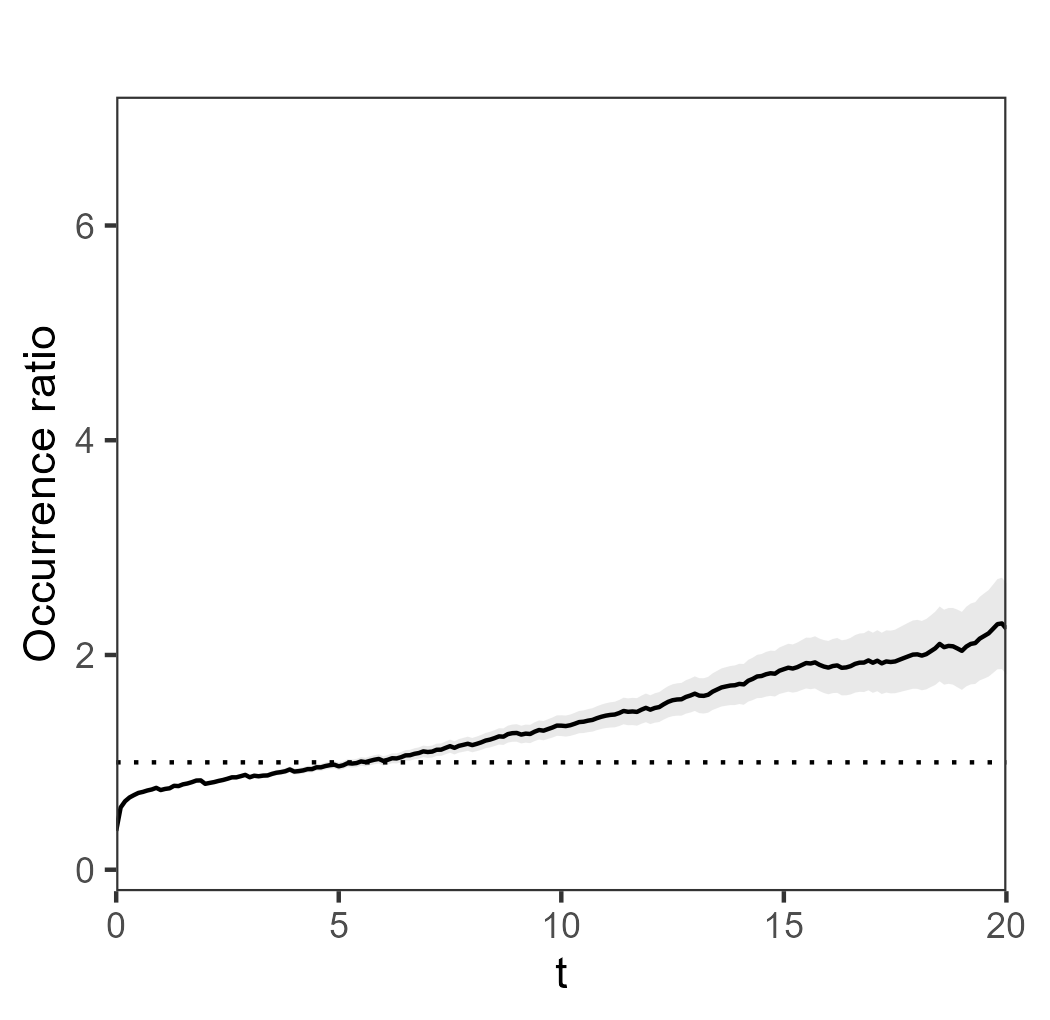}
    \newline
    \includegraphics[width=0.3\textwidth]{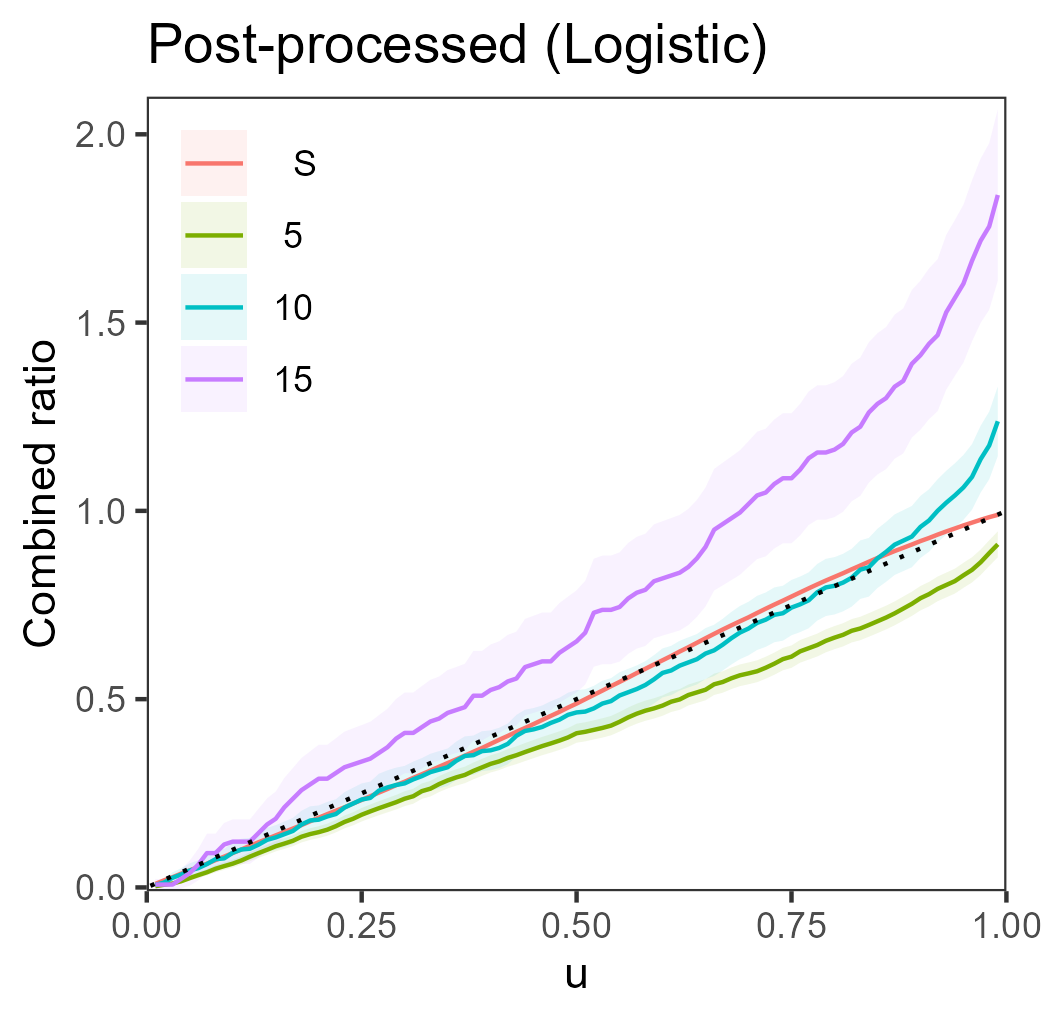}
    \includegraphics[width=0.3\textwidth]{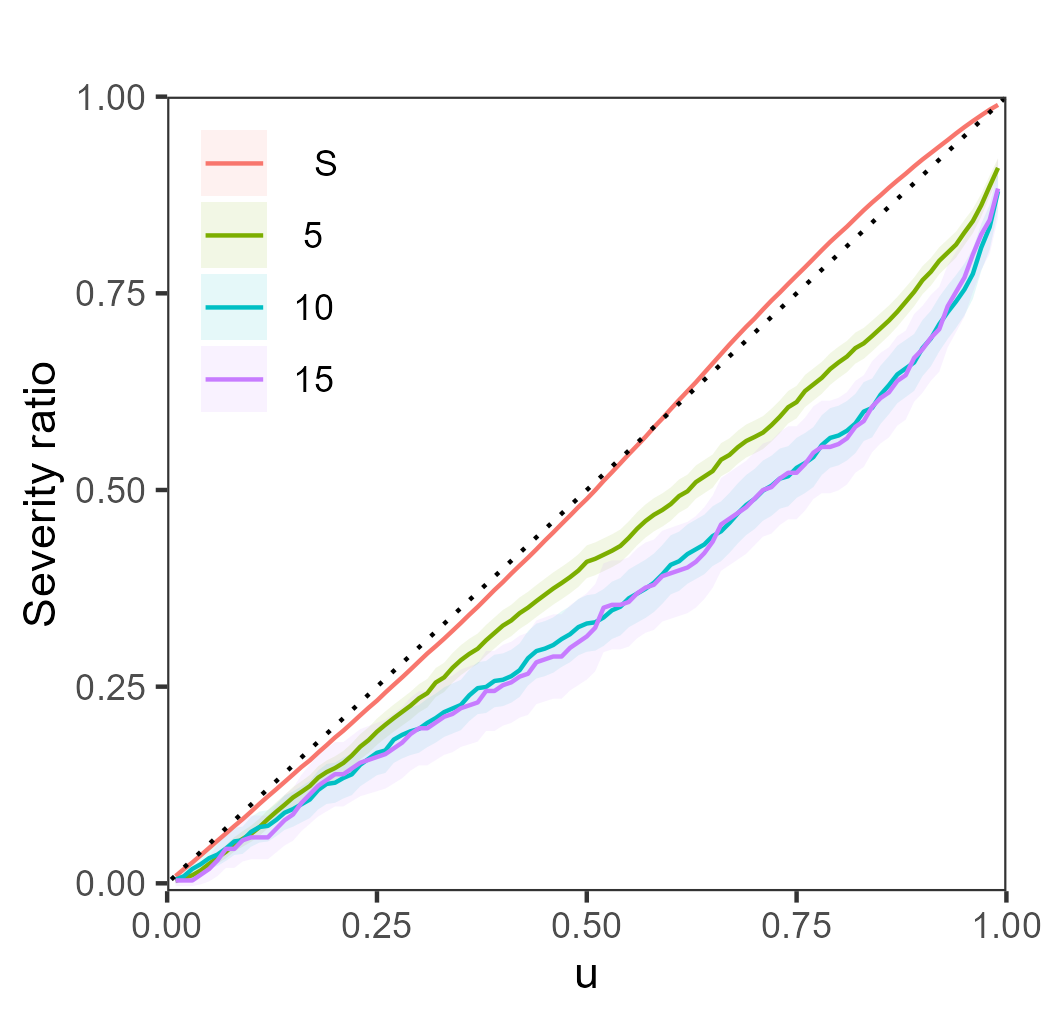}
    \includegraphics[width=0.3\textwidth]{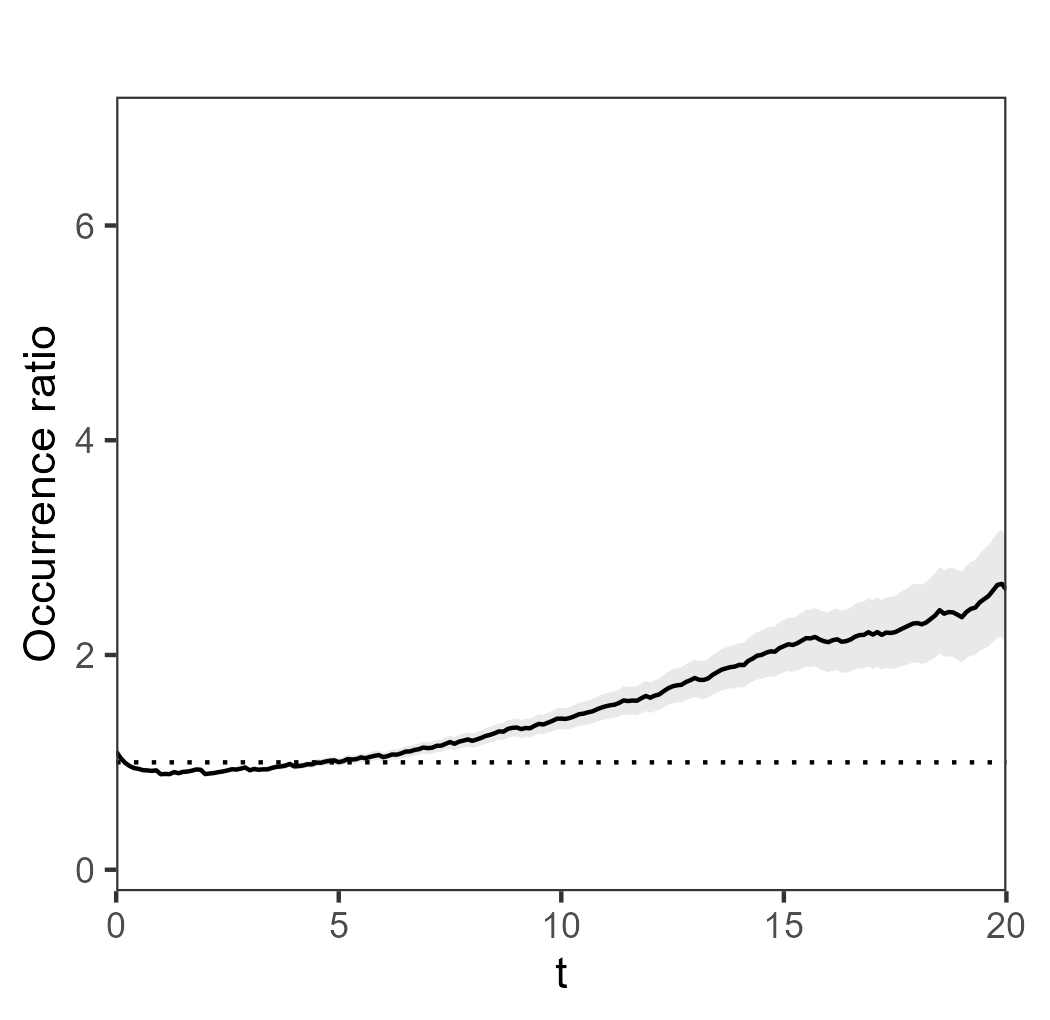}
    \newline
    \includegraphics[width=0.3\textwidth]{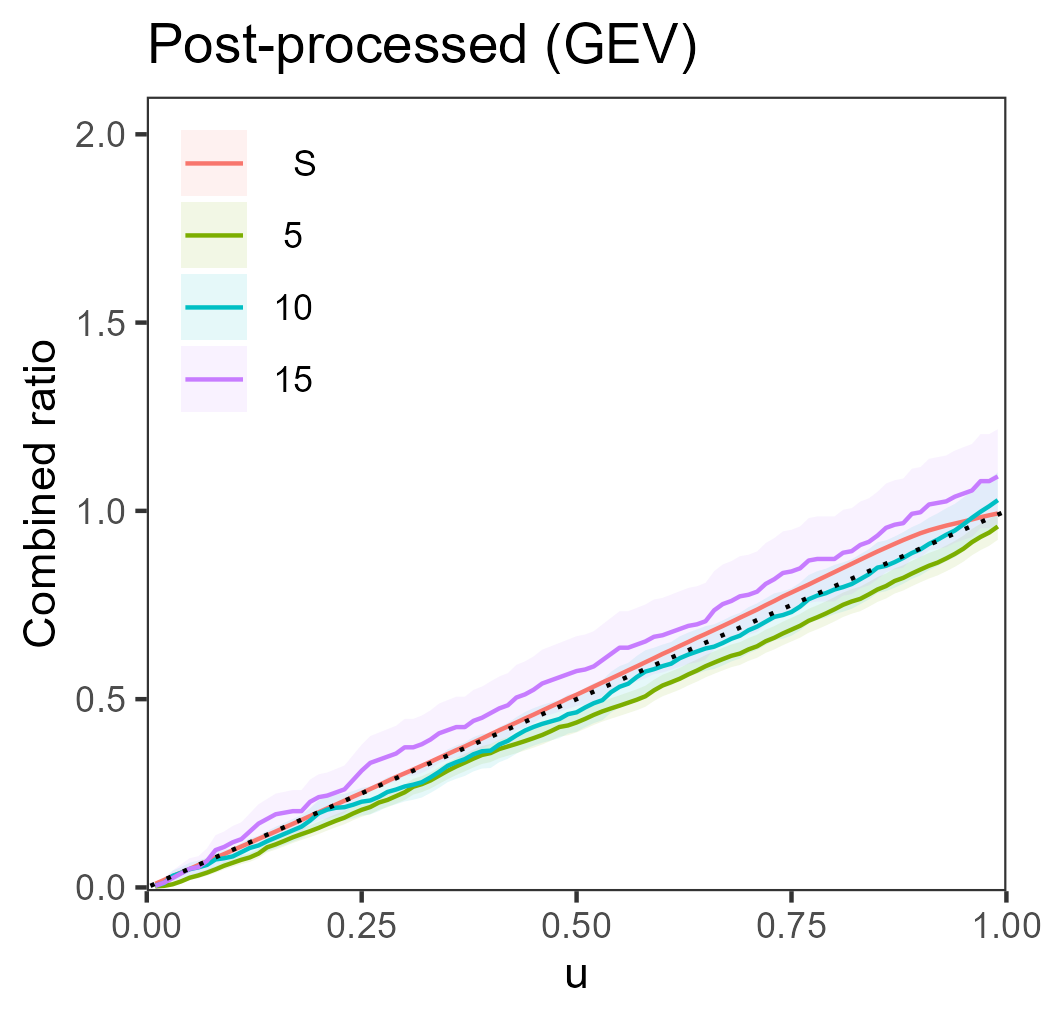}
    \includegraphics[width=0.3\textwidth]{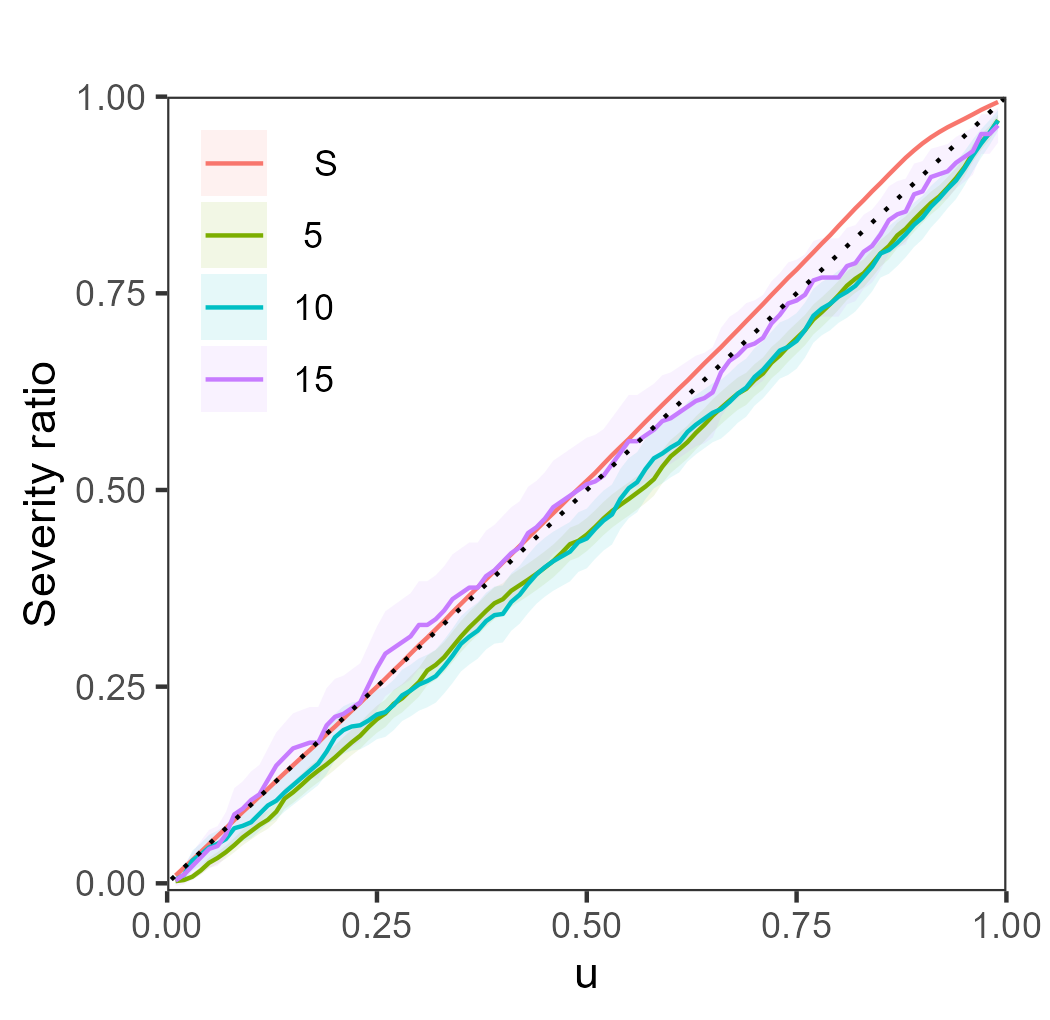}
    \includegraphics[width=0.3\textwidth]{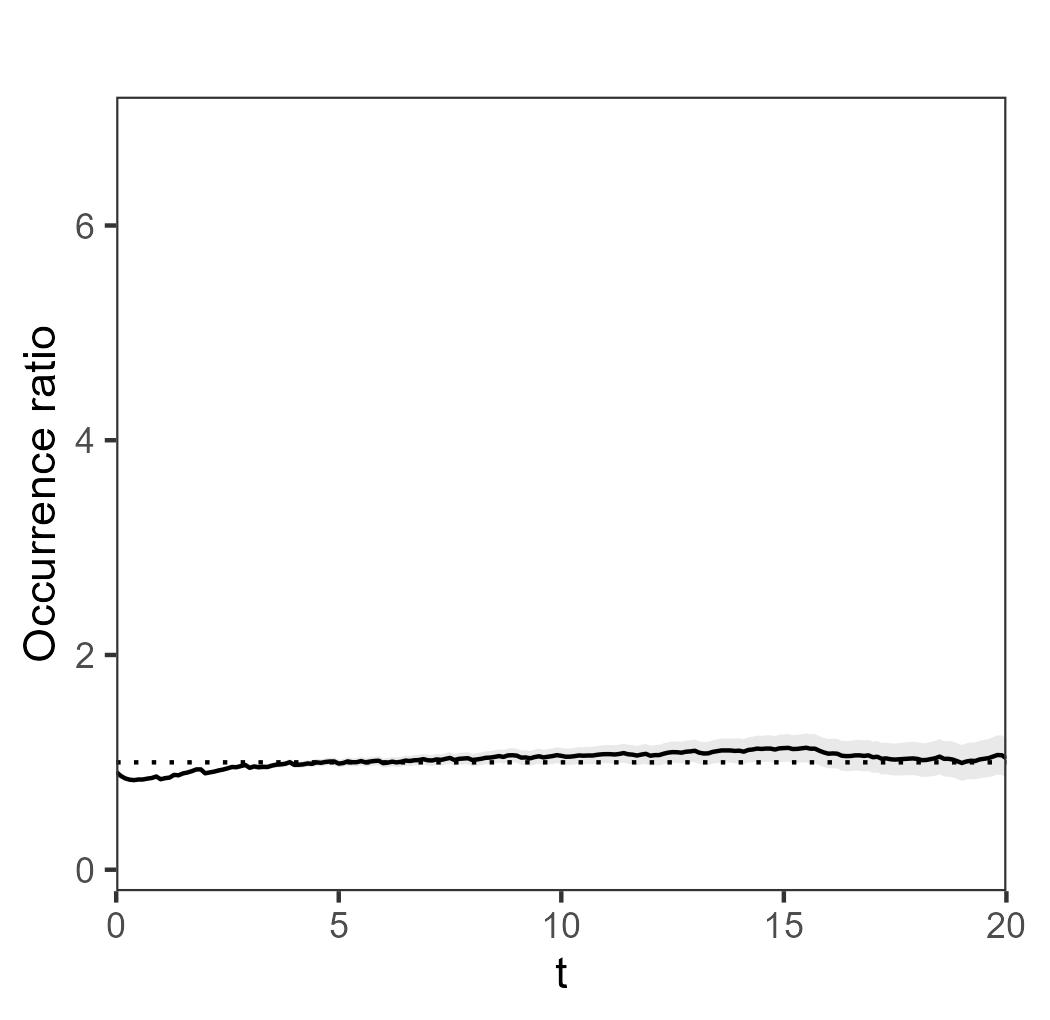}
    \newline
    \caption{As in Figure \ref{fig:cs_pit} of the main manuscript, but with pointwise central 95\% confidence intervals derived from the central limit theorem and the delta method, indicated by the shaded regions.}
    \label{fig:cs_pit_ci}
\end{figure}

\end{document}